\documentclass[cleveref, autoref]{paper}
 
\usepackage{fullpage,float,amsthm,amsfonts,amsmath,amssymb,keytheorems,url,tikz,circuitikz,tabularx,multirow,multicol,hhline,tablefootnote,xcolor,etoolbox,clipboard,authblk,titlesec,setspace,titletoc,tocloft,enumitem,xparse,xspace,pgfplots}
\usepackage[sortcites=true,maxnames=1000]{biblatex}
\usepackage[toc,page]{appendix}
\usepackage[margin=0pt,font={small},labelfont={bf},labelsep=colon,format=plain]{caption}
\usepackage[left]{lineno}
\usepackage[mathscr]{eucal}
\usepackage[ruled,vlined]{algorithm2e}
\usepackage[hidelinks]{hyperref}
\makeatletter 
\let\original@algocf@latexcaption\algocf@latexcaption
\long\def\algocf@latexcaption#1[#2]{%
  \@ifundefined{NR@gettitle}{%
    \def\@currentlabelname{#2}%
  }{%
    \NR@gettitle{#2}%
  }%
  \original@algocf@latexcaption{#1}[{#2}]%
}
\makeatother 

\usetikzlibrary{fadings}
\usetikzlibrary{patterns}
\usetikzlibrary{shadows.blur}
\usetikzlibrary{shapes}
\usetikzlibrary{positioning}

\pgfplotsset{width=10cm,compat=1.9}
\usepgfplotslibrary{external}

\def\titlename{}
\def\authorsnames{}

\def\illustrationsPath{Illustrations/}

\def\emulatorsImagesPath{\illustrationsPath Emulators/}

\def\squiggylinescale{0.08}
\def\squiggylinedist{0.6}
\def\squiggylinerange{3.2}
\def\checkmark{\tikz\fill[scale=0.4](0,.35) -- (.25,0) -- (1,.7) -- (.25,.15) -- cycle;}
\def\erdos{Erd\H{o}s' girth conjecture\xspace}

\definecolor{turquoise}{RGB}{64,224,208}
\definecolor{salmon}{RGB}{250,128,114}

\newcommand{\floor}[1]{\left\lfloor #1 \right\rfloor}
\newcommand{\ceil}[1]{\left\lceil #1 \right\rceil}
\newcommand{\paren}[1]{\left( #1 \right)}
\newcommand{\bracke}[1]{\left[ #1 \right]}
\newcommand{\bracce}[1]{\left\{ #1 \right\}}

\newcommand{\abs}[1]{\left| #1 \right|}

\newcommand{\tabref}[1]{\autoref{tab:#1}}
\newcommand{\secref}[1]{\hyperref[#1]{Section~\ref*{#1}}}
\newcommand{\appref}[1]{\hyperref[#1]{Appendix~\ref*{#1}}}
\newcommand{\thmref}[1]{~\protect\autoref{thm:#1}}
\newcommand{\lemref}[1]{\hyperref[lem:#1]{Lemma~\ref*{lem:#1}}}
\newcommand{\algref}[1]{\hyperref[alg:#1]{Algorithm~\ref*{alg:#1}}}

\newcommand{\cncref}[1]{~\protect\hyperref[cnc:#1]{Conclusion~\ref*{cnc:#1}}}
\newcommand{\corref}[1]{~\protect\hyperref[cor:#1]{Corollary~\ref*{cor:#1}}}
\newcommand{\figref}[1]{\autoref{fig:#1}}

\newcommand{\obvref}[1]{\hyperref[obv:#1]{Observation~\ref*{obv:#1}}}
\newcommand{\clmref}[1]{\hyperref[clm:#1]{Claim~\ref*{clm:#1}}}
\newcommand{\stpref}[3]{\hyperref[alg:#1_#2]{\textit{#3}}}

\newtheorem{question}{Question}
\newtheorem{definition}{Definition}
\newcommand{\prbref}[1]{\hyperref[prb:#1]{Problem~\ref*{prb:#1}}}
\newcommand{\defref}[1]{\hyperref[def:#1]{Definition~\ref*{def:#1}}}
\newcommand{\queref}[1]{\hyperref[que:#1]{Question~\ref*{que:#1}}}
\ExplSyntaxOn
\NewDocumentCommand{\querefs}{m}{
  Questions~
  \clist_set:Nn \l_tmpa_clist {#1}
  \int_set:Nn \l_tmpa_int {0}
  \clist_map_inline:Nn \l_tmpa_clist {
    \int_incr:N \l_tmpa_int
    \hyperref[que:##1]{\ref*{que:##1}}
    \int_compare:nNnT { \l_tmpa_int } < { \clist_count:N \l_tmpa_clist } { ,~ }
  }
}
\ExplSyntaxOff
\newcommand{\eqnref}[1]{\hyperref[eqn:#1]{Equation~(\ref*{eqn:#1})}}

\NewDocumentCommand\mmin{mgg}{\displaystyle \min\IfNoValueTF{#2}{}{_{#2}}\IfNoValueTF{#3}{}{^{#3}}\bracce{#1}}
\NewDocumentCommand\mmax{mgg}{\displaystyle \max\IfNoValueTF{#2}{}{_{#2}}\IfNoValueTF{#3}{}{^{#3}}\bracce{#1}}
\NewDocumentCommand\mprod{mgg}{\displaystyle \prod\IfNoValueTF{#2}{}{_{#2}}\IfNoValueTF{#3}{}{^{#3}}{#1}}
\NewDocumentCommand\msum{mgg}{\displaystyle \sum\IfNoValueTF{#2}{}{_{#2}}\IfNoValueTF{#3}{}{^{#3}}{#1}}
\NewDocumentCommand\mcup{mgg}{\displaystyle \bigcup\IfNoValueTF{#2}{}{_{#2}}\IfNoValueTF{#3}{}{^{#3}}{#1}}
\NewDocumentCommand\mcap{mgg}{\displaystyle \bigcap\IfNoValueTF{#2}{}{_{#2}}\IfNoValueTF{#3}{}{^{#3}}{#1}}
\newcommand{\textoverline}[1]{\={#1}}
\NewDocumentCommand\case{mg}
    {%
    \ensuremath{\textbf{Case (#1\IfNoValueTF{#2}{}{$_#2$})}}%
    }
\NewDocumentCommand\notcase{mg}
    {%
    \ensuremath{\textbf{Case (\textoverline{#1}\IfNoValueTF{#2}{}{$_#2$})}}%
    }
\newcommand{\midline}{\,\middle| \,\,}

\newcommand{\qoute}[1]{``#1''}
\newcommand{\squiggly}{
\begin{tikzpicture} 
    {
    \draw[domain=-\squiggylinedist:0,smooth,variable=\x,scale=\squiggylinescale] plot({\x},{0.75+sin(4*0 r)});
    \draw[domain=0:\squiggylinerange,smooth,variable=\x,scale=\squiggylinescale] plot({\x},{0.75+sin(4*\x r)});
    \draw[domain=\squiggylinerange:\squiggylinerange+\squiggylinedist,smooth,variable=\x,scale=\squiggylinescale] plot({\x},{0.75+sin(4*\squiggylinerange r)});
    }
\end{tikzpicture}
}
\newcommand{\straight}{
\begin{tikzpicture} 
    {
    \draw[domain=-\squiggylinedist:-\squiggylinedist+0.0001,smooth,variable=\x,scale=\squiggylinescale] plot({\x},{-0.25+sin(4*\squiggylinerange r)});
    \draw[domain=-\squiggylinedist+0.0001:\squiggylinerange+\squiggylinedist,smooth,variable=\x,scale=\squiggylinescale] plot({\x},{0.75+sin(4*\squiggylinerange r)});
    }
\end{tikzpicture}
}
\newcommand{\fatrightarrow}
{
    \begin{tikzpicture}
        \draw[-{Triangle[width=6pt,length=6pt]}, line width=2pt](0,0) -- (0.3, 0);
    \end{tikzpicture}
}
\newcommand{\hollowrightarrow}
{
    \begin{tikzpicture}
        \node[draw, single arrow,
              minimum height=1pt, minimum width=1pt,
              single arrow head extend=4pt,
              anchor=west, rotate=0] at (0,-2) {};
    \end{tikzpicture}
}
\newcommand{\ubrace}[2]{{\underbrace{#1}_{#2}}}
\newcommand{\textubrace}[2]{%
  \ubrace{#1}{%
    \mbox{\scriptsize\begin{socgtabular}{@{}c@{}}#2\end{socgtabular}}%
  }%
}
\newcommand{\rightabovearrow}[2]{\overset{#2}{\overrightarrow{#1}}}
\newcommand{\leftabovearrow}[2]{\overset{#2}{\overleftarrow{#1}}}
\newcommand{\sizeindex}[2]{\abs{#1}_{#2}}
\ifdefined\socgtabular\else
  \let\socgtabular\tabular
\fi

\newcommand{\longspace}[0]{\,\,\,\,\,\,\,\,\,\,\,\,\,\,\,\,\,\,\,\,\,\,\,\,\,\,\,\,\,\,\,\,\,\,\,\,\,\,\,\,\,\,\,\,\,\,\,\,\,}

\newcommand*{\settitlename}[1]{\title{#1}\def\titlename{#1}}
\newcommand*{\addauthor}[2]{\author[#2]{#1} 
\IfNoValueTF{\authorsnames}{\def\authorsnames{#1}}{\expandafter\def\expandafter\authorsnames\expandafter{\authorsnames, #1}}}
\renewcommand{\keywords}[1]{\def\keywordnames{#1}}

\newcommand{\dashuline}[1]
{
  \tikz[baseline=(X.base)]
  {
    \node[inner sep=0pt] (X) {#1};
    \draw[dash pattern=on 2pt off 2pt] (X.south west) -- (X.south east);
  }
}
\NewDocumentCommand\specball{mmg}
{
    ball\paren{#1 , #2 \IfNoValueTF{#3}{}{, #3}}
}
\ExplSyntaxOn
\cs_new:Npn \ballformatvar:n #1
 {
  \tl_clear:N \l_tmpa_tl
  \tl_map_inline:nn {#1}
   {
    \str_if_in:nnTF {ijkl\ell} {##1} 
      { \tl_put_right:Nn \l_tmpa_tl {\textit{##1}} }
      { \tl_put_right:Nn \l_tmpa_tl {##1} } 
   }
  \l_tmpa_tl
}
\cs_new:Npn \ballformatarg:n #1
 {
  \regex_match:nnTF { \A (?:[0-9]|i|j|k|\\ell)(?:[0-9ij k\\ell\+\-]*?) \Z } { #1 }
       {S\textsubscript{\ballformatvar:n{#1}}} 
       { #1 }
 }

\NewDocumentCommand \ball {mmg}
 {
    ball\paren{#1 ,
               \ballformatarg:n {#2}
               \IfNoValueTF{#3}
                           {}
                           {, \ballformatarg:n {#3}}
              }
 }
\ExplSyntaxOff

\newtheorem{theorem}{Theorem}[section]
\newtheorem{lemma}[theorem]{Lemma}
\newtheorem{remark}[theorem]{Remark}
\newtheorem{claim2}[theorem]{Claim}
\newtheorem{corollary}[theorem]{Corollary}

\newtheorem{observation}[theorem]{Observation}

\settitlename{Weighted Emulators with Local Heaviest Edges Stretch for Undirected Graphs}

\addauthor{Liam Roditty}{1,$\S$}
\addauthor{Ariel Sapir}{1,$\propto$}

\affil[1]{Department of \href{http://cs.biu.ac.il/}{Computer Science}, Bar-Ilan University}
\affil[$\S$]{liam.roditty@biu.ac.il}
\affil[$\propto$]{sapirar@biu.ac.il}

\date{\today}


\makeatletter
\def\@maketitle{%
  \newpage
  \null
  \vskip 2em%
  \begin{center}%
  \let \footnote \thanks
    {\huge \textsf{\textbf{\@title}} \par}%
    \vskip 1.5em%
    {\large
      \lineskip .5em%
      \begin{tabular}[t]{c}%
        \@author
      \end{tabular}\par}%
    \vskip 1em
    {\large \@date}
  \end{center}
  \par
  \vskip 1.5em}
\makeatother

\makeatletter
\renewenvironment{abstract}{
    \begin{center}%
    {\bfseries\sffamily \Large\abstractname\vspace{\z@}}
      \end{center}%
      \quotation
    }
    \endquotation
\makeatother

\titleformat*{\section}{\Large\bfseries\sffamily}
\titleformat*{\subsection}{\large\bfseries\sffamily}

\dottedcontents{section}[1.5em]{}{1.3em}{.6em}

\setlist[itemize,1]{label=\scalebox{1.6}[1.0]{$\blacklozenge$}}
\setlist[itemize,2]{label=\scalebox{0.9}[1.3]{$\blacktriangleright$}}
\setlist[itemize,3]{label=\textbf{\scalebox{0.9}[1.2]{\fatrightarrow}}}
\setlist[itemize,4]{label=\textbf{\scalebox{0.666}[0.47]{\hollowrightarrow}}}
\setlist[itemize,5]{label=\textbf{\scalebox{0.9}[1.2]{$>$}}}

\begin{document}

\let\originalleft\left
\let\originalright\right
\renewcommand{\left}{\mathopen{}\mathclose\bgroup\originalleft}
\renewcommand{\right}{\aftergroup\egroup\originalright}

\setlength{\textfloatsep}{8pt}


\maketitle

\begin{abstract}
\small 

We introduce a generalized  family of $\left( 2\cdot \left\lfloor \frac{k}{2} \right\rfloor-1, 2\cdot \left\lceil \frac{k}{2} \right\rceil \cdot W_{1} +\max\left\{0,2\cdot\left(\left\lceil\frac{k}{2}\right\rceil-2\right)\right\}\cdot W_{2} \right)$-emulators with $\tilde O \left(n^{1+\frac{1}{k}}\right)$ edges, for any $k\in\mathbb{N}$, where $W_{i}$ is the $i^{\textnormal{th}}$ heaviest edge on a shortest path between two vertices. Our construction generalizes the $+2W_{1}$-spanner of size $\tilde O\left(n^{\frac{3}{2}}\right)$ and the $+4W_{1}$-emulator of size $\tilde O \left(n^{\frac{4}{3}}\right)$, both by Elkin, Gitlitz and Neiman [DISC'21 and DICO'23]. 

When $k$ is even, these are $\left(k-1,k\cdot W_{1} + \left(k-4\right)\cdot W_{2}\right)$-emulators and when $k$ is odd, these are $\left(k-2,\left(k+1\right)\cdot W_{1} + \left(k-3\right) \cdot W_{2}\right)$-emulators. Our framework not only expands known constructions for weighted graphs but also yields an improved stretch over state of the art emulators and spanners for unweighted graphs within a specific distance regime. In particular, for all vertex pairs separated by a distance of $\delta \leq O\left(3^{k^{2}}\right)$, our construction improves upon the seminal additive $+\tilde O\left(\delta^{1-\frac{1}{k}}\right)$-emulator of size $\tilde O\left(n^{1+\frac{1}{2^{k+1}-1}}\right)$ by Thorup and Zwick [SODA'06].

\end{abstract}

\setstretch{1.3}

\newpage

\tableofcontents
\newpage

\section{Introduction}~\label{intro}
When the edge set of a graph is large, directly working with the graph can become computationally expensive. One way to alleviate this is by compressing the graph while preserving critical properties, like distances. This compression helps reduce the computational load and improve the efficiency of algorithms. In particular, one form of compression that has gained attention is the concept of \textit{emulators}, in which the vertex set remains unchanged, but edges can be added or removed, as long as the distance between any two vertices in the modified graph remains within a specified range.

Let $G=\paren{V,E}$ be an unweighted graph and denote by $\delta_{G}\paren{u,v}$ the distance between two vertices $u,v\in V$. An $\paren{\alpha,\beta}$-emulator  $H=\paren{V,F, w_H}$ is a weighted graph with the same vertex set $V$, an auxiliary edge set $F\subseteq V\times V$ and a weight function $ w_{H} :F\rightarrow \mathbb{R^{+}}$, such that, for any $u,v\in V$:

$$\delta_{H}\paren{u,v}\in \bracke{\delta_{G}\paren{u,v}, \alpha\cdot \delta_{G}\paren{u,v}+\beta }$$

In other words: $\delta_{H}\paren{u,v}$ lies within a range bounded by the original distance $\delta_{G} \paren{u,v}$, up to a \textit{multiplicative stretch} of $\alpha$ and an \textit{additive stretch} of $\beta$. Two special cases of emulators are a \textit{multiplicative  emulator}, denoted as an $\alpha$-emulator, where $\beta=0$, and an \textit{additive emulator}, denoted as a $+\beta$-emulator, where $\alpha=1$. The general case is referred to as a \textit{mixed emulator}, having both multiplicative and  additive stretches. 

A closely related object is the \textit{spanner}, a restricted variant of an emulator that is required to be a   subgraph of the original graph:  $F \subseteq E$.    A central question in this regime is what is the proper trade-off between the stretch $\paren{\alpha,\beta}$ of an emulator or a spanner and its size $\abs{F}$. 

This size-stretch trade-off is fundamentally tied to the \textit{girth} of a graph. According to \textit{\erdos} \cite{Erdos1963} for any $k\geq 1$ there are graphs with $\tilde \Omega \paren{n^{1+\frac{1}{k}}}$ edges and girth $2k+1$. In such graphs, removing any edge would  increase the distance between its endpoints from $1$ to at least $2k$. Any $\paren{\alpha,\beta}$-emulator or spanner for such graphs for which $\alpha + \beta \leq 2k-1$ must retain all original edges. This identifies $\tilde O \paren{n^{1+\frac{1}{k}}}$ as a meaningful sparsity threshold.

\begin{question}~\label{que:n^{1+1/k}}
Let $k\in\mathbb{N}$. Which $\paren{\alpha,\beta}$-emulators or spanners with size $\tilde O \paren{n^{1+\frac{1}{k}}}$ exist?
\end{question}

In the multiplicative regime, Alth{\"o}fer, Das, Dobkin and Soares \cite{AltDasDobJosSoa1993} presented a $\paren{2k-1}$-spanner with $O\paren{ n^{1+\frac{1}{k}}}$ edges.  Thorup and Zwick \cite{ThoZwi2005} and Baswana and Sen \cite{BasSen2007} also presented such constructions. In the additive regime, Dor, Halperin and Zwick \cite{DorHalZwi2000} presented a $+4$-emulator of size $\tilde O\paren{n^{\frac{4}{3}}}$. We observe that these spanners and emulators are optimal assuming \erdos. 

Less sparse constructions include the $+4$-spanner of Chechik \cite{Chechik2013} with $\tilde O \paren{n^{\frac{7}{5}}}$ edges and the $+6$-spanners of size $\tilde O\paren{n^{\frac{4}{3}}}$ by Baswana, Kavitha, Mehlhorn and Pettie \cite{BasKavMehPet2010} and Woodruff \cite{Woodruff2010}.  Stretches that grow as a function of $n$ (e.g. \cite{ElkGitNei2021, KogPar2023,HuaPet2019A}) are beyond the scope of this work and are only mentioned in the tables. A  summary of state of the art results is provided in \tabref{known_emulators} for emulators and in  \tabref{known_spanners} for spanners.

\begin{table}[H]
  \begin{center}
    
\small\addtolength{\tabcolsep}{-3pt}
    \begin{tabular}{|c|c|c|c|c|} 
      \hline
      \textbf{\small Stretch} & \textbf{\small Size}  & \textbf{\small Construction Time} & \textbf{\small Weighted}& \textbf{\small Reference}\\
      \hhline{|=|=|=|=|=|}
      &&&&\\ [-1em]
      $+4W_{1}$& $n^{\frac{4}{3}}$ & $n^{\frac{7}{3} } $ & \checkmark & \cite{ElkGitNei2021}  \\
      &&&&\\ [-1em]
      \hline
      &&&&\\ [-1em]

    $+ \paren{6^k -1}\cdot \delta ^{1-\frac{1}{k}}$& $ k\cdot n^{1+\frac{1}{2^{k+1}-1}}$ & -- &  & \cite{ThoZwi2006}  \\
      &&&&\\ [-1em]
      \hline
      &&&&\\ [-1em]

      $+n^{\frac{1-3\cdot\varepsilon}{5}}$& $n^{1+\varepsilon}$ & -- &  & \cite{KogPar2023}  \\
      &&&&\\ [-1em]
      \hline
      &&&&\\ [-1em]

      $+n^{\frac{25-87\cdot\varepsilon}{112}}$& $n^{1+\varepsilon}$ & -- &  & \cite{KogPar2023}  \\
      &&&&\\ [-1em]
      \hline
      &&&&\\ [-1em]

      $+n^{\frac{2}{9}-\frac{1}{1600}-o\paren{1}}$& $n$ & -- &  & \cite{KogPar2023}  \\
      &&&&\\ [-1em]
      \hline
      &&&&\\ [-1em]
      
      $\paren{1+\varepsilon, \paren{90+\frac{120\cdot\paren{k-1}}{\varepsilon}}^{k-1}\cdot W_{1}}$&$kn+n^{1+\frac{1}{2^{k}-1}}$  & --&  \checkmark& \cite{ElkGitNei2019} \\
      &&&&\\ [-1em]
      \hline
      &&&&\\ [-1em]

      $\paren{3+\varepsilon, \paren{3+\frac{8}{\varepsilon}}^{k-1} \cdot W_{1}}$& $kn+n^{1+\frac{1}{2^{k}-1}}$  &  -- & \checkmark & \cite{ElkGitNei2019}  \\
      \hline
    \end{tabular}
    \caption{Known emulators for weighted and unweighted graphs.}~\label{tab:known_emulators}
  \end{center}
\end{table}

\begin{table}[H]
  \begin{center}
    
\small\addtolength{\tabcolsep}{-3pt}
    \begin{tabular}{|c|c|c|c|c|} 
      \hline
      \textbf{\small Stretch} & \textbf{\small Size}  & \textbf{\small Construction Time} & \textbf{\small Weighted}& \textbf{\small Reference}\\
      \hhline{|=|=|=|=|=|}
      &&&&\\ [-1em]
    $+2$& $n^{\frac{3}{2}}$ & $mn^{\frac{1}{2}} $ &   & \cite{DorHalZwi2000}  \\
      &&&&\\ [-1em]
      \hline
      &&&&\\ [-1em]
      
      $+2W_{1}$& $n^{\frac{3}{2}}$ & $n^{2} $ & \checkmark  & \cite{ElkGitNei2021}  \\
      &&&&\\ [-1em]
      \hline
      &&&&\\ [-1em]

$+4$& $n^{\frac{7}{5}}$ & $mn^{\frac{3}{5}}$ &   & \cite{Chechik2013,Dhalaan2021}  \\
      &&&&\\ [-1em]
      \hline
      &&&&\\ [-1em]

$+4W_{1}$& $n^{\frac{7}{5}}$ & -- &  \checkmark & \cite{AhmBodHamKobSpe2021}  \\
      &&&&\\ [-1em]
      \hline
      &&&&\\ [-1em]

    $+6$& $n^{\frac{4}{3}}$ & $n^2$ &  & \cite{BasKavMehPet2010,Woodruff2010}  \\
    &&&&\\ [-1em]
    \hline
    &&&&\\ [-1em]

    $+6W_{1}$& $n^{\frac{4}{3}}$ & -- & \checkmark  & \cite{LaLe2024}  \\
      &&&&\\ [-1em]
      \hline
      &&&&\\ [-1em]

      $+n^{\frac{1-\varepsilon}{2}} \cdot \log n\cdot W_{1}$& $n^{1+\varepsilon}$ &  & \checkmark & \cite{ElkGitNei2021}  \\
      &&&&\\ [-1em]
      \hline
      &&&&\\ [-1em]
      
      $\paren{1+\varepsilon, \paren{5+\frac{8\cdot\paren{k-1}}{\varepsilon}}^{k-1}\cdot W_{1}}$&$kn+n^{1+\frac{1}{\paren{\frac{4}{3}}^{k}-1}}$  & --&  \checkmark& \cite{ElkGitNei2019} \\
      &&&&\\ [-1em]
      \hline
      &&&&\\ [-1em]

      $\paren{3+\varepsilon, \paren{3+\frac{8}{\varepsilon}}\cdot k^{3+\frac{8}{\varepsilon}} }$& $n^{1+\frac{1}{k}}+\paren{3+\frac{8}{\varepsilon}}\cdot k ^{3+\frac{8}{\varepsilon}} \cdot n $  &  -- &  & \cite{BenPar2020}  \\
      &&&&\\ [-1em]
      \hline
      &&&&\\ [-1em]

      $\paren{3+\varepsilon, \paren{5+\frac{16}{\varepsilon}}^{k-1} \cdot W_{1}}$& $kn+n^{1+\frac{1}{\paren{\frac{4}{3}}^{k}-1}}$  &  -- & \checkmark & \cite{ElkGitNei2019}  \\
      &&&&\\ [-1em]
      \hline
      &&&&\\ [-1em]

      $\paren{k-1,2k}$& $kn^{1+\frac{1}{k}}$  &  $mn^{\frac{1}{2}}$ &  & \cite{ElkPel2001}  \\

      &&&&\\ [-1em]
      \hline
      &&&&\\ [-1em]
      $\paren{k,k-1}$& $n^{1+\frac{1}{k}}$ & $km $ &  & \cite{BasKavMehPet2010}  \\

      &&&&\\ [-1em]
      \hline
      &&&&\\ [-1em]
      $\paren{k,\paren{k-1}\cdot W_{\max}}$& $n^{1+\frac{1}{k}}$ & -- & \checkmark & \cite{Tzalik2026}  \\

      &&&&\\ [-1em]
      \hline
      &&&&\\ [-1em]
      $2k-1$& $kn^{1+\frac{1}{k}}$ & $km $ & \checkmark & \cite{BasSen2007}  \\
      \hline
    \end{tabular}
    \caption{Known spanners for weighted and unweighted graphs.}~\label{tab:known_spanners}
  \end{center}
\end{table} 

While \erdos provides a baseline for tightness, it is insufficient for two reasons.  First, it remains unproven for general $k\in\mathbb{N}$, making any lower bound derived from it inherently conditional. Second, it restricts us to the case where $\alpha + \beta \leq 2k-1$. For instance, the $+6$-spanners \cite{BasKavMehPet2010,Woodruff2010} of size $\tilde O\paren{n^{\frac{4}{3}}}$ are not precluded by the conjecture. Resolving these gaps requires additional lower bounds.

\begin{table}[H]
  \begin{center}
    
\small\addtolength{\tabcolsep}{-3pt}
    \begin{tabular}{|c|c|c|c|} 
      \hline
      \textbf{\small Structure Type} & \textbf{\small Stretch} & \textbf{\small Size}  & \textbf{\small Reference}\\
      \hhline{|=|=|=|=|}
      &&&\\ [-1em]
      Spanner & $+2k-1$& $\frac{1}{k}\cdot n^{1+\frac{1}{k}}$ & \cite{Woodruff2006}  \\
      &&&\\ [-1em]
      \hline

    &&&\\ [-1em]
    Emulator & $+2k-1$& $\frac{1}{k^2}\cdot n^{1+\frac{1}{k}}$ &  \cite{Woodruff2006}  \\
    &&&\\ [-1em]
    \hline

    &&&\\ [-1em]
    Both & $+k$ (for $k\geq 4$)& $n^{\frac{4}{3}}$ &  \cite{AbbBod2015}  \\
    &&&\\ [-1em]
    \hline

    &&&\\ [-1em]
    Spanner & $+\Theta\paren{\frac{1}{k+1}}\cdot \delta^{1-\frac{1}{k}}$ & $n^{1+\frac{1}{2^{k+1}-1}}$ &  \cite{AbbBodPet2016}  \\
    &&&\\ [-1em]
    \hline

    &&&\\ [-1em]
    Emulator & $+n^{\frac{1}{18}}$ & $n$ & \cite{HuaPet2019A}  \\
    \hline
    \end{tabular}
    \caption{Known lower bounds for spanners and emulators. While some present a tighter bound than others, we present all, for completeness. All lower bounds are for unweighted graphs, yet carry an implication for weighted graphs, when one assumes $\forall e\in E: w\paren{e} = 1$.}~\label{tab:lower_bounds}
  \end{center}
\end{table}

Woodruff \cite{Woodruff2006} established the first unconditional lower bounds for additive approximations, proving that additive spanners and emulators with $+\paren{2k-1}$ stretch require $\Omega\paren{\frac{1}{k} n^{1+\frac{1}{k}}}$ and $\Omega\paren{\frac{1}{k^{2}} n^{1+\frac{1}{k}}}$ edges, respectively. These were subsequently improved by the \qoute{additive barrier} of Abboud and Bodwin \cite{AbbBod2015}, who proved that any constant additive stretch requires at least $\tilde \Omega\paren{n^{\frac{4}{3}}}$ edges. This renders the $+4$-emulator of size $\tilde O \paren{n^{\frac{4}{3}}}$ of Dor, Halperin and Zwick  \cite{DorHalZwi2000} \textbf{unconditionally} optimal.  Henceforth, for $k \geq 4$, achieving a sparsity of $\tilde O \paren{n^{1+\frac{1}{k}}}$ requires a multiplicative stretch $\alpha>1$, resulting in a mixed approximation.

Notable mixed constructions of size  $\tilde O \paren{n^{1+\frac{1}{k}}}$ include:
 a  $\paren{k-1,2k}$-spanner by Elkin and Peleg \cite{ElkPel2001}, a $\paren{k,k-1}$-spanner by Baswana, Kavitha, Mehlhorn and Pettie \cite{BasKavMehPet2010} and  a $\paren{3+\varepsilon,\paren{3+\frac{8}{\varepsilon}}\cdot k ^{3+\frac{8}{\varepsilon}}}$-spanner by Ben-Levy and Parter \cite{BenPar2020}. 

Yet, these constructions consider \textbf{unweighted} graphs.  Extending these results to the weighted setting introduces a fundamental conceptual challenge.  Let $G=\paren{V,E,w}$ be a weighted graph with a non-negative weight function $w:E\rightarrow \mathbb{R^{+}}$. An additive stretch with a constant $\beta$
lacks scale invariance: multiplying all edge weights makes the additive stretch negligible. To address this, Roditty and Sapir \cite{RodSap2025} introduced the framework of \textit{Commensurate Versions}. This ensures the additive stretch is relative to the weight function $w$, making the approximation meaningful regardless of scaling. We adapt their definition for emulators:

\begin{definition}~\label{def:commensurate}
Let $G=\paren{V,E,w}$ be a weighted graph. Consider a pair of vertices $u,v\in V$, a path $P$ between $u$ and $v$ and a function $f\paren{\beta,G,P}$. The problem of constructing an $\paren{\alpha,f\paren{\beta,G,P}}$-emulator for the weighted setting is a \textbf{Commensurate Version} of the problem of constructing an $\paren{\alpha,\beta}$-emulator for the unweighted setting if $f\paren{\beta,G,P}=\beta$  when $ w\paren{e}=1$ for all $e\in E$.
\end{definition}

Roditty and Sapir \cite{RodSap2025} also introduced the notion of \textit{Strongly Commensurate Versions}, originally defined for algorithmic problems such as \textbf{A}ll-\textbf{P}airs \textbf{A}pproximate \textbf{S}hortest \textbf{P}aths (APASP), where the focus is  on the algorithms' runtimes. We adapt this notion for emulators, focusing naturally on their size:

\begin{definition}~\label{def:strongly_commensurate}
Let $H_1=\paren{V,F_1,\hat w_1}$ be an $\paren{\alpha,\beta}$-emulator for the unweighted setting and let $H_2=\paren{V,F_2,\hat w_2}$ be an $\paren{\alpha,f\paren{\beta,G,P}}$-emulator for the weighted setting. $H_2$ is a \textbf{Strongly Commensurate Version} of $H_1$ if the problem of finding an  $\paren{\alpha,f\paren{\beta,G,P}}$-emulator is a commensurate version of the problem of finding an $\paren{\alpha,\beta}$-emulator and $\abs{F_2} \in \tilde O \paren{\abs{F_1}}$. 
\end{definition}

Both definitions apply naturally to spanners as well. Early examples of strongly commensurate versions utilized $W_{\max}$: the weight of the heaviest edge in the \textbf{entire} graph. However, if a heavy edge exists but is rarely utilized in shortest paths, an approximation depending on $W_{\max}$ becomes highly \qoute{loose}. To overcome this, Cohen and Zwick \cite{CohZwi1997} introduced the notion $W_{i}\paren{u\squiggly v}$, which is the weight of the $i^{\textnormal{th}}$ heaviest edge on a shortest path $u\squiggly v$ for vertices $u,v\in V$. Approximations dependent on $W_{i}$ provide significantly tighter and \qoute{localized} estimations than those relying on $W_{\max}$.

\begin{question}~\label{que:weighted_vs_unweighted}
For which values of $\alpha,\beta$ and $f$, are there  $\paren{\alpha,f\paren{\beta,G,P}}$-emulators for weighted graphs, which are a strongly commensurate version of an $\paren{\alpha,\beta}$-emulator for unweighted graphs?
\end{question}

Partial answers to \queref{weighted_vs_unweighted} have gained popularity in the last several years both for spanners and emulators. The $+4$-spanner of size $\tilde O\paren{n^{\frac{7}{5}}}$ of Chechik \cite{Chechik2013} was adapted into a $+4W_{\max}$-spanner by Ahmed, Bodwin, Sahneh, Kobourov and Spence \cite{AhmBodSahKobSpe2020} and subsequently refined to a tighter $+4W_{1}$-spanner by Ahmed, Bodwin, Hamm, Kobourov and Spence \cite{AhmBodHamKobSpe2021}. The $+2$-spanner of Dor, Halperin and Zwick of size $\tilde O\paren{n^{\frac{3}{2}}}$ \cite{DorHalZwi2000} was adapted to a $+2W_{1}$-spanner of Elkin, Gitlitz and Neiman \cite{ElkGitNei2021}; the  $+6$-spanners of size  $\tilde  O \paren{n^{\frac{4}{3}}}$ by Baswana, Kavitha, Mehlhorn and Pettie \cite{BasKavMehPet2010} or Woodruff \cite{Woodruff2010} were adapted to a $+6W_{1}$-spanner by  La and Le \cite{LaLe2024}; the $\paren{k,k-1}$-spanner of size $\tilde O \paren{n^{1+\frac{1}{k}}}$ of Baswana, Kavitha, Mehlhorn and Pettie \cite{BasKavMehPet2010} was adapted  to a $\paren{k,\paren{k-1}\cdot W_{\max}}$-spanner by Tzalik \cite{Tzalik2026}; the $+4$-emulator of Dor, Halperin and Zwick of size  $\tilde O \paren{n^{\frac{4}{3}}}$\cite{DorHalZwi2000} was adapted to a $+4W_{1}$-emulator by Elkin, Gitlitz and Neiman \cite{ElkGitNei2021}.

\tabref{known_emulators} and \tabref{known_spanners} summarize these results, revealing a notable gap in the weighted regime: known constructions rely exclusively on either $W_{\max}$ or $W_{1}$. To the best of our knowledge, no existing construction utilizes $W_i$ for $i \geq 2$. The significance of this becomes apparent when considering a shortest path $u\squiggly v$ dominated by a single heavy edge: $\delta_{G}\paren{u,v}\approx W_{1}\paren{u\squiggly v} $. An additive stretch proportional to $W_{1}$ offers no meaningful sub-multiplicative guarantee. Conversely, in this exact scenario  the remaining edges must be proportionately small. An approximation utilizing $W_{2}$ provides in such cases a strictly tighter guarantee than either $W_{1}$ or $W_{\max}$. 

\begin{question}~\label{que:w2}
Which $\paren{\alpha, f\paren{G, W_{1}, W_{2}, \ldots,W_{k}}}$-emulators or spanners exist, where $f$  explicitly depends on at least one of the local heaviest edge weights $W_{i}$ for $i \geq 2$?
\end{question}

Our framework addresses this fundamental gap, introducing the first explicit dependency on the second heaviest edge, $W_{2}$. We obtain a unified framework of $\paren{2\cdot \floor{\frac{k}{2}}-1, 2\cdot \ceil{\frac{k}{2}} \cdot W_{1} +\max\bracce{0,2\cdot\paren{\ceil{\frac{k}{2}}-2}}\cdot W_{2} }$-emulators with $\tilde O\paren{n^{1+\frac{1}{k}}}$ edges, for any integer $k \geq 2$. More precisely, the resulting stretch depends on the parity of $k$; for an even value of $k$, these are $\paren{k-1,k\cdot W_{1} + \paren{k-4}\cdot W_{2}}$-emulators; for an odd value of $k$, these are $\paren{k-2,\paren{k+1}\cdot W_{1} + \paren{k-3} \cdot W_{2}}$-emulators. We emphasize that the dependence we have over $W_{1}$ is linear, in contrast to other mixed approximations, such as the $\paren{1+\varepsilon, \paren{90+\frac{120\cdot\paren{k-1}}{\varepsilon}}^{k-1}\cdot W_{1}}$-emulator or the $\paren{3+\varepsilon, \paren{3+\frac{8}{\varepsilon}}^{k-1} \cdot W_{1}}$-spanner of Elkin, Gitlitz and Neiman \cite{ElkGitNei2019}. 

In this paper we show that the $+4W_{1}$-emulator of size $\tilde O\paren{n^{\frac{4}{3}}}$ of Elkin, Gitlitz and Neiman \cite{ElkGitNei2021} is not an isolated construction, but rather the second instance ($k=3$) of our result. Starting from $k\geq 4$, we introduce an unavoidable multiplicative stretch. By \defref{strongly_commensurate}, any  additive weighted emulator or spanner must correspond to an unweighted emulator with a constant additive stretch and the same size (e.g. $+4W_{\max}$, $+4W_{1}$ or $+2W_{1}+2W_{2}$ all correspond to $+4$). However, the \qoute{additive barrier} of Abboud and Bodwin \cite{AbbBod2015} states that any such unweighted structure requires at least $\tilde \Omega\paren{n^{\frac{4}{3}}}$ edges. Consequently, going below the $\tilde \Omega\paren{n^{\frac{4}{3}}}$ density threshold to achieve a sparsity of $\tilde O\paren{n^{1+\frac{1}{k}}}$ for $k \geq 4$ cannot be done by an additive approximation.  

The $k=2$ instance of our framework is a special case that yields an additive $+2W_{1}$-spanner. However, for all $k \geq 3$, the construction strictly produces an emulator. This transition is motivated by a  density discrepancy between the two structures. For instance, a $+4W_{1}$-emulator requires only $\tilde O\paren{n^{\frac{4}{3}}}$ edges \cite{ElkGitNei2021}, whereas the sparsest known $+4W_{1}$-spanner requires $\tilde O\paren{n^{\frac{7}{5}}}$ edges \cite{AhmBodHamKobSpe2021}. 

These density gaps suggest a fundamental \textit{separation}, where an $\paren{\alpha,\beta}$-emulator successfully bypass the size of an $\paren{\alpha,\beta}$-spanner. While a   proof for this gap remains a major open problem, the recurring size differences in state of the art constructions  lead to the following question:

\begin{question}~\label{que:separation}
For which values of $\alpha,\beta$ does there exist an $\paren{\alpha,\beta}$-emulator that is strictly sparser than any known $\paren{\alpha,\beta}$-spanner?
\end{question}

To better understand these discrepancies, we examine this question specifically within the unweighted regime. Two widely known results are the $\paren{2k-1}$-spanners \cite{BasSen2007, ThoZwi2005, AltDasDobJosSoa1993} and the $\paren{k, k-1}$-spanner \cite{BasKavMehPet2010}, each of size $\tilde O\paren{n^{1+\frac{1}{k}}}$. While the $\paren{k, k-1}$-spanner essentially \qoute{shifts} some multiplicative stretch from the $\paren{2k-1, 0}$-spanner into an additive stretch, no other intermediate trade-offs are currently known. 

Our framework fills this gap by providing a new family of trade-offs of the same $\tilde O\paren{n^{1+\frac{1}{k}}}$ density. In the unweighted setting, we may assume $w\paren{e} = 1$ for all $e\in E$. This implies $W_{1} = W_{2} = 1$ regardless of the vertices in question. Our construction therefore yields a $\paren{2\cdot \floor{\frac{k}{2}} - 1, 2\cdot \ceil{\frac{k}{2}} + \max\{0, 2\cdot\paren{\ceil{\frac{k}{2}} - 2}\}}$-emulator of size $\tilde{O}\paren{n^{1+\frac{1}{k}}}$. This simplifies to a $\paren{k-1, 2k-4}$-emulator for an even $k$, and a $\paren{k-2, 2k-2}$-emulator for an odd $k$. There is currently no known spanner with the same stretch. 

Lastly, we evaluate the performance of our framework in the unweighted regime against the most prominent constructions of the same size. When compared to the $\paren{k,k-1}$-spanner of Baswana, Kavitha, Mehlhorn and Pettie \cite{BasKavMehPet2010}, our construction provides a tighter approximation for pairs of vertices with a distance of $\delta_G\paren{u,v}\geq k-3$ for even $k$, or $\delta_G\paren{u,v} \geq \frac{k-1}{2}$ for odd $k$. Furthermore, our framework strictly improves upon the $\paren{k-1,2k}$-spanner of Elkin and Peleg \cite{ElkPel2001}, for any $k\in\mathbb{N}$.

\begin{table}[t]
  \begin{center}
\small\addtolength{\tabcolsep}{-3pt}
    \renewcommand\arraystretch{1.5}
    \begin{tabular}{|c|c|c|c|c|}
      \hline
      \textbf{\small Stretch} & \textbf{\small Weighted} & \textbf{\small Limitation} & \textbf{\small Reference} & \textbf{\small Partially Answers}\\
      \hhline{|=|=|=|=|=|}
      
      $( 2\cdot \lfloor \frac{k}{2} \rfloor-1, 2\cdot \lceil \frac{k}{2} \rceil \cdot W_{1} $ & \multirow{4}{*}{\checkmark} & \multirow{2}{*}{--} & \multirow{4}{*}{\secref{_k-1,kW1+_k-4_W2_}} & \multirow{4}{*}{\querefs{n^{1+1/k}, w2}} \\
      $+\max\{0,2\cdot(\lceil\frac{k}{2}\rceil-2)\}\cdot W_{2} )$ & & & & \\ \cline{1-1} \cline{3-3}
      
      $\paren{k-1,k\cdot W_{1} +\paren{k-4}\cdot W_2}$ & & even $k$ & & \\ \cline{1-1} \cline{3-3}
      
      $\paren{k-2,\paren{k+1}\cdot W_{1} +\paren{k-3}\cdot W_2}$ & & odd $k$ & & \\ \hline

      $\paren{k-1,2k-4}$ & \multirow{2}{*}{--} & even $k$ & \multirow{2}{*}{\secref{_k-1,2k-4_}} & \multirow{2}{*}{\querefs{n^{1+1/k}, separation}} \\ \cline{1-1} \cline{3-3}
      
      $\paren{k-2,2k-2}$ & & odd $k$ & & \\ \hline

    \end{tabular}
    \renewcommand\arraystretch{1}
    \caption{Implications of our result. The size is  $\tilde O \paren{n^{1+\frac{1}{k}}}$ and the construction time is $\tilde O \paren{n^{2+\frac{1}{k}}}$ for $k\leq 4$ and $\tilde O \paren{n^{2+\frac{2}{k}}}$ for $k\geq 5$.}~\label{tab:our_res}
  \end{center}

  \vspace{-6mm}
\end{table}

It is also interesting to compare our result to the $+\paren{6^k -1}\cdot \delta ^{\paren{1-\frac{1}{k}}}$-emulator of size $\tilde{O}\paren{n^{1+\frac{1}{2^{k+1}-1}}}$ of Thorup and Zwick \cite{ThoZwi2006}. To do so, observe that our result computes a $\paren{2^{k+1}-3, 2^{k+2}-4}$-emulator of size $\tilde O \paren{n^{1+\frac{1}{2^{k+1}-1}}}$. While the Thorup and Zwick emulator eventually provides a better stretch for sufficiently large distances, our construction maintains a tighter approximation for all vertices within \qoute{short distance} of one another. Notably, the threshold for this range grows exponentially at a rate of approximately $3^{k^{2}}$. For $k=2$, our stretch is tighter for all pairs whose original distance is at most $70$, a threshold that expands rapidly to $5,744$ for $k=3$ and reaches $4,575,579$ for $k=4$. This implies that for most practical values of $k$, our framework provides a superior approximation for essentially all relevant distances in the graph. We provide a comprehensive list of comparative values in \tabref{_k-1,2k-4__vsThorupZwick} and a visual illustration in \figref{_k-1,2k-4__vsThorupZwick}. A complete summary of our results and their implications is provided in \tabref{our_res}.

\section{Toolkit}\label{toolkit}
\subsection{Basic graph notions}\label{notions}
Unless stated otherwise, we consider an undirected weighted graph $G=\paren{V,E,w}$, where $w:E\rightarrow \mathbb{R}^{\geq 0}$ is a non-negative weight function. We denote the weight of an edge by $w\paren{u,v}$, the number of vertices by $\abs{V} = n$ and the number of edges by $\abs{E}= m$. For any $k\in \mathbb{N}$, let $\bracke{k}=\bracce{1,2,\ldots,k}$. 

A \emph{path} from a vertex $u$ to a vertex $v$ is an ordered sequence of vertices $u=y_{0},y_{1},y_{2},\ldots,y_{k} , y_{k+1}=v$ such that for every $i\in\bracke{k+1}$ it holds that $\paren{y_{i-1},y_{i}}\in E$. The weight of a path is defined as $w\paren{P} = \msum{e\in P}{w\paren{e}}$. If the graph is unweighted, the weight of a path is simply the number of edges in the path $w\paren{P} = \abs{P}$. A path minimizing $w\paren{P}$ is referred to as a \textit{shortest path} and its weight is the \emph{distance} between its endpoints $u$ and $v$, denoted by $\delta\paren{u,v}$.  For a vertex $u$ and a set $S\subseteq V$, we define the distance $\delta\paren{u,S}$ as the minimum distance from $u$ to any vertex in $S$, namely $\delta\paren{u,S} = \mmin{\delta\paren{u,s}}{s\in S}$.

We denote by $u \squiggly v$ a shortest path between $u$ and $v$, hence $w\paren{u\squiggly v} = \delta\paren{u,v}$. In the special case where this path consists of a single edge, we denote it by $u\squiggly v = u\straight v$, and consequently we have $w\paren{u,v}=\delta\paren{u,v}$. Along a fixed path $u\squiggly v$, we define $W_{i}\paren{u\squiggly v}$ to be the weight of its $i^{\textnormal{th}}$ heaviest edge. Therefore, $W_{1}\paren{u\squiggly v}$ is the heaviest, $W_{2}\paren{u\squiggly v}$ is the second heaviest, etc. Expressions of the form $\msum{W_{i}}{i\in I}$ are always interpreted with respect to a single, fixed path $u\squiggly v$. When no explicit path is specified, we consider a shortest path between $u$ and $v$ that minimizes the upper-bound.

Throughout our construction, we would denote: by $d\bracke{u,v}$ the distance estimate between $u$ and $v$ computed by our algorithm, by $\delta_H \paren{u,v}$ the distance in the emulator $H$ and by $\delta_G\paren{u,v}$ the distance in the original graph $G$. When providing a construction for an $\paren{\alpha,\beta}$-emulator $H=\paren{V,F}$, we always prove that $\delta_H\paren{u,v} \in \bracke{\delta_G\paren{u,v}, \alpha\cdot \delta_G\paren{u,v}+\beta}$.

\subsection{Balls and pivots}~\label{nearest}
For a parameter $q\in\paren{0,1}$ we now consider a set $S$ of vertices $s\in V$ sampled randomly and independently with probability $q$ each. For a vertex $u\in V$, we define its \emph{pivot} with respect to $S$ as a vertex $p_{S}\paren{u}\in S$ satisfying $\delta\paren{u,p_{S}\paren{u}}=\delta\paren{u,S}$. That is, the pivot of $u$ is the closest vertex to $u$ among members of  $S$. If there are several candidates we can break ties by a consistent criteria. We also associate $u$ with its \emph{ball} with respect to $S$, defined as: $\ball{u}{S} = \bracce{ v \in V \midline \delta\paren{u,v}<\delta\paren{u,p_{S}\paren{u}}}$.

\begin{observation}[\cite{ThoZwi2005}]~\label{obv:hssize} 
$E\bracke{\abs{S}} \in \tilde O \paren{nq}$.
\end{observation}

\begin{observation}[\cite{ThoZwi2005}]~\label{obv:ballsize} 
$E\bracke{\abs{\ball{u}{S}}} \in \tilde O \paren{\frac{1}{q}}$.
\end{observation}

Additionally, we define a  classification of the graph's edges. We associate each vertex $u$ with the set  $E_{S}\paren{u}=\bracce{\paren{u,v} \midline w\paren{u,v}<\delta\paren{u,p_S\paren{u}}}$. We then define the set: $E_{S} = \mcup{E_{S}\paren{u}}{u\in V}$. 

\begin{observation}[\cite{ThoZwi2005}]~\label{obv:edgessize} 
$E\bracke{\abs{E_{S}}} \in \tilde O \paren{\frac{n}{q}}$.
\end{observation}

By definition, every  edge $\paren{x,y} \notin E_S$ satisfies $\delta\paren{x,p_{S}\paren{x}} \leq w\paren{x,y}$ and $\delta\paren{y,p_{S}\paren{y}} \leq w\paren{x,y}$. Consequently, any such edge along a  $u\squiggly v$ provides an upper bound for $\delta_{G}\paren{p_{S}\paren{x}, p_{S}\paren{y}}$. This structural guarantee is fundamental in our subsequent distance estimates.

\begin{observation}[\cite{BasKav2010}]~\label{obv:pathinball} 
Let $u\in V$, $S\subseteq V$ and $v\in \ball{u}{S}$. Then: $u\squiggly v \subseteq E_{S}$.
\end{observation}

We observe that both the pivots themselves and the exact distances from every vertex to its pivot can be computed using a single SSSP computation, by introducing an auxiliary vertex connected to all vertices in $S$ with edges of weight $0$.

\begin{observation}[\cite{AinCheIndMot1999}]~\label{obv:pivotsdistance} 
The pivot $p_{S}\paren{u}$ of each vertex $u\in V$ and the exact distance $d\bracke{u,p_{S}\paren{u}} = \delta\paren{u,p_{S}\paren{u}}$ can be computed in $\tilde O \paren{m}$ time.
\end{observation}

Additionally, we can extend the notion of a ball to consider only vertices from a certain set $A\subseteq V$. Formally, we define
$\ball{u}{A}{S} = \bracce{ a\in A \midline \delta\paren{u,a} < \delta\paren{u,p_{S}\paren{u}}}$. Let $q'\in \paren{0,1}$ be the probability for a vertex from $A$ to belong to $S$. 

\begin{observation}[\cite{ThoZwi2005}]~\label{obv:bunchsize} 
$E\bracke{\abs{\ball{u}{A}{S}}} \in \tilde O \paren{\frac{1}{q'}}$.
\end{observation}

\subsection{Hierarchy of hitting sets}~\label{hierarchy}
We now consider $\ell$ sets $S_1, S_2, \ldots, S_\ell \subseteq V$ and probabilities $q_{1},q_{2},\ldots , q_{\ell} \in\paren{0,1}$ where, for each $i\in\bracke{\ell}$, a vertex from $S_{i-1}$ is independently sampled to $S_{i}$ with probability $q_{i}$. It follows that $S_{i-1}\supseteq S_{i}$. 

We denote $p_i\paren{u} = p_{S_i}\paren{u}$ and $E_{i}=E_{S_{i}}$. Let us consider $\ball{u}{i-1}$ and  $\ball{u}{i}$. As $S_{i-1}\supseteq S_{i}$, it follows that any vertex $v$ that satisfies $\delta\paren{u,v}<\delta\paren{u,S_{i-1}}$ also satisfies $\delta\paren{u,v}<\delta\paren{u,S_{i}}$. Hence, $\ball{u}{i-1}\subseteq \ball{u}{i}$. By the same argument it follows that $E_{i-1}\paren{u} \subseteq E_{i}\paren{u}$. Thus, $E_{i-1}\subseteq E_{i}$. This hierarchical structure of pivots, balls, and edge sets provides a framework for our distance estimations.

\begin{observation}[\cite{DorHalZwi2000, BasKav2006}]~\label{obv:ESiwithinESi+1}
For $i<j$ it holds that $S_{j}\subseteq S_{i}$ and  $ E_{i} \subseteq E_{j}$.
\end{observation}

For completeness, we set $S_{0}= V$ and $S_{\ell+1} = \varnothing$. This proves to be consistent with the hierarchy definition we presented, and it ensures that $E_0 = \varnothing$ and $E_{\ell+1}=E$. Hence, $V = S_0 \supseteq S_1 \supseteq S_2 \supseteq \ldots \supseteq S_{\ell} \supseteq S_{\ell+1} = \varnothing$ and $\varnothing = E_{0} \subseteq E_{1} \subseteq E_{2} \subseteq \ldots \subseteq E_{\ell} \subseteq E_{\ell+1} = E$. Let $j\in \bracke{\ell}$. By applying \obvref{hssize} and \obvref{ballsize}:

\begin{observation}[\cite{DorHalZwi2000, BasKav2006}]~\label{obv:hssize2}
$\abs{S_{j}} \in \tilde O\paren{n\cdot \mprod{q_{k}}{k=1}{j}}$. 
\end{observation}

\begin{observation}[\cite{DorHalZwi2000, ThoZwi2005}]~\label{obv:ballsize2}
$\abs{\ball{u}{j}} \in \tilde O\paren{\mprod{\frac{1}{q_{k}}}{k=1}{j}}$. 
\end{observation}

Let $P$ be a path between $u$ and $v$ and let $i\in\bracke{\ell}$. As mentioned in \secref{notions}, an edge  $\paren{x,y}\in P$ such that $\paren{x,y}\notin E_i$ provides an upper bound on the distance from each endpoint $x$ or $y$ to its corresponding pivot $p_{i}\paren{x}$ or $p_{i}\paren{y}$. To properly make-use of this observation, we denote by $\sizeindex{P}{i}$ the number of edges in $P$ that are not in $E_i$. If $\sizeindex{P}{i}>0$, we denote by $\rightabovearrow{P}{i}$ the first edge, from $u$'s perspective, that does not belong to $E_{i} $ and by $\leftabovearrow{P}{i}$ the last edge, from $u$'s perspective, that does not belong to $E_{i}$. 

\begin{observation}~\label{obv:Pi=0,1}
$\sizeindex{P}{i}=0$ iff $P\subseteq E_{i}$ and $\sizeindex{P}{i}=1$ iff $\rightabovearrow{P}{i} = \leftabovearrow{P}{i}$.
\end{observation}

We consider auxiliary edges, that connect each vertex to its pivot. Namely, these are $D_{i} = \bracce{\paren{u,p_{i}\paren{u}} \midline u\in V}$. We set $D = \mcup{D_{i}}{i=1}{\ell} $ as their union. 

\begin{observation}[\cite{DorHalZwi2000, CohZwi1997}]~\label{obv:edgesToPivots}
Let $i\in \bracke{\ell}$. Then: $\abs{D_{i}}\in O \paren{n}$. Additionally, $\abs{D}\in O\paren{n\cdot \ell}$. If $\ell \in \tilde O \paren{1}$ then $\abs{D} \in \tilde O\paren{n}$. 
\end{observation}

Additionally, we consider refined balls of the form $\ball{u}{i}{j}$ for $i < j$, which are generally smaller than the standard $\ball{u}{j}$. By \obvref{ballsize2} and \obvref{bunchsize}:

\begin{observation}[\cite{ThoZwi2005}]~\label{obv:bunchsize2}
Let $i,j\in \bracke{\ell}\cup\bracce{0}$ such that $i<j$. Then: $\abs{\ball{u}{i}{j}} \in \tilde O\paren{\mprod{\frac{1}{q_{k}}}{k=i+1}{j}}$. 
\end{observation}

For example, when $j = i+1$, we define
$\ball{u}{i}{i+1} = \bracce{ s\in S_{i} \midline \delta\paren{u,s} < \delta\paren{u,p_{i+1}\paren{u}}}$. These are
the members of $S_i$ that are closer to $u$ than its pivot $p_{i+1}\paren{u}$ in the next-level set $S_{i+1}$. Our result is focused around such sets, as they may provide a finer distance estimation than what could be achieved solely with the standard sets $\ball{u}{i}$ and the edge sets $E_i$.

Lastly, having these sets, we consider auxiliary edges from a vertex $u$ to members of $\ball{u}{i}{i+j}$ for $j\in\bracke{\ell}$ and $i\in\bracke{\ell-j}$. We define: $B_{j}\paren{u} = \bracce{ \paren{u,s} \midline \exists {i\in\bracke{\ell-j}\cup\bracce{0}}: s\in\ball{u}{i}{i+j}}$. For example, when $j=1$ we consider the set $B_{1}\paren{u} = \bracce{\paren{u,s} \midline \exists{i\in\bracce{0,1,\ldots,\ell-1}: s\in \ball{u}{i}{i+1}}}$. This is the set of edges from $u$ to any $s\in S_{i}$ such that $s\in \ball{u}{i}{i+1}$. By \obvref{bunchsize2}:

\begin{observation}~\label{obv:edgestobunches}
Let $u\in V$ and $j\in\bracke{\ell}$.  Then: $\abs{B_{j}\paren{u}} \in \tilde O\paren{ \msum{  \mprod{\frac{1}{q_{k}}}{k=i+1}{i+j}}{i=0}{\ell-j}}\subseteq \tilde O\paren{ \ell \cdot \mmax{\mprod{\frac{1}{q_{k}}}{k=i+1}{i+j}}{i=0}{\ell-j}}$. If $\ell\in \tilde O\paren{1}$ this becomes $\tilde O\paren{ \mmax{\mprod{\frac{1}{q_{k}}}{k=i+1}{i+j}}{i=0}{\ell-j}}$.
\end{observation}

For a subset $A\subseteq V$, we define $B_{j}\paren{A} = \mcup{B_{j}\paren{a}}{a\in A}$. The size of $B_{j}\paren{A}$ for certain subsets $A$ will be of high importance for our emulators' size analysis:

\begin{observation}~\label{obv:edgestobunches2}
Let $A\subseteq  V$,  $j\in\bracke{\ell}$ and assume $\ell\in \tilde O\paren{1}$.  Then: $\abs{B_{j}\paren{A}} \in \tilde O\paren{\abs{A} \cdot \mmax{\mprod{\frac{1}{q_{k}}}{k=i+1}{i+j}}{i=0}{\ell-j}}$.
\end{observation}

\section{Overview}\label{overview}
To construct our emulator $H=\paren{V,F,w_{H}}$, we utilize a hierarchy of nested hitting sets $V = S_{0} \supseteq S_{1} \supseteq \ldots \supseteq S_{k} = \varnothing$, where each set $S_{i}$ defines an $i^{\textnormal{th}}$ level pivot $p_{i}\paren{u}$ for every vertex $u \in V$. At a high level, our construction proceeds in two phases:

\begin{enumerate}
    \item Sample each $u\in S_i$ into $S_{i+1}$ independently with probability $q_{i+1} = q = n^{-\frac{1}{k}}$ to form the hierarchy,
    \item Construct $F$ as the union of the edge sets\footnote{We observe that each set is of size $\tilde O \paren{n^{1+\frac{1}{k}}}$.}: $D$, $E_{1}$, $S_{i-1}\times S_{k-i}$ for $i\in \bracke{k}$, $B_{1}\paren{V}$ and $B_{2}\paren{S_{1}}$.
\end{enumerate}

While such hierarchies are also foundational for \textbf{A}ll-\textbf{P}airs \textbf{A}pproximate \textbf{S}hortest \textbf{P}aths (\emph{APASP}) \cite{CohZwi1997, DorHalZwi2000, BasKav2010} and \textbf{A}pproximate \textbf{D}istance \textbf{O}racles (\emph{ADO}) \cite{ThoZwi2005, BasKav2010, DorForKirNazVasVos2023}, the construction of an emulator fundamentally differs from them.  APASP and ADO algorithms optimize runtime by invoking \textbf{S}ingle \textbf{S}ource \textbf{S}hortest \textbf{P}aths (\emph{SSSP}) from small vertex  sets over edge set (possibly, different) of corresponding size.  In contrast, an emulator must commit to a single edge set $F$.

This commitment precludes techniques standard in APASP or ADO. For example, Cohen and Zwick \cite{CohZwi1997} can \qoute{bypass} denser regions of the graph by adding auxiliary edges $\bracce{u} \times V$ per SSSP invocation from $u$. For an emulator to commit to  all such sets yields the complete cartesian product $V \times V$. A similar phenomenon appears in Dor, Halperin and Zwick \cite{DorHalZwi2000} or in  Roditty and Sapir \cite{RodSap2025}, when an APASP algorithm utilizes cartesian products $S_j \times S_\ell$ in an SSSP invocation. The size of these products is proportional to the source set from which the SSSP is invoked. In contrast, emulators may commit only to one size of a cartesian product. Finally, there is an issue of the nested edge sets: invoking SSSP from a source set $S_{i}$ over a specific edge set $E_{i+1}$, ensures the total complexity remains within   $\abs{S_{i}} \cdot \abs{E_{i+1}}$. This is most notably used by Baswana and Kavitha \cite{BasKav2010}.  In the case of $i=k-1$ for example, this results in $E_{k}=E$. While this is permissible in APASP due to the small size of $S_{k-1}$, an emulator attempting to replicate this step would be forced to include $E$ entirely. 

Consequently, identical hierarchies yield strictly weaker bounds for emulators than for APASP. For instance, the same hierarchy used to achieve the $+2W_{1}$-APASP algorithm of Cohen and Zwick \cite{CohZwi1997} can only be utilized to achieve the  $+4W_{1}$-emulator of Elkin, Gitlitz and Neiman \cite{ElkGitNei2021}. This discrepancy is emphasized in the additive regime. While additive APASP permits various runtime-stretch trade-offs \cite{SahYe2023, RodSap2025, CohZwi1997}, the $\tilde \Omega\paren{n^{\frac{4}{3}}}$ lower bound of Abboud and Bodwin \cite{AbbBod2015} prohibits the existence of additive emulators.

Our approach relies on bounding distance estimations between pivots of vertices along the shortest path. To maintain the required emulator size of $\tilde O\paren{n^{1+\frac{1}{k}}}$, we can only afford auxiliary edges bridging pivots of levels $i$ and $j$ if $i+j=k-1$. Consequently, as $k$ increases, our construction must connect increasingly higher level pivots. Since $\delta_{G}\paren{x,p_{i+1}\paren{x}} \geq \delta_{G}\paren{x,p_{i}\paren{x}}$ for any $x \in V$, upper bounding distances through higher level pivots inherently yields a \qoute{weaker} approximation. In other words, as $k$ increases, the overall stretch increases as well. While the resulting stretch and size  depends on $k$, our case analysis approach circumvents any structural dependence on $k$. Let  $u,v\in V$ and fix a  $u\squiggly v$. We classify $u\squiggly v$ into exactly one of  three cases, based solely on edges missing from $E_{1}$:

\begin{enumerate}
    \item \case{a}: $\sizeindex{u\squiggly v}{1}=0$,
    \item \case{b}: $\paren{x,y}=\rightabovearrow{u\squiggly v}{1}=\leftabovearrow{u\squiggly v}{1}$,
    \item \case{c}:  $\paren{x,y}=\rightabovearrow{u\squiggly v}{1}\neq \leftabovearrow{u\squiggly v}{1}=\paren{z,w}$.
\end{enumerate}

If \case{a} occurs, $\delta_{H}\paren{u,v}=\delta_{G}\paren{u,v}$ as $E_{1}\subseteq F$. Roughly speaking, \case{b} increases the additive stretch, while \case{c} increases the multiplicative stretch. For example, when moving from $k=3$ to $k=4$, we increase the multiplicative stretch from $1$ to $3$ due to \case{c}. As $\abs{S_{1}} = n^{\frac{3}{4}}$ and we aim for size $n^{\frac{5}{4}}$, we cannot use an auxiliary edge of the form $S_{1}\times S_{1}$. We can, however, use $S_{1}\times S_{2}$, since $\abs{S_{2}} = n^{\frac{1}{2}}$. While $\delta_{G}\paren{x,p_{1}\paren{x}}\leq w\paren{x,y}$ and $ \delta_{G}\paren{w,p_{1}\paren{w}}\leq w\paren{z,w}$, we bound the distances to the second level pivots by $\delta_{G}\paren{x,p_{2}\paren{x}}\leq \delta_{G}\paren{x,w}+ w\paren{z,w}$ and $\delta_{G}\paren{w,p_{2}\paren{w}}\leq \delta_{G}\paren{x,w}+ w\paren{x,y}$, under additional assumptions specified later. At worst $u=x$ and $w=v$, thus our  upper bound is inherently restricted by $\delta_{G}\paren{x,w} \leq \delta_{G}\paren{u,v}$. Consequently, the multiplicative stretch increases from $1$ to $3$.

This establishes a stretch-size trade-off parameterized by $k$. Increasing $k$ improves sparsity but yields alternating penalties: an odd-to-even transition increases the multiplicative stretch by $2$, whereas an even-to-odd transition increases the additive stretch by $2W_{1}+2W_{2}$.  Formalizing this alternating behavior, we derive our generalized $\paren{2\cdot \floor{\frac{k}{2}}-1, 2\cdot \ceil{\frac{k}{2}} \cdot W_{1} + \max\bracce{0, 2\cdot\paren{\ceil{\frac{k}{2}}-2}}\cdot W_{2}}$-emulator. 

We initially present our constructions in a simplified conceptual form without prioritizing runtime efficiency. We assume that  the weight of any auxiliary edge $\paren{u,v}\in F$ that did not exist in $E$ is $w_{H} \paren{u,v} = \delta_{G} \paren{u,v}$. Finally, all our constructions can be  adapted to run in $\tilde O \paren{n^{2+\frac{1}{k}}}$ time.  This is achieved by a two-phase algorithmic approach:

\begin{enumerate}
    \item Invoke an $\paren{\tau,\eta}$-APASP algorithm for certain parameters $\tau$ and $\eta$, resulting in a distance estimation matrix $d$,
    \item Use these estimates to construct a sparse set of edges, resulting in an  $\paren{\alpha,\beta}$-emulator. The weights of each edge are given by the distance estimation matrix $d$.
\end{enumerate}

\subsection{Main Technical Contributions}\label{contributions}

Our framework introduces two  theoretical innovations. The first is algorithmic. While previous constructions utilize edges such as $B_{1}\paren{V}$, these sets are primarily effective for bounding distances to second level pivots. Relying exclusively on $B_{1}\paren{V}$ for higher level pivots leads to a rapidly increasing multiplicative stretch as  it accumulates over  the levels. To circumvent this, we introduce the auxiliary edge set $B_{2}\paren{S_{1}}$, which allows us to bound third level pivots directly through first level pivots. By maintaining tight distance estimates at these base hierarchy levels, we prevent the rapid growth of multiplicative stretch in all subsequent levels. Surprisingly, extending this approach to higher level edge sets $B_{j}\paren{S_{j-1}}$ for $j \geq 3$ (e.g. $B_{3}\paren{S_{2}}$) is ineffective. Such edge sets \qoute{shift} additive stretch into multiplicative stretch, making them redundant.

The second innovation regards our proof methodology. Classifying paths strictly by edges absent from $E_{1}$ ensures our
analysis remains constrained to three cases regardless of the hierarchy depth $k$. In contrast, a naive case distinction based on missing edges across the full hierarchy $E_{1}, E_{2}, \ldots, E_{k-1}$ might appear advantageous as it provides immediate distance information to higher level pivots, however, it necessitates an escalating number of $5+6k$ cases. Furthermore, it results in the same stretch bound. 

\subsection{Structure of the paper}\label{structure}
As a warm-up, we review in \secref{+4W1}  the $+4W_{1}$-emulator ($k=2$) of Elkin, Gitlitz and Neiman \cite{ElkGitNei2021}, establishing the case analysis framework for all subsequent proofs. Building on these mechanics, \secref{_3,4W1_} introduces a $\paren{3,4W_{1}}$-emulator ($k=3$) of size $\tilde O \paren{n^{\frac{5}{4}}}$. These results are generalized to arbitrary $k \in \mathbb{N}$ in \secref{_k-1,kW1+_k-4_W2_}, yielding a $\paren{2\cdot \left\lfloor \frac{k}{2} \right\rfloor-1, 2\cdot \left\lceil \frac{k}{2} \right\rceil \cdot W_{1}+\max\left\{0,2\cdot\left(\left\lceil\frac{k}{2}\right\rceil-2\right)\right\}\cdot W_{2}}$-emulator of size $\tilde O \paren{n^{1+\frac{1}{k}}}$. Finally, we discuss in \secref{_k-1,2k-4_}  the advantages of our emulator  compared to the classical result of Thorup and Zwick \cite{ThoZwi2006}, for vertices whose distance is \qoute{short} (i.e. $\leq O\paren{3^{k^2}}$).
We present a runtime efficient implementation for our weighted emulators in \appref{runtime}. Lastly, we present a slightly refined stretch analysis in \appref{refined}.

\section{Warm-up: additive $+4W_{1}$-emulator}\label{+4W1}

\begin{algorithm}[H]
\caption{$+4W_{1}$-emulator$\paren{G=\paren{V,E,w}}$}\label{alg:+4W1}

\textbf{Input:} An undirected non-negative weighted graph $G=\paren{V,E,w}$.

\textbf{Output:} A $+4W_{1}$-emulator. 

\medskip

Compute APSP

Let $\beta,\gamma \in \paren{ 0 , 1}$ to be fixed later

Sample each vertex of $V$ with probability $n^{-\beta}$ to $S_{1}$

Sample each vertex of $S_{1}$ with probability $n^{-\gamma}$ to $S_{2}$

For every $u\in V$ find the pivot $p_{1}\paren{u}$ (respectively, $p_{2}\paren{u}$)  

Construct the set of edges $D$

Construct the set of edges $E_{1}$

Set $F\leftarrow D \cup E_{1} \cup \paren{S_{1} \times S_{1} }$

Initialize $w_{H}\leftarrow \varnothing$

\For{$\paren{u,v}\in F$}
    {
        $w_{H} \paren{u,v} \leftarrow \delta_{G} \paren{u,v}$
    }

\Return $H=\paren{V,F,w_{H}}$

\end{algorithm}

As a warm-up, we review the construction of the $+4W_{1}$-emulator by Elkin, Gitlitz and Neiman \cite{ElkGitNei2021}, which yields an emulator of size $\tilde O \paren{n^{\frac{4}{3}}}$. Following the framework in \secref{overview}, we assume distances are known a priori, as the final weights are defined by $w_{H}\paren{u,v} \leftarrow \delta_{G}\paren{u,v}$ for all $\paren{u,v}\in F$.

\algref{+4W1} works as follows: first, we sample vertices of $V$, with probability $n^{-\beta}$ to form the set $S_{1}$. Second, we sample vertices of $S_{1}$ with probability $n^{-\gamma}$ to form the set $S_{2}$. For any $u\in V$, we identify the designated pivots $p_{1}\paren{u}$ and $p_{2}\paren{u}$ and construct the set $D$ containing all edges between vertices and their respective pivots (See: \obvref{pivotsdistance}). Utilizing these distances, we construct the set $E_{1}$ consisting of edges strictly nearer than a pivot (See: \obvref{edgessize}). 

Finally, we define $F$ as the union of $D$, $E_{1}$, $S_{1}\times S_{1}$ and $V\times S_{2}$. We note that the set $S_{2}$, and consequently the edges $V\times S_{2}$, are not required for the stretch analysis of the present construction. Indeed, one may omit the second level altogether and obtain the same $+4W_{1}$-emulator guarantee using only $D$, $E_{1}$ and $S_{1}\times S_{1}$. We nevertheless retain $S_{2}$ and $V\times S_{2}$ in the construction, as they will become necessary when we later consider a runtime efficient implementation that does not assume all distances are known a priori (See: \appref{runtime_k=3}).

This concludes the description of \algref{+4W1}. Let $u,v\in V$ and fix a $u\squiggly v$. Our goal is to prove that  $\delta_{G}\paren{u,v}\leq \delta_{H}\paren{u,v}\leq \delta_{G}\paren{u,v}+4W_{1}\paren{u\squiggly v}$. To prove this, we consider a case analysis approach, based on the values of $\sizeindex{u\squiggly v}{1}$, $\rightabovearrow{u\squiggly v}{1}$ and $\leftabovearrow{u\squiggly v}{1}$ (See: \secref{hierarchy}).

\begin{itemize}
    \item \case{a}: $\sizeindex{u\squiggly v}{1} = 0 $,
    \item \case{b}: $\paren{x,y}=\rightabovearrow{u\squiggly v}{1}=\leftabovearrow{u\squiggly v}{1}$,
    \item \case{c}: $\paren{x,y}=\rightabovearrow{u\squiggly v}{1} \neq \leftabovearrow{u\squiggly v}{1}=\paren{z,w}$.
\end{itemize}

This case analysis approach serves as the foundational template for our generalized framework. \lemref{+4W1_case_a}, \lemref{+4W1_case_b} and \lemref{+4W1_case_c} bound the distance $\delta_{H}\paren{u,v}$ for \case{a}, \case{b} and \case{c}, respectively, after which we prove the overall sparsity is bounded by $\tilde O\paren{n^{\frac{4}{3}}}$.

\begin{lemma}~\label{lem:+4W1_case_a} 
   If \case{a} holds then $\delta_H\paren{u,v} = \delta_G\paren{u,v}$.
\end{lemma}
\begin{proof} 
Recall that $F$ includes $E_1$. As $\sizeindex{u\squiggly v}{1} =0$, it follows that $u\squiggly v\subseteq E_{1}$. Hence, $u\squiggly v \subseteq F$ which means that $\delta_H\paren{u,v} = \delta_G\paren{u,v}$.
\end{proof}

\begin{lemma}~\label{lem:+4W1_case_b} 
   If \case{b} holds then $\delta_H\paren{u,v} \leq \delta_G\paren{u,v} + 4W_{1}\paren{u\squiggly v}$.
\end{lemma}
\begin{proof} 
Let $\paren{x,y} \in u\squiggly v$ be the sole edge such that $\paren{x,y} \notin E_{1}$. By the definition of $E_{1}$, this implies $\delta_{G}\paren{x,p_{1}\paren{x}}, \delta_{G}\paren{y,p_{1}\paren{y}} \leq w\paren{x,y}$. Since $\paren{x,y}$ is the only edge from $u\squiggly v$ missing from $E_{1}$, the subpaths satisfy $u\squiggly x, y\squiggly v \subseteq E_{1}$. Furthermore, the auxiliary edges $\paren{x,p_{1}\paren{x}},\paren{y,p_{1}\paren{y}}\in D$ and $\paren{p_{1}\paren{x},p_{1}\paren{y}}\in S_{1}\times S_{1}$. Consequently, the entire path $u\squiggly x \straight p_{1}\paren{x} \straight p_{1}\paren{y} \straight y \squiggly v\subseteq  E_{1} \cup D \cup \paren{S_{1}\times S_{1}} \subseteq F$. Therefore:

\begin{alignat*}{3}
        \delta_H\paren{u,v} & \leq \textubrace{ \delta_G\paren{u,x}+ \delta_G\paren{x,p_{1}\paren{x}} +  \delta_G\paren{p_{1}\paren{x},p_{1}\paren{y}} +  \delta_G\paren{p_{1}\paren{y},y}+ \delta_G\paren{y,v}}{Our preceding discussion} &\\
        & \leq \textubrace{ \delta_G\paren{u,x}+ 2\delta_G\paren{x,p_{1}\paren{x}} +  w\paren{x,y} +  2\delta_G\paren{p_{1}\paren{y},y}+ \delta_G\paren{y,v}}{Triangle inequality}\\
        & \leq \textubrace{ \delta_G\paren{u,v}+ 2\delta_G\paren{x,p_{1}\paren{x}} +  2\delta_G\paren{p_{1}\paren{y},y}}{$u\squiggly x\straight y \squiggly v$ is a shortest path}\\
        &\leq \textubrace{\delta_G\paren{u,v } + 2 \cdot w\paren{x,y}+ 2 \cdot w\paren{x,y}}{$\paren{x,y}\notin E_{1}$}\\
        & \leq \textubrace{ \delta_G\paren{u,v } + 4 \cdot W_{1}\paren{u\squiggly v}}{$w\paren{x,y} \leq W_{1}\paren{u\squiggly v}$ }
    \end{alignat*}

\end{proof}

\begin{lemma}~\label{lem:+4W1_case_c} 
   If \case{c} holds then $\delta_H\paren{u,v} \leq \delta_G\paren{u,v} + 2W_{1}\paren{u\squiggly v}+2W_{2}\paren{u\squiggly v}$.
\end{lemma}
\begin{proof} 
Let $\paren{x,y}=\rightabovearrow{u\squiggly v}{1}$ and $\paren{z,w} = \leftabovearrow{u\squiggly v}{1}$. Therefore, $\delta_{G}\paren{x,p_{1}\paren{x}} \leq w\paren{x,y}$ and $\delta_{G}\paren{w,p_{1}\paren{w}} \leq w\paren{z,w}$. By the selection of these edges, the subpaths satisfy $u\squiggly x, w\squiggly v \subseteq E_{1}$. Furthermore, the auxiliary edges $\paren{x,p_{1}\paren{x}}, \paren{w,p_{1}\paren{w}} \in D$ and $\paren{p_{1}\paren{x},p_{1}\paren{w}} \in S_{1}\times S_{1}$. Consequently, the path $u\squiggly x \straight p_{1}\paren{x} \straight p_{1}\paren{w} \straight w \squiggly v \subseteq E_{1} \cup D \cup \paren{S_{1}\times S_{1}} \subseteq F$. Therefore:

\begin{alignat*}{3}
        \delta_H\paren{u,v} & \leq \textubrace{ \delta_G\paren{u,x}+ \delta_G\paren{x,p_{1}\paren{x}} +  \delta_G\paren{p_{1}\paren{x},p_{1}\paren{w}} +  \delta_G\paren{p_{1}\paren{w},w}+ \delta_G\paren{w,v}}{Our preceding discussion} &\\
        & \leq \textubrace{ \delta_G\paren{u,x}+ 2\delta_G\paren{x,p_{1}\paren{x}} +  \delta_{G}\paren{x,w} +  2\delta_G\paren{p_{1}\paren{w},w}+ \delta_G\paren{w,v}}{Triangle inequality}\\
        & \leq \textubrace{ \delta_G\paren{u,v}+ 2\delta_G\paren{x,p_{1}\paren{x}} +  2\delta_G\paren{p_{1}\paren{w},w}}{$u\squiggly x\squiggly w \squiggly v$ is a shortest path}\\
        &\leq \textubrace{\delta_G\paren{u,v } + 2 \cdot w\paren{x,y}+ 2 \cdot w\paren{z,w}}{$\paren{x,y},\paren{z,w}\notin E_{1}$}\\
        & \leq \textubrace{ \delta_G\paren{u,v } + 2 \cdot W_{1}\paren{u\squiggly v}+ 2 \cdot W_{2}\paren{u\squiggly v}}{$w\paren{x,y} +w\paren{z,w}\leq W_{1}\paren{u\squiggly v}+W_{2}\paren{u\squiggly v}$ }
    \end{alignat*}

\end{proof}

We observe that \lemref{+4W1_case_b} constitutes the bottleneck for the current $k=2$ construction. However, for $k \geq 3$, both \case{b} and \case{c} will emerge as the dominant bounds. We  establish the emulator's size:

\begin{theorem}~\label{thm:+4W1} 
   \algref{+4W1} constructs a $+4W_{1}$-emulator of expected size $\tilde O\paren{n^{\frac{4}{3}}}$.
\end{theorem}
\begin{proof}
The correctness of the emulator follows from the case analysis in \lemref{+4W1_case_a} through \lemref{+4W1_case_c}. To analyze the expected size of $F$, we first note that \obvref{hssize} implies $\abs{S_{1}} = \tilde O \paren{n^{1-\beta}}$ and $\abs{S_{2}} = \tilde O \paren{n^{1-\beta-\gamma}}$. By \obvref{edgessize} the edge sets satisfy $\abs{E_{1}} = \tilde O \paren{n^{1+\beta}}$. By \obvref{edgesToPivots} the edge set $D$ contributes  $\tilde O\paren{n}$ edges. Finally, the product $S_{1}\times S_{1}$ and $V\times S_{2}$ have sizes $\tilde O \paren{n^{2-2\beta}}$ and $\tilde O\paren{n^{2-\beta-\gamma}}$, respectively. Thus, the total size of the emulator is bounded by:

$$ \abs{F} = \tilde O \paren{\abs{D} + \abs{E_{1}}+\abs{S_{1}}\cdot\abs{S_{1}} +\abs{V}\cdot\abs{S_{2}}} = \tilde O\paren{n^{1+\beta} + n^{2-2\beta} + n^{2-\beta-\gamma}} $$

Setting $\beta = \gamma = \frac{1}{3}$ yields the desired bound of $\tilde O\paren{n^{\frac{4}{3}}}$.

\end{proof}

\section{First step: $\paren{3,4W_{1}}$-emulator}\label{_3,4W1_}

We now present the second step ($k=3$) of our generalized framework: a $\paren{3,4W_{1}}$-emulator of size $\tilde O\paren{n^{\frac{5}{4}}}$. This construction preserves the additive guarantee of the $+4W_{1}$-emulator from \secref{+4W1} while increasing the multiplicative stretch from $1$ to $3$.
Increasing the multiplicative stretch is mandatory due to the $\tilde \Omega\paren{n^{\frac{4}{3}}}$ lower bound of Abboud and Bodwin \cite{AbbBod2015}. Notably, this construction introduces an additional set of auxiliary edges $B_{1}\paren{V}$, that plays a crucial role in the subsequent stages of our generalized framework.

\begin{algorithm}[H]
\caption{$\paren{3,4W_{1}}$-emulator$\paren{G=\paren{V,E,w}}$}\label{alg:_3,4W1_}

\textbf{Input:} An undirected non-negative weighted graph $G=\paren{V,E,w}$.

\textbf{Output:} A $\paren{3,4W_{1}}$-emulator. 

\medskip

Compute APSP

Let $\beta_1,\beta_2,\beta_3 \in \paren{ 0 , 1}$ to be fixed later

Set $S_0 \leftarrow V$

\For{$i \in \bracke{3}$}
{
    Sample each vertex of $S_{i-1}$ with probability $n^{-\beta_i}$ to $S_{i}$
    
    For every  $u\in V$ find the pivot $p_{i}\paren{u}$ 

    Construct the set of edges $D_{i}$
}

Set $S_{4} \leftarrow \varnothing$

$D\leftarrow \mcup{D_{i}}{i=1}{3}$

Construct the set of edges $E_{1}$

Construct the set of edges $B_{1}\paren{V}$

Set $F\leftarrow D \cup E_{1} \cup \paren{S_{1} \times S_{2} }  \cup B_{1}\paren{V}$

Initialize $w_{H}\leftarrow \varnothing$

\For{$\paren{u,v}\in F$}
    {
        $w_{H} \paren{u,v} \leftarrow \delta_{G} \paren{u,v}$
    }

\Return $H=\paren{V,F,w_{H}}$

\end{algorithm}

\algref{_3,4W1_} uses three parameters $\beta_1,\beta_2,\beta_3 \in \paren{0,1}$ such that a vertex from $S_{i-1}$ has probability $n^{-\beta_i}$ to be sampled to $S_{i}$, for $i\in\bracke{3}$.  For completeness, we set $S_{0} = V$ and $S_{4}=\varnothing$, hence $E_{4}=E$. After sampling, we identify the pivots $p_{1}\paren{u}$, $p_{2}\paren{u}$ and $p_{3}\paren{u}$ for every $u\in V$, forming the auxiliary edge sets $D$ and $E_{1}$. Additionally, we  construct the set of edges $B_{1}\paren{V}$ from vertices to members of their bunch (See: \secref{hierarchy}).

We set $F$ to be $D\cup E_{1}\cup \paren{S_{1}\times S_{2}}\cup B_{1}\paren{V}\cup\paren{V\times S_{3}}$ and the weight $w_{H}\paren{u,v}$ to be $\delta_{G}\paren{u,v}$ for $\paren{u,v}\in F\setminus E$.

We note that, in this special case, the stretch analysis only uses those edges of $B_{1}\paren{V}$ arising from refined balls of the form $\ball{u}{1}{2}$. We nevertheless retain the notation $B_{1}\paren{V}$ for consistency with the general construction, where multiple refined ball components are required. On a similar note, the product $V\times S_{3}$ is not needed for the stretch analysis, but is included because it is mandatory in the runtime efficient implementation as a fallback when the relevant shortest path is not contained in $E_{3}$ (See: \appref{runtime_k=4}).

This concludes  \algref{_3,4W1_}. To argue that $H=\paren{V,F,w_{H}}$ is a $\paren{3,4W_{1}}$-emulator we consider arbitrary $u,v\in V$ and fix a $u\squiggly v$. Consistent with the framework established in \secref{overview}, we distinguish between three cases:

\begin{itemize}
    \item \case{a}: $\sizeindex{u\squiggly v}{1} = 0 $,
    \item \case{b}: $\paren{x,y} = \rightabovearrow{u\squiggly v}{1} = \leftabovearrow{u\squiggly v}{1}$,
    \item \case{c}: $\paren{x,y} = \rightabovearrow{u\squiggly v}{1}  \neq \leftabovearrow{u\squiggly v}{1} = \paren{z,w}$.
\end{itemize}

This case analysis dictates our proof structure. \lemref{_3,4W1__case_a}, \lemref{_3,4W1__case_b} and \lemref{_3,4W1__case_c} establish that $\delta_{H}\paren{u,v} \leq 3\cdot \delta_{G}\paren{u,v} + 4W_{1}\paren{u\squiggly v}$ for \case{a}, \case{b} and \case{c}, respectively, after which we bound the overall size by $\abs{F} \in \tilde O\paren{n^{\frac{5}{4}}}$.

\begin{figure}[h]
\centering

\tikzset{every picture/.style={line width=0.75pt}} 

\begin{tikzpicture}[x=0.75pt,y=0.75pt,yscale=-1,xscale=1]

\draw [color={rgb, 255:red, 245; green, 166; blue, 35 }  ,draw opacity=1 ]   (331.15,146.14) .. controls (332.81,144.47) and (334.48,144.47) .. (336.15,146.13) .. controls (337.82,147.79) and (339.49,147.78) .. (341.15,146.11) .. controls (342.82,144.44) and (344.48,144.44) .. (346.15,146.1) .. controls (347.82,147.76) and (349.49,147.75) .. (351.15,146.08) .. controls (352.82,144.41) and (354.48,144.41) .. (356.15,146.07) .. controls (357.82,147.73) and (359.49,147.72) .. (361.15,146.05) .. controls (362.82,144.38) and (364.48,144.38) .. (366.15,146.04) .. controls (367.82,147.7) and (369.49,147.69) .. (371.15,146.02) .. controls (372.82,144.35) and (374.48,144.35) .. (376.15,146.01) .. controls (377.82,147.67) and (379.49,147.66) .. (381.15,145.99) .. controls (382.82,144.32) and (384.48,144.32) .. (386.15,145.98) .. controls (387.82,147.64) and (389.49,147.63) .. (391.15,145.96) .. controls (392.82,144.29) and (394.48,144.29) .. (396.15,145.95) .. controls (397.82,147.61) and (399.49,147.6) .. (401.15,145.93) .. controls (402.82,144.26) and (404.48,144.26) .. (406.15,145.92) .. controls (407.82,147.58) and (409.49,147.57) .. (411.15,145.9) .. controls (412.82,144.23) and (414.48,144.23) .. (416.15,145.89) .. controls (417.82,147.55) and (419.49,147.54) .. (421.15,145.87) -- (423.32,145.87) -- (423.32,145.87) ;
\draw  [color={rgb, 255:red, 0; green, 0; blue, 0 }  ,draw opacity=1 ][fill={rgb, 255:red, 0; green, 0; blue, 0 }  ,fill opacity=1 ] (185.9,146.03) .. controls (185.9,143.82) and (187.69,142.03) .. (189.9,142.03) .. controls (192.11,142.03) and (193.9,143.82) .. (193.9,146.03) .. controls (193.9,148.24) and (192.11,150.03) .. (189.9,150.03) .. controls (187.69,150.03) and (185.9,148.24) .. (185.9,146.03) -- cycle ;
\draw  [color={rgb, 255:red, 124; green, 180; blue, 60 }  ,draw opacity=1 ] (195,128.25) -- (277.5,128.25) -- (277.5,134.42) ;
\draw  [color={rgb, 255:red, 124; green, 180; blue, 60 }  ,draw opacity=1 ] (277.5,128.25) -- (195,128.25) -- (195,134.42) ;

\draw  [color={rgb, 255:red, 0; green, 0; blue, 0 }  ,draw opacity=1 ][fill={rgb, 255:red, 0; green, 0; blue, 0 }  ,fill opacity=1 ] (277.5,146.17) .. controls (277.5,143.96) and (279.29,142.17) .. (281.5,142.17) .. controls (283.71,142.17) and (285.5,143.96) .. (285.5,146.17) .. controls (285.5,148.38) and (283.71,150.17) .. (281.5,150.17) .. controls (279.29,150.17) and (277.5,148.38) .. (277.5,146.17) -- cycle ;
\draw [color={rgb, 255:red, 0; green, 0; blue, 0 }  ,draw opacity=1 ]   (189.9,146.14) .. controls (191.56,144.47) and (193.23,144.47) .. (194.9,146.13) .. controls (196.57,147.79) and (198.24,147.78) .. (199.9,146.11) .. controls (201.57,144.44) and (203.23,144.44) .. (204.9,146.1) .. controls (206.57,147.76) and (208.24,147.75) .. (209.9,146.08) .. controls (211.57,144.41) and (213.23,144.41) .. (214.9,146.07) .. controls (216.57,147.73) and (218.24,147.72) .. (219.9,146.05) .. controls (221.57,144.38) and (223.23,144.38) .. (224.9,146.04) .. controls (226.57,147.7) and (228.24,147.69) .. (229.9,146.02) .. controls (231.57,144.35) and (233.23,144.35) .. (234.9,146.01) .. controls (236.57,147.67) and (238.24,147.66) .. (239.9,145.99) .. controls (241.57,144.32) and (243.23,144.32) .. (244.9,145.98) .. controls (246.57,147.64) and (248.24,147.63) .. (249.9,145.96) .. controls (251.57,144.29) and (253.23,144.29) .. (254.9,145.95) .. controls (256.57,147.61) and (258.24,147.6) .. (259.9,145.93) .. controls (261.57,144.26) and (263.23,144.26) .. (264.9,145.92) .. controls (266.57,147.58) and (268.24,147.57) .. (269.9,145.9) .. controls (271.57,144.23) and (273.23,144.23) .. (274.9,145.89) .. controls (276.57,147.55) and (278.24,147.54) .. (279.9,145.87) -- (282.07,145.87) -- (282.07,145.87) ;
\draw  [color={rgb, 255:red, 0; green, 0; blue, 0 }  ,draw opacity=1 ][fill={rgb, 255:red, 0; green, 0; blue, 0 }  ,fill opacity=1 ] (327.15,146.03) .. controls (327.15,143.82) and (328.94,142.03) .. (331.15,142.03) .. controls (333.36,142.03) and (335.15,143.82) .. (335.15,146.03) .. controls (335.15,148.24) and (333.36,150.03) .. (331.15,150.03) .. controls (328.94,150.03) and (327.15,148.24) .. (327.15,146.03) -- cycle ;
\draw  [color={rgb, 255:red, 0; green, 0; blue, 0 }  ,draw opacity=1 ][fill={rgb, 255:red, 0; green, 0; blue, 0 }  ,fill opacity=1 ] (418.75,146.17) .. controls (418.75,143.96) and (420.54,142.17) .. (422.75,142.17) .. controls (424.96,142.17) and (426.75,143.96) .. (426.75,146.17) .. controls (426.75,148.38) and (424.96,150.17) .. (422.75,150.17) .. controls (420.54,150.17) and (418.75,148.38) .. (418.75,146.17) -- cycle ;

\draw (183.98,133.03) node [anchor=north west][inner sep=0.75pt]  [font=\footnotesize,color={rgb, 255:red, 0; green, 0; blue, 0 }  ,opacity=1 ,xslant=-0.03] [align=left] {$\displaystyle u$};
\draw (277.25,132.5) node [anchor=north west][inner sep=0.75pt]  [font=\footnotesize,xslant=-0.03] [align=left] {$\displaystyle v$};
\draw (230.5,113.75) node [anchor=north west][inner sep=0.75pt]  [font=\scriptsize] [align=left] {\textcolor[rgb]{0.18,0.5,0.89}{$\displaystyle \textcolor[rgb]{0.49,0.71,0.24}{E}\textcolor[rgb]{0.49,0.71,0.24}{_{1}}$}};
\draw (325.23,133.03) node [anchor=north west][inner sep=0.75pt]  [font=\footnotesize,color={rgb, 255:red, 0; green, 0; blue, 0 }  ,opacity=1 ,xslant=-0.03] [align=left] {$\displaystyle u$};
\draw (418.5,132.5) node [anchor=north west][inner sep=0.75pt]  [font=\footnotesize,xslant=-0.03] [align=left] {$\displaystyle v$};

\end{tikzpicture}

\caption[.]{\label{fig:_3,4W1__case_a} To the left: \case{a} of \algref{_3,4W1_}. To the right: the path that will be in  $H$.}
\vspace{1mm}
\end{figure}
\begin{lemma}~\label{lem:_3,4W1__case_a} 
   If \case{a} holds then $\delta_H\paren{u,v} = \delta_G\paren{u,v}$.
\end{lemma}
\begin{proof} 
Recall that $F$ includes $E_1$. As $\sizeindex{u\squiggly v}{1} =0$, it follows that $u\squiggly v\subseteq E_{1} \subseteq F$ which means that $\delta_H\paren{u,v} = \delta_G\paren{u,v}$. See: \figref{_3,4W1__case_a}.
\end{proof}

\input{\emulatorsImagesPath _3,4W1_/case_b_generic}

We move to \case{b}, which requires bridging only a single missing edge. Contrary to the flexibility of APASP algorithms discussed in \secref{overview}, we  already incur a substantial additive penalty in the emulator setting.

\begin{lemma}~\label{lem:_3,4W1__case_b} 
   If \case{b} holds then $\delta_H\paren{u,v} \leq \delta_G\paren{u,v} + 6W_{1}\paren{u\squiggly v}$.
\end{lemma}
\begin{proof} 
Let $\paren{x,y} = \rightabovearrow{u\squiggly v}{1}=\leftabovearrow{u\squiggly v}{1}$. By definition of $E_{1}$, it follows that $\delta_G\paren{x,p_{1}\paren{x}},\delta_G\paren{y,p_{1}\paren{y}} \leq w\paren{x,y}$. We consider two possible cases:
\begin{enumerate}
    \item \dashuline{\case{b}{1}: $p_{1}\paren{x}\in \ball{y}{1}{2}$}\\
    By definition of $B_{1}\paren{y}$, it follows that $\paren{p_{1}\paren{x},y}\in B_{1}\paren{y} \subseteq B_{1}\paren{V} \subseteq F$. Hence:
    \begin{alignat*}{3}
        \delta_H\paren{u,v} & \leq \textubrace{ \delta_G\paren{u,x} +  \delta_G\paren{x,p_{1}\paren{x}}+\delta_G\paren{p_{1}\paren{x},y}+\delta_G\paren{y,v}}{Triangle inequality} \\
        &
        \leq \textubrace{ \delta_G\paren{u,x}+ 2\delta_G\paren{x,p_{1}\paren{x}} +  w\paren{x,y}+\delta_G\paren{y,v}}{Triangle inequality}  \\
        &
        \leq \textubrace{ \delta_G\paren{u,v}+ 2\delta_G\paren{x,p_{1}\paren{x}} }{$u\squiggly x \straight y \squiggly v$ is a shortest path}  \\
        & \leq \textubrace{ \delta_G\paren{u,v}+ 2w\paren{x,y}  }{$\paren{x,y}\notin E_{1}$}\\
        & \leq \textubrace{ \delta_G\paren{u,v } + 2 \cdot W_{1}\paren{u\squiggly v}}{$w\paren{x,y} \leq W_{1}\paren{u\squiggly v}$ } &&
    \end{alignat*}

    \item \dashuline{\case{b}{2}: $p_{1}\paren{x}\notin \ball{y}{1}{2}$}\\
    It follows that $\delta_G\paren{y,p_2\paren{y}}\leq \delta_G \paren{y,p_{1}\paren{x}} \leq  w\paren{x,y}+\delta_G\paren{x,p_{1}\paren{x}}  \leq 2 w \paren{x,y}$. Recall that $S_{1}\times S_{2}\subseteq F$. Hence:
    \begin{alignat*}{3}
        \delta_H\paren{u,v} & \leq \textubrace{ \delta_G\paren{u,x} +  \delta_G\paren{x,p_{1}\paren{x}}+\delta_G\paren{p_{1}\paren{x},p_{2}\paren{y}}+\delta_G\paren{p_{2}\paren{y},y}+\delta_G\paren{y,v}}{Triangle inequality} \\
        &
        \leq \textubrace{ \delta_G\paren{u,x}+ 2\delta_G\paren{x,p_{1}\paren{x}} +  w\paren{x,y}+2\delta_G\paren{y,p_{2}\paren{y}}+\delta_G\paren{y,v}}{Triangle inequality}  \\
        &
        \leq \textubrace{ \delta_G\paren{u,v}+ 2\delta_G\paren{x,p_{1}\paren{x}} + 2\delta_G\paren{y,p_{2}\paren{y}} }{$u\squiggly x \straight y \squiggly v$ is a shortest path}  \\
        &
        \leq \textubrace{ \delta_G\paren{u,v}+ 2w\paren{x,y}+4w\paren{x,y} }{By our preceding discussion}\\
        & = \textubrace{ \delta_G\paren{u,v}+ 6w\paren{x,y}  }{Simply summing}\\
        & \leq \textubrace{ \delta_G\paren{u,v } + 6 \cdot W_{1}\paren{u\squiggly v}}{$w\paren{x,y} \leq W_{1}\paren{u\squiggly v}$ } &&
    \end{alignat*}
    
    See: \figref{_3,4W1__case_b_generic}.
\end{enumerate}

\end{proof}

\input{\emulatorsImagesPath _3,4W1_/case_c_generic}

Finally, we address \case{c}, where we have at least two edges to \qoute{bridge over}. 
Utilizing $B_{1}\paren{V}$ to bound distances to third level pivots is inefficient, hence we use $B_{2}\paren{S_{1}}$. While we get, unavoidably, a multiplicative stretch, we have two edges that can be used as an upper bound, hence we can use $W_{2}$ instead of $W_{1}$. 

\begin{lemma}~\label{lem:_3,4W1__case_c} 
   If \case{c} holds then $\delta_H\paren{u,v} \leq 3\delta_G\paren{u,v} + 4W_{2}\paren{u\squiggly v}$.
\end{lemma}
\begin{proof} 
Let $\paren{x,y} = \rightabovearrow{u\squiggly v}{1}$ and $\paren{z,w}=\leftabovearrow{u\squiggly v}{1}$. By definition of $E_{1}$, it follows that $\delta_G\paren{x,p_{1}\paren{x}}\leq w\paren{x,y}$ and $\delta_G\paren{w,p_{1}\paren{w}} \leq w\paren{z,w}$. As $\paren{x,y}\neq \paren{z,w}$, we may assume w.l.o.g. that $w\paren{x,y}\leq w\paren{z,w}$. We consider two possible cases:
\begin{enumerate}
    \item \dashuline{\case{c}{1}: $p_{1}\paren{x}\in \ball{w}{1}{2}$ or $p_{1}\paren{w}\in \ball{x}{1}{2}$}\\
    If $p_{1}\paren{x}\in \ball{w}{1}{2}$ by the definition of $B_{1}\paren{w}$, it follows that $\paren{p_{1}\paren{x},w}\in B_{1}\paren{w} \subseteq B_{1}\paren{V} \subseteq F$. Therefore:
    \begin{alignat*}{3}
        \delta_H\paren{u,v} & \leq \textubrace{ \delta_G\paren{u,x} +  \delta_G\paren{x,p_{1}\paren{x}}+\delta_G\paren{p_{1}\paren{x},w}+\delta_G\paren{w,v}}{Triangle inequality} \\
        &
        \leq \textubrace{ \delta_G\paren{u,x}+ 2\delta_G\paren{x,p_{1}\paren{x}} +  \delta_G\paren{x,w}+\delta_G\paren{w,v}}{Triangle inequality}  \\
        &
        \leq \textubrace{ \delta_G\paren{u,v}+ 2\delta_G\paren{x,p_{1}\paren{x}} }{$u\squiggly x \squiggly w \squiggly v$ is a shortest path}  \\
        & \leq \textubrace{ \delta_G\paren{u,v}+ 2w\paren{x,y}  }{$\paren{x,y}\notin E_{1}$}\\
        & \leq \textubrace{ \delta_G\paren{u,v } + 2 \cdot W_{2}\paren{u\squiggly v}}{$w\paren{x,y} \leq W_{2}\paren{u\squiggly v}$ } &&
    \end{alignat*}
    Otherwise, $p_{1}\paren{w}\in \ball{x}{1}{2}$. By a symmetric argument $\delta_{H}\paren{u,v}\leq \delta_{G}\paren{u,v}+2w\paren{z,w} \leq \delta_{G}\paren{u,v}+2W_{1}\paren{u\squiggly v}$.

    \item \dashuline{\case{c}{2}: $p_{1}\paren{x}\notin \ball{w}{1}{2}$ and $p_{1}\paren{w}\notin \ball{x}{1}{2}$}\\
    It follows that $\delta_G\paren{w,p_2\paren{w}}\leq \delta_G \paren{w,p_{1}\paren{x}} \leq  \delta_G\paren{x,w}+\delta_G\paren{x,p_{1}\paren{x}}  \leq \delta_G\paren{x,w}+ w \paren{x,y}$. Recall that $S_{1}\times S_{2}\subseteq F$. Therefore:
    \begin{alignat*}{3}
        \delta_H\paren{u,v} & \leq \textubrace{ \delta_G\paren{u,x} +  \delta_G\paren{x,p_{1}\paren{x}}+\delta_G\paren{p_{1}\paren{x},p_{2}\paren{w}}+\delta_G\paren{p_{2}\paren{w},w}+\delta_G\paren{w,v}}{Triangle inequality} \\
        &
        \leq \textubrace{ \delta_G\paren{u,x}+ 2\delta_G\paren{x,p_{1}\paren{x}} +  \delta_G\paren{x,w}+2\delta_G\paren{w,p_{2}\paren{w}}+\delta_G\paren{w,v}}{Triangle inequality}  \\
        &
        \leq \textubrace{ \delta_G\paren{u,v}+ 2\delta_G\paren{x,p_{1}\paren{x}} + 2\delta_G\paren{w,p_{2}\paren{w}} }{$u\squiggly x \squiggly w \squiggly v$ is a shortest path}  \\
        &
        \leq \textubrace{ \delta_G\paren{u,v}+ 2w\paren{x,y}+2\delta_G\paren{x,w}+2w\paren{x,y}}{By our preceding discussion}\\
        & \leq \textubrace{ 3\delta_G\paren{u,v}+ 4w\paren{x,y}  }{$u\squiggly x \squiggly w \squiggly v$ is a shortest path}\\
        & \leq \textubrace{ 3\delta_G\paren{u,v } + 4 \cdot W_{2}\paren{u\squiggly v}}{Our w.l.o.g. assumption } &&
    \end{alignat*}
    
    See: \figref{_3,4W1__case_c_generic}.
\end{enumerate}
\end{proof}

Having proven the three cases, we conclude with the size of the emulator: 

\begin{theorem}~\label{thm:_3,4W1_} 
   \algref{_3,4W1_} constructs a $\paren{3,4W_{1}}$-emulator of expected size $\tilde O\paren{n^\frac{5}{4}}$.
\end{theorem}
\begin{proof}

The upper bound of \lemref{_3,4W1__case_a} is trivial. We consider the upper bounds of \lemref{_3,4W1__case_b} and \lemref{_3,4W1__case_c}, which are $+6W_{1}$ and $\paren{3,4W_{2}}$, respectively. Both of these are smaller than $\paren{3,4W_{1}}$. Hence, we conclude that for any $u,v\in V$ and $u\squiggly v$ it holds that: $\delta_H\paren{u,v}\leq 3\delta_G \paren{u,v}+4W_{1}\paren{u\squiggly v}$. 

We now analyze the expected size of the emulator. By \obvref{hssize}, it follows that $\abs{S_1} = \tilde O \paren{n^{1-\beta_1}}$, $\abs{S_2} = \tilde O \paren{n^{1-\beta_1-\beta_2}}$ and $\abs{S_3} = \tilde O \paren{n^{1-\beta_1-\beta_2-\beta_3}}$. Hence, the product $S_{1}\times S_{2}$  contributes $\tilde O \paren{n^{2-2\beta_1-\beta_2}}$ 
 edges; the product $V\times S_{3}$ contributes $\tilde O\paren{n^{2-\beta_{1}-\beta_{2}-\beta{3}}} $ edges; the set $B_{1}\paren{V}$   has at most $\tilde O\paren{n^{1+\beta_1}+n^{1+\beta_2}+n^{1+\beta_3}} $ edges due to \obvref{edgestobunches2}; the set $ D$ contributes  $\tilde O\paren{n}$ edges due to \obvref{edgesToPivots};  lastly,  $\abs{E_{1}} = \tilde O \paren{n^{1+\beta_1}}$ by \obvref{edgessize}. The total size of $\abs{F}$ is therefore:

$$ \tilde O \paren{\abs{D} + \abs{E_{1}}+\abs{S_{1}}\cdot\abs{S_{2}} +\abs{V}\cdot\abs{S_{3}}  +\abs{B_{1}\paren{V}}} = \tilde O\paren{n^{1+\beta_1} + n^{2-2\beta_1-\beta_2} +n^{2-\beta_{1}-\beta_{2}-\beta_{3}} + n^{1+\beta_2} + n^{1+\beta_3}} $$

By selecting $\beta_1=\beta_2=\beta_3=\frac{1}{4}$, the size becomes $\tilde O\paren{n^\frac{5}{4}}$. 
\end{proof}

\section{$\paren{2\cdot \floor{\frac{k}{2}} -1, 2\cdot \ceil{\frac{k}{2}}\cdot W_{1} +\mmax{0,2\cdot\paren{\ceil{\frac{k}{2}}-2}}\cdot W_{2}}$-emulators}\label{_k-1,kW1+_k-4_W2_}

In this section we present a $(k-1,k\cdot W_{1}+\paren{k-4}\cdot W_{2})$-emulator for even values $k$, and a $(k-2,\paren{k+1}\cdot W_{1} + \paren{k-3}\cdot W_{2})$-emulator for odd values of $k$. The size of the resulting emulator is $\tilde O \paren{n^{1+\frac{1}{k}}}$. When $k=3$ this results in the $+4W_{1}$-emulator with $\tilde O \paren{n^{\frac{4}{3}}}$ edges from \secref{+4W1}, and when $k=4$ this is the $\paren{3,4W_{1}}$-emulator with $\tilde O \paren{n^{\frac{5}{4}}}$ edges from \secref{_3,4W1_}.

\algref{_k-1,kW1+_k-4_W2_} introduces $k-1$ parameters $\beta_{1},\beta_{2},\dots,\beta_{k-1} \in \paren{0,1}$ such that a vertex from $S_{i-1}$ has probability $n^{-\beta_{i}}$ to be sampled into $S_{i}$, for $i\in\bracke{k-1}$. For completeness of our hierarchical structure, we define $S_{0} = V$ and $S_{k}=\varnothing$. After sampling these sets, we identify the $k-1$ pivots $p_{i}\paren{u}$ for each vertex $u\in V$ and $i \in \bracke{k-1}$. We then construct the auxiliary edge set $D_{i}$ connecting each vertex $u\in V$ to its respective pivot $p_{i}\paren{u}$, and set $D = \mcup{D_{i}}{i=1}{k-1}$. Utilizing the distances $\delta\paren{u,p_{1}\paren{u}}$ for any $u\in V$, we compute the edge set $E_{1}\paren{u}$ and then $E_{1}$. Furthermore, extending the idea presented in \algref{_3,4W1_}, we construct the edge sets $B_{1}\paren{V}$ and $B_{2}\paren{S_{1}}$ (See: \secref{hierarchy}).

We set $F$ to be the union of the edge sets: $D$, $E_1$, $B_{1}\paren{V}$, $B_{2}\paren{S_{1}}$. In addition, we add to $F$ any pair $S_{i-1} \times S_{k-i}$ for $i\in \bracke{k}$. Finally, we set the weight $w_{H}$ of each edge $\paren{u,v} \in F$ to equal its exact distance $\delta_G \paren{u,v}$ in the original graph $G$. This concludes the construction of \algref{_k-1,kW1+_k-4_W2_}. We now show that $H=\paren{V,F,w_{H}}$ is a $\paren{2\cdot \floor{\frac{k}{2}} -1, 2\cdot \ceil{\frac{k}{2}}\cdot W_{1} +\mmax{0,2\cdot\paren{\ceil{\frac{k}{2}}-2}}\cdot W_{2}}$-emulator.

Let $u,v\in V$ and fix a shortest path $u\squiggly v$. To establish the stretch bounds, we proceed with a  case analysis that considers different values of $\sizeindex{u\squiggly v}{1}, \rightabovearrow{u\squiggly v}{1}$ and $\leftabovearrow{u\squiggly v}{1}$, similarly to the case analysis presented in  \algref{+4W1} and \algref{_3,4W1_}. For each case, we analyze the even and odd values of $k$ separately.

\begin{algorithm}[H]
\caption{generalEmulator$\paren{G=\paren{V,E,w},k}$}\label{alg:_k-1,kW1+_k-4_W2_}

\textbf{Input:} An undirected non-negative weighted graph $G=\paren{V,E,w}$ and an integer $k\geq 2$.

\textbf{Output:} A $(2\cdot \floor{\frac{k}{2}} -1, 2\cdot \ceil{\frac{k}{2}}\cdot W_{1} +\max\{0,2\cdot(\ceil{\frac{k}{2}}-2)\}\cdot W_{2})$-emulator. 

\medskip

Compute APSP

Let $\beta_{1},\beta_{2},\ldots, \beta_{k-1} \in \paren{ 0 , 1}$ to be fixed later

Set $S_0 \leftarrow V$

\For{$i \in \bracke{k-1}$}
{
    Sample each vertex of $S_{i-1}$ with probability $n^{-\beta_i}$ to $S_{i}$
    
    For every  $u\in V$ find the pivot $p_{i}\paren{u}$ 

    Construct the set of edges $D_{i}$
}

Set $S_{k} \leftarrow \varnothing$

$D\leftarrow \mcup{D_{i}}{i=1}{k-1}$

Construct the set of edges $E_{1}$

Construct the set of edges $B_{1}\paren{V}$

Construct the set of edges $B_{2}\paren{S_{1}}$

Set $F\leftarrow D \cup E_{1}  \cup B_{1}\paren{V} \cup B_{2}\paren{S_{1}}$

\For{$i \in \bracke{k}$}
{
    $F\leftarrow F \cup \paren{S_{i-1}\times S_{k-i}}$
}

Initialize $w_{H}\leftarrow \varnothing$

\For{$\paren{u,v}\in F$}
    {
        $w_{H} \paren{u,v} \leftarrow \delta_{G} \paren{u,v}$
    }

\Return $H=\paren{V,F,w_{H}}$

\end{algorithm}

\begin{itemize}
    \item \case{a}: $\sizeindex{u\squiggly v}{1} = 0 $,
    \item \case{b}: $\paren{x,y} = \rightabovearrow{u\squiggly v}{1} = \leftabovearrow{u\squiggly v}{1}$,
    \item \case{c}: $\paren{x,y} = \rightabovearrow{u\squiggly v}{1}  \neq \leftabovearrow{u\squiggly v}{1} = \paren{z,w}$.
\end{itemize}

We show that distinguishing between these three cases is sufficient to establish the stretch of the emulator.  We prove the desired stretch in \lemref{_k-1,kW1+_k-4_W2__case_a} to \lemref{_k-1,kW1+_k-4_W2__case_c} for cases $\case{a}$ to $\case{c}$, respectively. We later show that the emulator satisfies $\abs{F} \in \tilde O\paren{n^{1+\frac{1}{k}}}$. We now begin by considering each case individually:

\begin{figure}[h]
\centering

\tikzset{every picture/.style={line width=0.75pt}} 

\begin{tikzpicture}[x=0.75pt,y=0.75pt,yscale=-1,xscale=1]

\draw [color={rgb, 255:red, 245; green, 166; blue, 35 }  ,draw opacity=1 ]   (331.15,146.14) .. controls (332.81,144.47) and (334.48,144.47) .. (336.15,146.13) .. controls (337.82,147.79) and (339.49,147.78) .. (341.15,146.11) .. controls (342.82,144.44) and (344.48,144.44) .. (346.15,146.1) .. controls (347.82,147.76) and (349.49,147.75) .. (351.15,146.08) .. controls (352.82,144.41) and (354.48,144.41) .. (356.15,146.07) .. controls (357.82,147.73) and (359.49,147.72) .. (361.15,146.05) .. controls (362.82,144.38) and (364.48,144.38) .. (366.15,146.04) .. controls (367.82,147.7) and (369.49,147.69) .. (371.15,146.02) .. controls (372.82,144.35) and (374.48,144.35) .. (376.15,146.01) .. controls (377.82,147.67) and (379.49,147.66) .. (381.15,145.99) .. controls (382.82,144.32) and (384.48,144.32) .. (386.15,145.98) .. controls (387.82,147.64) and (389.49,147.63) .. (391.15,145.96) .. controls (392.82,144.29) and (394.48,144.29) .. (396.15,145.95) .. controls (397.82,147.61) and (399.49,147.6) .. (401.15,145.93) .. controls (402.82,144.26) and (404.48,144.26) .. (406.15,145.92) .. controls (407.82,147.58) and (409.49,147.57) .. (411.15,145.9) .. controls (412.82,144.23) and (414.48,144.23) .. (416.15,145.89) .. controls (417.82,147.55) and (419.49,147.54) .. (421.15,145.87) -- (423.32,145.87) -- (423.32,145.87) ;
\draw  [color={rgb, 255:red, 0; green, 0; blue, 0 }  ,draw opacity=1 ][fill={rgb, 255:red, 0; green, 0; blue, 0 }  ,fill opacity=1 ] (185.9,146.03) .. controls (185.9,143.82) and (187.69,142.03) .. (189.9,142.03) .. controls (192.11,142.03) and (193.9,143.82) .. (193.9,146.03) .. controls (193.9,148.24) and (192.11,150.03) .. (189.9,150.03) .. controls (187.69,150.03) and (185.9,148.24) .. (185.9,146.03) -- cycle ;
\draw  [color={rgb, 255:red, 124; green, 180; blue, 60 }  ,draw opacity=1 ] (195,128.25) -- (277.5,128.25) -- (277.5,134.42) ;
\draw  [color={rgb, 255:red, 124; green, 180; blue, 60 }  ,draw opacity=1 ] (277.5,128.25) -- (195,128.25) -- (195,134.42) ;

\draw  [color={rgb, 255:red, 0; green, 0; blue, 0 }  ,draw opacity=1 ][fill={rgb, 255:red, 0; green, 0; blue, 0 }  ,fill opacity=1 ] (277.5,146.17) .. controls (277.5,143.96) and (279.29,142.17) .. (281.5,142.17) .. controls (283.71,142.17) and (285.5,143.96) .. (285.5,146.17) .. controls (285.5,148.38) and (283.71,150.17) .. (281.5,150.17) .. controls (279.29,150.17) and (277.5,148.38) .. (277.5,146.17) -- cycle ;
\draw [color={rgb, 255:red, 0; green, 0; blue, 0 }  ,draw opacity=1 ]   (189.9,146.14) .. controls (191.56,144.47) and (193.23,144.47) .. (194.9,146.13) .. controls (196.57,147.79) and (198.24,147.78) .. (199.9,146.11) .. controls (201.57,144.44) and (203.23,144.44) .. (204.9,146.1) .. controls (206.57,147.76) and (208.24,147.75) .. (209.9,146.08) .. controls (211.57,144.41) and (213.23,144.41) .. (214.9,146.07) .. controls (216.57,147.73) and (218.24,147.72) .. (219.9,146.05) .. controls (221.57,144.38) and (223.23,144.38) .. (224.9,146.04) .. controls (226.57,147.7) and (228.24,147.69) .. (229.9,146.02) .. controls (231.57,144.35) and (233.23,144.35) .. (234.9,146.01) .. controls (236.57,147.67) and (238.24,147.66) .. (239.9,145.99) .. controls (241.57,144.32) and (243.23,144.32) .. (244.9,145.98) .. controls (246.57,147.64) and (248.24,147.63) .. (249.9,145.96) .. controls (251.57,144.29) and (253.23,144.29) .. (254.9,145.95) .. controls (256.57,147.61) and (258.24,147.6) .. (259.9,145.93) .. controls (261.57,144.26) and (263.23,144.26) .. (264.9,145.92) .. controls (266.57,147.58) and (268.24,147.57) .. (269.9,145.9) .. controls (271.57,144.23) and (273.23,144.23) .. (274.9,145.89) .. controls (276.57,147.55) and (278.24,147.54) .. (279.9,145.87) -- (282.07,145.87) -- (282.07,145.87) ;
\draw  [color={rgb, 255:red, 0; green, 0; blue, 0 }  ,draw opacity=1 ][fill={rgb, 255:red, 0; green, 0; blue, 0 }  ,fill opacity=1 ] (327.15,146.03) .. controls (327.15,143.82) and (328.94,142.03) .. (331.15,142.03) .. controls (333.36,142.03) and (335.15,143.82) .. (335.15,146.03) .. controls (335.15,148.24) and (333.36,150.03) .. (331.15,150.03) .. controls (328.94,150.03) and (327.15,148.24) .. (327.15,146.03) -- cycle ;
\draw  [color={rgb, 255:red, 0; green, 0; blue, 0 }  ,draw opacity=1 ][fill={rgb, 255:red, 0; green, 0; blue, 0 }  ,fill opacity=1 ] (418.75,146.17) .. controls (418.75,143.96) and (420.54,142.17) .. (422.75,142.17) .. controls (424.96,142.17) and (426.75,143.96) .. (426.75,146.17) .. controls (426.75,148.38) and (424.96,150.17) .. (422.75,150.17) .. controls (420.54,150.17) and (418.75,148.38) .. (418.75,146.17) -- cycle ;

\draw (183.98,133.03) node [anchor=north west][inner sep=0.75pt]  [font=\footnotesize,color={rgb, 255:red, 0; green, 0; blue, 0 }  ,opacity=1 ,xslant=-0.03] [align=left] {$\displaystyle u$};
\draw (277.25,132.5) node [anchor=north west][inner sep=0.75pt]  [font=\footnotesize,xslant=-0.03] [align=left] {$\displaystyle v$};
\draw (230.5,113.75) node [anchor=north west][inner sep=0.75pt]  [font=\scriptsize] [align=left] {\textcolor[rgb]{0.18,0.5,0.89}{$\displaystyle \textcolor[rgb]{0.49,0.71,0.24}{E}\textcolor[rgb]{0.49,0.71,0.24}{_{1}}$}};
\draw (325.23,133.03) node [anchor=north west][inner sep=0.75pt]  [font=\footnotesize,color={rgb, 255:red, 0; green, 0; blue, 0 }  ,opacity=1 ,xslant=-0.03] [align=left] {$\displaystyle u$};
\draw (418.5,132.5) node [anchor=north west][inner sep=0.75pt]  [font=\footnotesize,xslant=-0.03] [align=left] {$\displaystyle v$};

\end{tikzpicture}

\caption[.]{\label{fig:_k-1,kW1+_k-4_W2__case_a} To the left: \case{a} of \algref{_k-1,kW1+_k-4_W2_}. To the right: the path that will be in the emulator $H$.}
\vspace{1mm}
\end{figure}
\begin{lemma}~\label{lem:_k-1,kW1+_k-4_W2__case_a} 
   If \case{a} holds then $\delta_H\paren{u,v} = \delta_G\paren{u,v}$.
\end{lemma}
\begin{proof} 
Recall that $F$ includes $E_1$. As $\sizeindex{u\squiggly v}{1} =0$, it follows that $u\squiggly v\subseteq E_{1}$. Hence, $u\squiggly v \subseteq F$ which means that $\delta_H\paren{u,v} = \delta_G\paren{u,v}$. See: \figref{_k-1,kW1+_k-4_W2__case_a}.
\end{proof}

For \case{b}, we first establish a crucial structural property. While our previous arguments naturally bound $\delta_{G}\paren{x,p_{1}\paren{x}}$ and $\delta_{G}\paren{y,p_{1}\paren{y}}$, bounding higher level pivots $\delta_{G}\paren{x,p_{i}\paren{x}}$ and $\delta_{G}\paren{y,p_{i}\paren{y}}$ for $i \geq 2$ requires further care. We prove that as long as the pivots of $x$ and $y$ avoid the respective opposing bunches at all lower levels, the distance to the $i^{\textnormal{th}}$ pivot remains strictly bounded by $i \cdot w\paren{x,y}$. Specifically:

\begin{claim2}~\label{clm:_k-1,kW1+_k-4_W2__case_b_index}
Assume \case{b} holds. Let $i \leq k$ such that for every $j\in \bracke{i-1}$ it holds that $p_{j} \paren{x} \notin \ball{y}{j}{j+1}$ and $p_{j} \paren{y} \notin \ball{x}{j}{j+1}$, then $\delta_{G}\paren{x,p_{i}\paren{x}}, \delta_{G}\paren{y,p_{i}\paren{y}} \leq i \cdot w\paren{x,y}$.
\end{claim2}
\begin{proof} 

We prove this claim by induction on $i$. For the base case, assume $i=2$. As $\paren{x,y}\notin E_{1}$, it follows that $\delta_{G}\paren{x,p_{1}\paren{x}}, \delta_{G}\paren{y,p_{1}\paren{y}} \leq w\paren{x,y}$. By the assumption for $i=2$, it also holds that $p_{1} \paren{x} \notin \ball{y}{1}{2}$ and $p_{1} \paren{y} \notin \ball{x}{1}{2}$. Hence it follows that:

\begin{alignat*}{3}
        \delta_{G}\paren{x,p_{2}\paren{x}} & \leq \textubrace{ \delta_G\paren{x,p_{1}\paren{y}} }{$p_{1} \paren{y} \notin \ball{x}{1}{2}$} 
        &
        \leq \textubrace{ w\paren{x,y}+ \delta_G\paren{y,p_{1}\paren{y}}}{Triangle inequality}  
        &
        \leq \textubrace{ 2w\paren{x,y} }{$\delta_{G}\paren{y,p_{1}\paren{y}} \leq w\paren{x,y}$} 
\end{alignat*}

By a symmetric argument 
$ \delta_{G}\paren{y,p_{2}\paren{y}} \leq 2  w\paren{x,y}$.
Now, assume the claim holds inductively for $i-1$. Therefore, $ \delta_G\paren{x,p_{i-1}\paren{x}}, \delta_G\paren{y,p_{i-1}\paren{y}} \leq \paren{i-1} \cdot w\paren{x,y}$. Consequently:

\begin{alignat*}{3}
        \delta_{G}\paren{x,p_{i}\paren{x}} & \leq \textubrace{ \delta_G\paren{x,p_{i-1}\paren{y}} }{$p_{i-1} \paren{y} \notin \ball{x}{i-1}{i}$}
        &
        \leq \textubrace{ w\paren{x,y}+ \delta_G\paren{y,p_{i-1}\paren{y}}}{Triangle inequality}
        &
        \leq \textubrace{ i \cdot w\paren{x,y} }{Induction hypothesis}\end{alignat*}

By a symmetric argument utilizing the premise that $p_{i-1} \paren{x} \notin \ball{y}{i-1}{i}$, we obtain $\delta_{G}\paren{y,p_{i}\paren{y}} \leq i \cdot w\paren{x,y}$. This concludes the proof. 
\end{proof}

\input{\emulatorsImagesPath _k-1,kW1+_k-4_W2_/case_b}

Utilizing \clmref{_k-1,kW1+_k-4_W2__case_b_index}, we can now prove:

\begin{lemma}~\label{lem:_k-1,kW1+_k-4_W2__case_b} 
   If \case{b} holds then $\delta_H\paren{u,v} \leq \delta_G\paren{u,v} + 2\cdot\paren{k-1}\cdot W_{1}\paren{u\squiggly v}$.
\end{lemma}
\begin{proof} 
Let $\paren{x,y} = \rightabovearrow{u\squiggly v}{1}=\leftabovearrow{u\squiggly v}{1}$. By definition of $E_{1}$, it follows that $\delta_G\paren{x,p_{1}\paren{x}},\delta_G\paren{y,p_{1}\paren{y}} \leq w\paren{x,y}$. We consider two possible cases:
\begin{enumerate}
    \item \dashuline{\case{b}{1}: There exists an $i \leq \ceil{\frac{k}{2}}$ such that either $p_{i}\paren{x}\in \ball{y}{i}{i+1}$ or $p_{i}\paren{y}\in \ball{x}{i}{i+1}$}\\
    Let $i$ be the smallest such index and assume w.l.o.g. that $p_{i}\paren{x}\in \ball{y}{i}{i+1}$. As $i$ is the smallest such index, for every $j < i$ it holds that $p_j\paren{x} \notin \ball{y}{j}{j+1}$ and $p_j\paren{y} \notin \ball{x}{j}{j+1}$. By \clmref{_k-1,kW1+_k-4_W2__case_b_index}, it follows that $\delta_G\paren{x,p_{i}\paren{x}} \leq i \cdot w\paren{x,y}$. It should be observed that the auxiliary edge $\paren{p_{i}\paren{x},y}\in B_{1}\paren{y} \subseteq B_{1}\paren{V}$. Hence:
    
    \begin{alignat*}{3}
        \delta_{H}\paren{u,v} & \leq \textubrace{ \delta_{H}\paren{u,x} + w_{H}\paren{x,p_{i}\paren{x}} + w_{H}\paren{p_{i}\paren{x},y} + \delta_{H}\paren{y,v} }{$\paren{x,p_{i}\paren{x}}\in D \subseteq F$, $\paren{p_{i}\paren{x},y}\in B_{1}\paren{V} \subseteq F$}
        \\&
        = \textubrace{ \delta_{G}\paren{u,x} + \delta_{G}\paren{x,p_{i}\paren{x}} + \delta_{G}\paren{p_{i}\paren{x},y} + \delta_{G}\paren{y,v} }{Definition of $w_{H}$}\\
        &
        \leq \textubrace{ \delta_{G}\paren{u,x} + 2\delta_G\paren{x,p_{i}\paren{x}} +w\paren{x,y} + \delta_{G}\paren{y,v} }{Triangle inequality}
    \end{alignat*}\vspace{-1cm}
    \begin{alignat*}{3}
        \longspace
        &
        \leq \textubrace{ \delta_{G}\paren{u,v} + 2\delta_G\paren{x,p_{i}\paren{x}} }{$u\squiggly x \straight y \squiggly v$ is a shortest path} 
        &
        \leq \textubrace{ \delta_{G}\paren{u,v} + 2i \cdot w\paren{x,y} }{\clmref{_k-1,kW1+_k-4_W2__case_b_index}}\\
        &
        \leq \textubrace{ \delta_{G}\paren{u,v} + \paren{k+1} \cdot w\paren{x,y} }{$i\leq \ceil{\frac{k}{2}}$}
        &
        < \textubrace{ \delta_{G}\paren{u,v} + \paren{2\cdot\paren{k-1}} \cdot w\paren{x,y} }{$k\geq 3$}
    \end{alignat*}

    \item \dashuline{\case{b}{2}: For any $i\leq \ceil{\frac{k}{2}}$ both $p_{i}\paren{x}\notin \ball{y}{i}{i+1}$ and $p_{i}\paren{y}\notin \ball{x}{i}{i+1}$}\\
    As the above holds for any $i \leq \ceil{\frac{k}{2}}$, consider $i_{1}=\floor{\frac{k}{2}}-1$ and $i_{2}=\ceil{\frac{k}{2}}$. By \clmref{_k-1,kW1+_k-4_W2__case_b_index}, it holds that $\delta_{G}\paren{x,p_{\floor{\frac{k}{2}}-1}\paren{x}} \leq \paren{\floor{\frac{k}{2}} -1}\cdot w\paren{x,y}$ and $\delta_{G}\paren{y,p_{\ceil{\frac{k}{2}}}\paren{y}} \leq \ceil{\frac{k}{2}} \cdot w\paren{x,y}$.  Furthermore, observe that the auxiliary edge $\paren{p_{\floor{\frac{k}{2}}-1}\paren{x}, p_{\ceil{\frac{k}{2}}}\paren{y}} \in S_{\floor{\frac{k}{2}}-1} \times S_{\ceil{\frac{k}{2}}} \subseteq F$, as $\floor{\frac{k}{2}}-1 + \ceil{\frac{k}{2}} = k-1$. It hence follows that:
    
    \begin{alignat*}{3}
        \delta_{H}\paren{u,v} & \leq \textubrace{ \delta_{H}\paren{u,x} +  w_{H}\paren{x,p_{\floor{\frac{k}{2}-1}}\paren{x}}+w_{H}\paren{p_{\floor{\frac{k}{2}-1}}\paren{x},p_{\ceil{\frac{k}{2}}}\paren{y}}+w_{H}\paren{p_{\ceil{\frac{k}{2}}}\paren{y},y}+\delta_G\paren{y,v}}{Auxiliary edges in $F$} \\
        &
        = \textubrace{ \delta_{G}\paren{u,x} + \delta_{G}\paren{x,p_{\floor{\frac{k}{2}-1}}\paren{x}} + \delta_{G}\paren{p_{\floor{\frac{k}{2}-1}}\paren{x},p_{\ceil{\frac{k}{2}}}\paren{y}} + \delta_{G}\paren{p_{\ceil{\frac{k}{2}}}\paren{y},y} + \delta_{G}\paren{y,v} }{Definition of $w_{H}$} \\
        &
        \leq \textubrace{ \delta_{G}\paren{u,x} + 2\delta_{G}\paren{x,p_{\floor{\frac{k}{2}-1}}\paren{x}} + w\paren{x,y} + 2\delta_{G}\paren{y,p_{\ceil{\frac{k}{2}}}\paren{y}} + \delta_{G}\paren{y,v} }{Triangle inequality}\\
        &
        \leq \textubrace{ \delta_{G}\paren{u,v} + 2\delta_{G}\paren{x,p_{\floor{\frac{k}{2}-1}}\paren{x}} + 2\delta_{G}\paren{y,p_{\ceil{\frac{k}{2}}}\paren{y}} }{$u\squiggly x \straight y \squiggly v$ is a shortest path}
        \\
        &
        \leq \textubrace{ \delta_{G}\paren{u,v} + 2\cdot \paren{\floor{\frac{k}{2}}-1} \cdot w\paren{x,y} + 2\cdot \ceil{\frac{k}{2}}  \cdot w\paren{x,y} }{\clmref{_k-1,kW1+_k-4_W2__case_b_index}}\\
        &
        \leq \textubrace{ \delta_{G}\paren{u,v} + 2\paren{k-1} \cdot w\paren{x,y} }{Simply summing}
    \end{alignat*}
    
    (See \figref{_k-1,kW1+_k-4_W2__case_b} for an illustration of \case{b}{2}).
\end{enumerate}

\end{proof}

For \case{c}, we establish a refined structural property. The presence of two distinct edges $\paren{x,y}, \paren{z,w} \in u\squiggly v$ such that $\paren{x,y}, \paren{z,w} \notin E_{1}$ implies that we could naively upper bound the distances $\delta_{G}\paren{x,p_{i}\paren{x}}$ and $\delta_{G}\paren{w,p_{i}\paren{w}}$ directly by $\delta_{G}\paren{x,w}$. However, the weight of the subpath $x\squiggly w$ may be arbitrarily close to the weight of the entire path $u\squiggly v$. Consequently, relying on the upper bound $\delta_{G}\paren{x,w} \leq \delta_{G}\paren{u,v}$ to deduce $\delta_{G}\paren{x,p_{i}\paren{x}} \leq i \cdot \delta_{G}\paren{x,w} \leq i \cdot \delta_{G}\paren{u,v}$ introduces an unacceptably large multiplicative stretch. To achieve a refined upper bound for $\delta_{G}\paren{x,p_{i}\paren{x}}$, we  utilize instead the distance from $p_{1}\paren{w}$ to the pivot's pivots $p_{i}\paren{p_{1}\paren{x}}$ (and symmetrically for $w$). Specifically:

\begin{claim2}~\label{clm:_k-1,kW1+_k-4_W2__case_c_index}
Assume \case{c} holds. Let $3 \leq i \leq k$. If $p_{1}\paren{x} \notin \ball{w}{1}{2}$ and $p_{1}\paren{w} \notin \ball{x}{1}{2}$, and additionally for every $j \in \bracke{i-2}$ it holds that $p_{j}\paren{x} \notin \ball{p_{1}\paren{w}}{j}{j+2}$ and $p_{j}\paren{w} \notin \ball{p_{1}\paren{x}}{j}{j+2}$, then:

$$  \delta_{G}\paren{x,p_{i}\paren{x}} 
    \leq 
    \begin{cases}
            \frac{i}{2}\cdot \delta_{G}\paren{x,w} + \frac{i+2}{2}\cdot w\paren{x,y} + \frac{i-4}{2}\cdot w\paren{z,w} & i = 0 \pmod 4 \\
            \frac{i}{2}\cdot \delta_{G}\paren{x,w} + \frac{i-2}{2}\cdot w\paren{x,y} + \frac{i}{2}\cdot w\paren{z,w} & i = 2 \pmod 4 \\
            \frac{i-1}{2}\cdot \delta_{G}\paren{x,w} + \frac{i+1}{2}\cdot w\paren{x,y} + \frac{i-1}{2}\cdot w\paren{z,w} & i = 1 \pmod 2
    \end{cases} 
$$

$$  \delta_{G}\paren{w,p_{i}\paren{w}} 
    \leq 
    \begin{cases}
            \frac{i}{2}\cdot \delta_{G}\paren{x,w} + \frac{i-4}{2}\cdot w\paren{x,y} + \frac{i+2}{2}\cdot w\paren{z,w} & i = 0 \pmod 4 \\
            \frac{i}{2}\cdot \delta_{G}\paren{x,w} + \frac{i}{2}\cdot w\paren{x,y} + \frac{i-2}{2}\cdot w\paren{z,w} & i = 2 \pmod 4 \\
            \frac{i-1}{2}\cdot \delta_{G}\paren{x,w} + \frac{i-1}{2}\cdot w\paren{x,y} + \frac{i+1}{2}\cdot w\paren{z,w} & i = 1 \pmod 2
    \end{cases} 
$$
\end{claim2}

\begin{table}[H]
  \begin{center}
    
\small\addtolength{\tabcolsep}{-3pt}
    \begin{tabular}{|c|c|c|c|} 
      \hline
      \multirow{2}{*}{\textbf{\small $i$} }& \multirow{2}{*}{\textbf{\small Additional Assumption} }& \textbf{\small Implications for} & \textbf{\small Implications for}  \\
      && \textbf{\small pivots of $p_{1}\paren{x}, p_{1}\paren{w}$} & \textbf{\small pivots of $x,w$}  \\
      \hhline{|=|=|=|=|}
      &&&\\ [-1em]
      \multirow{2}{*}{$1$} & \multirow{2}{*}{$\paren{x,y},\paren{z,w}\notin E_{1} $ }& \multirow{2}{*}{--} & $\delta_{G} \paren{x,p_{1}\paren{x}}\leq w\paren{x,y}$   \\
      &&& $\delta_{G} \paren{w,p_{1}\paren{w}}\leq w\paren{z,w}$ \\
      &&&\\ [-1em]
      \hline

    &&&\\ [-1em]
    \multirow{2}{*}{$2$} & $p_{1}\paren{w}\notin \ball{x}{1}{2}$& \multirow{2}{*}{--} & $\delta_{G} \paren{x,p_{2}\paren{x}}\leq\delta_{G}\paren{x,w}+ w\paren{z,w}$   \\
      &$p_{1}\paren{x}\notin \ball{w}{1}{2}$&& $\delta_{G} \paren{w,p_{2}\paren{w}}\leq \delta_{G}\paren{x,w}+w\paren{x,y}$ \\
      &&&\\ [-1em]
      \hline

    &&&\\ [-1em]
    \multirow{4}{*}{$3$} & 
    \multirow{2}{*}{$p_{1}\paren{w}\notin \ball{p_{1}\paren{x}}{1}{3}$} & \multirow{2}{*}{$\delta_{G}\paren{p_{1}\paren{x},p_{3}\paren{p_{1}\paren{x}}}, \delta_{G}\paren{p_{1}\paren{w},p_{3}\paren{p_{1}\paren{w}}}$} & 
    $\delta_{G} \paren{x,p_{3}\paren{x}}\leq$   \\
    &&& $\delta_{G}\paren{x,w}+ 2w\paren{x,y}+w\paren{z,w}$   \\
    & \multirow{2}{*}{$p_{1}\paren{x}\notin \ball{p_{1}\paren{w}}{1}{3}$}& 
    \multirow{2}{*}{$\leq\delta_{G}\paren{x,w}+w\paren{x,y}+w\paren{z,w}$} & 
    $\delta_{G} \paren{w,p_{3}\paren{w}}\leq$ \\
    &&&$\delta_{G}\paren{x,w}+ w\paren{x,y}+2w\paren{z,w}$ \\
    &&&\\ [-1em]
    \hline

    &&&\\ [-1em]
    \multirow{4}{*}{$4$} &
    \multirow{2}{*}{$p_{2}\paren{w}\notin \ball{p_{1}\paren{x}}{2}{4}$} & 
    \multirow{2}{*}{$\delta_{G}\paren{p_{1}\paren{x},p_{4}\paren{p_{1}\paren{x}}} \leq 2\delta_{G}\paren{x,w}+2w\paren{x,y}$ }&
    \multirow{2}{*}{$\delta_{G} \paren{x,p_{4}\paren{x}}\leq 2\delta_{G}\paren{x,w}+ 3w\paren{x,y}$}   \\
    && & \\
    & \multirow{2}{*}{$p_{2}\paren{x}\notin \ball{p_{1}\paren{w}}{2}{4}$}&
    \multirow{2}{*}{$\delta_{G}\paren{p_{1}\paren{w},p_{4}\paren{p_{1}\paren{w}}} \leq 2\delta_{G}\paren{x,w}+2w\paren{z,w} $} &
    \multirow{2}{*}{$\delta_{G} \paren{w,p_{4}\paren{w}}\leq 2\delta_{G}\paren{x,w}+ 3w\paren{z,w}$} \\
    &&&\\
    &&&\\ [-1em]
    \hline

    &&&\\ [-1em]
    \multirow{4}{*}{$5$} &
    \multirow{2}{*}{$p_{3}\paren{w}\notin \ball{p_{1}\paren{x}}{3}{5}$}&
    \multirow{2}{*}{$\delta_{G}\paren{p_{1}\paren{x},p_{5}\paren{p_{1}\paren{x}}}, \delta_{G}\paren{p_{1}\paren{w},p_{5}\paren{p_{1}\paren{w}}}$} &
    $\delta_{G} \paren{x,p_{5}\paren{x}}\leq$ \\
    &&& $2\delta_{G}\paren{x,w}+ 3w\paren{x,y}+2w\paren{z,w}$ \\
    & \multirow{2}{*}{$p_{3}\paren{x}\notin \ball{p_{1}\paren{w}}{3}{5}$}&
    \multirow{2}{*}{$\leq 2\delta_{G}\paren{x,w}+2w\paren{x,y}+2w\paren{z,w}$} &
    $\delta_{G} \paren{w,p_{5}\paren{w}}\leq$ \\
    &&& $2\delta_{G}\paren{x,w}+ 2w\paren{x,y}+3w\paren{z,w}$ \\
    &&&\\ [-1em]
    \hline

    &&&\\ [-1em]
    \multirow{4}{*}{$6$} &
    \multirow{2}{*}{$p_{4}\paren{w}\notin \ball{p_{1}\paren{x}}{4}{6}$} & 
    $\delta_{G}\paren{p_{1}\paren{x},p_{6}\paren{p_{1}\paren{x}}} \leq$ &
    $\delta_{G} \paren{x,p_{6}\paren{x}}\leq$ \\
    && $3\delta_{G}\paren{x,w}+w\paren{x,y} +3w\paren{z,w}$ &
    $3\delta_{G}\paren{x,w}+ 2w\paren{x,y} + 3w\paren{z,w}$ \\
    & \multirow{2}{*}{$p_{4}\paren{x}\notin \ball{p_{1}\paren{w}}{4}{6}$}&
    $\delta_{G}\paren{p_{1}\paren{w},p_{6}\paren{p_{1}\paren{w}}} \leq$ &
    $\delta_{G} \paren{w,p_{6}\paren{w}}\leq$ \\
    && $3\delta_{G}\paren{x,w}+3w\paren{x,y}+w\paren{z,w}$ &
    $3\delta_{G}\paren{x,w}+ 3w\paren{x,y}+2w\paren{z,w}$ \\
    &&&\\ [-1em]
    \hline

    &&&\\ [-1em]
    \multirow{4}{*}{$7$} & 
    \multirow{2}{*}{$p_{5}\paren{w}\notin \ball{p_{1}\paren{x}}{5}{7}$}& 
    \multirow{2}{*}{$\delta_{G}\paren{p_{1}\paren{x},p_{7}\paren{p_{1}\paren{x}}}, \delta_{G}\paren{p_{1}\paren{w},p_{7}\paren{p_{1}\paren{w}}}$} & 
    $\delta_{G} \paren{x,p_{7}\paren{x}}\leq$ \\
    &&& $3\delta_{G}\paren{x,w}+ 4w\paren{x,y}+3w\paren{z,w}$ \\
    & \multirow{2}{*}{$p_{5}\paren{x}\notin \ball{p_{1}\paren{w}}{5}{7}$}& 
    \multirow{2}{*}{$\leq 3\delta_{G}\paren{x,w}+3w\paren{x,y}+3w\paren{z,w}$} & 
    $\delta_{G} \paren{w,p_{7}\paren{w}}\leq$ \\
    &&& $3\delta_{G}\paren{x,w}+ 3w\paren{x,y}+4w\paren{z,w}$ \\
    &&&\\ [-1em]
    \hline

    &&&\\ [-1em]
   \multirow{4}{*}{$8$} &
   \multirow{2}{*}{$p_{6}\paren{w}\notin \ball{p_{1}\paren{x}}{6}{8}$} & 
   $\delta_{G}\paren{p_{1}\paren{x},p_{8}\paren{p_{1}\paren{x}}} \leq$ &
   $\delta_{G} \paren{x,p_{8}\paren{x}}\leq$ \\
   && $4\delta_{G}\paren{x,w}+4w\paren{x,y}+2w\paren{z,w}$ &
   $4\delta_{G}\paren{x,w}+ 5w\paren{x,y}+2w\paren{z,w}$ \\
   & \multirow{2}{*}{$p_{6}\paren{x}\notin \ball{p_{1}\paren{w}}{6}{8}$}&
   $\delta_{G}\paren{p_{1}\paren{w},p_{8}\paren{p_{1}\paren{w}}} \leq$ &
   $\delta_{G} \paren{w,p_{8}\paren{w}}\leq$ \\
   && $4\delta_{G}\paren{x,w}+2w\paren{x,y}+4w\paren{z,w}$ &
   $4\delta_{G}\paren{x,w}+ 2w\paren{x,y}+5w\paren{z,w}$ \\
   &&&\\ [-1em]
   \hline

   &&&\\ [-1em]
   \multirow{4}{*}{$9$} & 
   \multirow{2}{*}{$p_{7}\paren{w}\notin \ball{p_{1}\paren{x}}{7}{9}$}& 
   \multirow{2}{*}{$\delta_{G}\paren{p_{1}\paren{x},p_{9}\paren{p_{1}\paren{x}}}, \delta_{G}\paren{p_{1}\paren{w},p_{9}\paren{p_{1}\paren{w}}}$} & 
   $\delta_{G} \paren{x,p_{9}\paren{x}}\leq$ \\
   &&& $4\delta_{G}\paren{x,w}+ 5w\paren{x,y}+4w\paren{z,w}$ \\
   & \multirow{2}{*}{$p_{7}\paren{x}\notin \ball{p_{1}\paren{w}}{7}{9}$}& 
   \multirow{2}{*}{$\leq 4\delta_{G}\paren{x,w}+4w\paren{x,y}+4w\paren{z,w}$} & 
   $\delta_{G} \paren{w,p_{9}\paren{w}}\leq$ \\
   &&& $4\delta_{G}\paren{x,w}+ 4w\paren{x,y}+5w\paren{z,w}$ \\
   &&&\\ [-1em]
   \hline

   &&&\\ [-1em]
   \multirow{4}{*}{$10$} &
   \multirow{2}{*}{$p_{8}\paren{w}\notin \ball{p_{1}\paren{x}}{8}{10}$} & 
   $\delta_{G}\paren{p_{1}\paren{x},p_{10}\paren{p_{1}\paren{x}}} \leq$ &
   $\delta_{G} \paren{x,p_{10}\paren{x}}\leq$ \\
   && $5\delta_{G}\paren{x,w}+3w\paren{x,y}+5w\paren{z,w}$ &
   $5\delta_{G}\paren{x,w}+ 4w\paren{x,y}+5w\paren{z,w}$ \\
   & \multirow{2}{*}{$p_{8}\paren{x}\notin \ball{p_{1}\paren{w}}{8}{10}$}&
   $\delta_{G}\paren{p_{1}\paren{w},p_{10}\paren{p_{1}\paren{w}}} \leq$ &
   $\delta_{G} \paren{w,p_{10}\paren{w}}\leq$ \\
   && $5\delta_{G}\paren{x,w}+5w\paren{x,y}+3w\paren{z,w}$ &
   $5\delta_{G}\paren{x,w}+ 5w\paren{x,y}+4w\paren{z,w}$ \\
   &&&\\ [-1em]
   \hline

   &&&\\ [-1em]
   \multirow{4}{*}{$11$} & 
   \multirow{2}{*}{$p_{9}\paren{w}\notin \ball{p_{1}\paren{x}}{9}{11}$}& 
   \multirow{2}{*}{$\delta_{G}\paren{p_{1}\paren{x},p_{11}\paren{p_{1}\paren{x}}}, \delta_{G}\paren{p_{1}\paren{w},p_{11}\paren{p_{1}\paren{w}}}$} & 
   $\delta_{G} \paren{x,p_{11}\paren{x}}\leq$ \\
   &&& $5\delta_{G}\paren{x,w}+ 6w\paren{x,y}+5w\paren{z,w}$ \\
   & \multirow{2}{*}{$p_{9}\paren{x}\notin \ball{p_{1}\paren{w}}{9}{11}$}& 
   \multirow{2}{*}{$\leq 5\delta_{G}\paren{x,w}+5w\paren{x,y}+5w\paren{z,w}$} & 
   $\delta_{G} \paren{w,p_{11}\paren{w}}\leq$ \\
   &&& $5\delta_{G}\paren{x,w}+ 5w\paren{x,y}+6w\paren{z,w}$ \\
   &&&\\ [-1em]
   \hline

   &&&\\ [-1em]
   \multirow{4}{*}{$12$} &
   \multirow{2}{*}{$p_{10}\paren{w}\notin \ball{p_{1}\paren{x}}{10}{12}$} & 
   $\delta_{G}\paren{p_{1}\paren{x},p_{12}\paren{p_{1}\paren{x}}} \leq$ &
   $\delta_{G} \paren{x,p_{12}\paren{x}}\leq$ \\
   && $6\delta_{G}\paren{x,w}+6w\paren{x,y}+4w\paren{z,w}$ &
   $6\delta_{G}\paren{x,w}+ 7w\paren{x,y}+4w\paren{z,w}$ \\
   & \multirow{2}{*}{$p_{10}\paren{x}\notin \ball{p_{1}\paren{w}}{10}{12}$}&
   $\delta_{G}\paren{p_{1}\paren{w},p_{12}\paren{p_{1}\paren{w}}} \leq$ &
   $\delta_{G} \paren{w,p_{12}\paren{w}}\leq$ \\
   && $6\delta_{G}\paren{x,w}+4w\paren{x,y}+6w\paren{z,w}$ &
   $6\delta_{G}\paren{x,w}+ 4w\paren{x,y}+7w\paren{z,w}$ \\
   \hline

    \end{tabular}
    \caption{Several demonstrations of \clmref{_k-1,kW1+_k-4_W2__case_c_index} for $i\geq 3$ and the separate cases of $i=1$ and $i=2$. The left column holds the assumptions for these specific base cases. The middle column presents the immediate implications for the distances $\delta_{G}\paren{p_{1}\paren{x},p_{i}\paren{p_{1}\paren{x}}}$ and $\delta_{G}\paren{p_{1}\paren{w},p_{i}\paren{p_{1}\paren{w}}}$. The right column establishes our desired upper bounds for the distances $\delta_{G}\paren{x,p_{i}\paren{p_{1}\paren{x}}}$ and $\delta_{G}\paren{w,p_{i}\paren{p_{1}\paren{w}}}$. }~\label{tab:_k-1,kW1+_k-4_W2__case_c_index}
  \end{center}
\end{table} 

\begin{proof}

We prove this claim by induction on $i$. As the upper bound depends on the value of $i \pmod 4$, the induction requires four consecutive base cases. We establish the base cases for $i \in \bracce{2, 3, 4, 5}$. 

We begin with the first base case, $i=2$ (where $2 = 2 \pmod 4$). By the premise of the claim, we have $p_{1}\paren{w} \notin \ball{x}{1}{2}$. Thus, we obtain:

\begin{alignat*}{3}
        \delta_{G}\paren{x,p_{2}\paren{x}} & \leq \textubrace{ \delta_G\paren{x,p_{1}\paren{w}} }{$p_{1} \paren{w} \notin \ball{x}{1}{2}$} 
        &
        \leq \textubrace{ \delta_{G}\paren{x,w}+ \delta_{G}\paren{w,p_{1}\paren{w}}}{Triangle inequality}  
        &
        \leq \textubrace{ \delta_{G}\paren{x,w}+w\paren{z,w} }{$\delta_{G}\paren{w,p_{1}\paren{w}} \leq w\paren{z,w}$} 
    \end{alignat*}

This matches the required upper bound for $i = 2 \pmod 4$. Symmetrically, by reversing the roles of $x$ and $w$ and applying $p_{1}\paren{x} \notin \ball{w}{1}{2}$, it follows that $\delta_{G}\paren{w,p_{2}\paren{w}} \leq \delta_{G}\paren{x,w}+w\paren{x,y}$.

For the second base case, $i=3$ (where $3 = 1 \pmod 2$), we consider as well the additional assumption that $p_{1}\paren{w} \notin \ball{p_{1}\paren{x}}{1}{3}$. Therefore:

\begin{alignat*}{3}
        \delta_{G}\paren{p_{1}\paren{x},p_{3}\paren{p_{1}\paren{x}}} & \leq \textubrace{ \delta_{G}\paren{p_{1}\paren{x},p_{1}\paren{w}}  }{$p_{1}\paren{w} \notin \ball{p_{1}\paren{x}}{1}{3}$} 
        &
        \leq \textubrace{  \delta_{G}\paren{x,p_{1}\paren{x}} + \delta_{G}\paren{x,w} + \delta_{G}\paren{w,p_{1}\paren{w}}}{Triangle inequality} 
        &
        \leq \textubrace{ \delta_{G}\paren{x,w} + w\paren{x,y} +w\paren{z,w} }{Assumptions on $p_{1}\paren{x},p_{1}\paren{w}$} 
\end{alignat*}

Observe that the distance between $x$ and the pivot $p_{3}\paren{x}$ is upper bounded by routing through $p_{1}\paren{x}$ to its respective pivot $p_{3}\paren{p_{1}\paren{x}}$. Consequently, we obtain: 

\begin{alignat*}{3}
        \delta_{G}\paren{x,p_{3}\paren{x}} & \leq \textubrace{ \delta_{G}\paren{x,p_{1}\paren{x}} + \delta_{G}\paren{p_{1}\paren{x},p_{3}\paren{p_{1}\paren{x}}} }{Pivot definition} 
        &
        \leq \textubrace{ \delta_{G}\paren{x,w} + 2w\paren{x,y} +w\paren{z,w} }{Our previous conclusion} 
\end{alignat*}

Overall, $\delta_{G}\paren{x,w} + 2w\paren{x,y} + w\paren{z,w}$ fits the stated bound for $i = 3 = 1 \pmod 2$ case. A fully symmetric argument holds for $w$: by exchanging $x$ and $w$ and utilizing the complementary assumption $p_{1}\paren{x} \notin \ball{p_{1}\paren{w}}{1}{3}$, we directly deduce the symmetric upper bound bound $\delta_{G}\paren{w,p_{3}\paren{w}} \leq \delta_{G}\paren{x,w} + w\paren{x,y} + 2w\paren{z,w}$.

For the third base case, $i=4$ (where $4 = 0 \pmod 4$), we consider the additional assumption that $p_{2}\paren{w} \notin \ball{p_{1}\paren{x}}{2}{4}$. Thus:

\begin{alignat*}{3}
        \delta_{G}\paren{p_{1}\paren{x},p_{4}\paren{p_{1}\paren{x}}} & \leq \textubrace{ \delta_{G}\paren{p_{1}\paren{x},p_{2}\paren{w}} }{$p_{2}\paren{w} \notin \ball{p_{1}\paren{x}}{2}{4}$} 
        &
        \leq \textubrace{ \delta_{G}\paren{x,p_{1}\paren{x}} + \delta_{G}\paren{x,w} + \delta_{G}\paren{w,p_{2}\paren{w}} }{Triangle inequality} 
        &
        \leq \textubrace{ 2\delta_{G}\paren{x,w} + 2w\paren{x,y} }{Base case $i=2$ for $w$} 
\end{alignat*}

Note that the distance between $x$ and its pivot $p_{4}\paren{x}$ is upper bounded by routing through $p_{1}\paren{x}$ to its respective pivot $p_{4}\paren{p_{1}\paren{x}}$. Therefore:

\begin{alignat*}{3}
        \delta_{G}\paren{x,p_{4}\paren{x}} & \leq \textubrace{ \delta_{G}\paren{x,p_{1}\paren{x}} + \delta_{G}\paren{p_{1}\paren{x},p_{4}\paren{p_{1}\paren{x}}} }{Pivot definition} 
        &
        \leq \textubrace{ 2\delta_{G}\paren{x,w} + 3w\paren{x,y} }{Our previous conclusion} 
\end{alignat*}

The upper bound we received fits the stated bound for the case of $i = 4 = 0 \pmod 4$. A fully symmetric argument holds for $w$: by exchanging $x$ and $w$ and utilizing the complementary assumption $p_{2}\paren{x} \notin \ball{p_{1}\paren{w}}{2}{4}$, we directly deduce the symmetric upper bound $\delta_{G}\paren{w,p_{4}\paren{w}} \leq 2\delta_{G}\paren{x,w} + 3w\paren{z,w}$.

For the final base case, $i=5$ (where $5 = 1 \pmod 2$), we rely on the given premise that $p_{3}\paren{w} \notin \ball{p_{1}\paren{x}}{3}{5}$. From this, we derive:

\begin{alignat*}{3}
        \delta_{G}\paren{p_{1}\paren{x},p_{5}\paren{p_{1}\paren{x}}} & \leq \textubrace{ \delta_{G}\paren{p_{1}\paren{x},p_{3}\paren{w}} }{$p_{3}\paren{w} \notin \ball{p_{1}\paren{x}}{3}{5}$} \\
        &
        \leq \textubrace{ \delta_{G}\paren{x,p_{1}\paren{x}} + \delta_{G}\paren{x,w} + \delta_{G}\paren{w,p_{3}\paren{w}} }{Triangle inequality} \\
        &
        \leq \textubrace{ 2\delta_{G}\paren{x,w} + 2w\paren{x,y} + 2w\paren{z,w} }{Base case $i=3$ for $w$} 
\end{alignat*}

We can upper bound the distance from $x$ to its pivot $p_{5}\paren{x}$ by routing through $p_{1}\paren{x}$ to reach $p_{5}\paren{p_{1}\paren{x}}$. This yields:

\begin{alignat*}{3}
        \delta_{G}\paren{x,p_{5}\paren{x}} & \leq \textubrace{ \delta_{G}\paren{x,p_{1}\paren{x}} + \delta_{G}\paren{p_{1}\paren{x},p_{5}\paren{p_{1}\paren{x}}} }{Pivot definition} 
        &
        \leq \textubrace{ 2\delta_{G}\paren{x,w} + 3w\paren{x,y} + 2w\paren{z,w} }{Our previous conclusion} 
\end{alignat*}

This satisfies the required bound for $i = 5 = 1 \pmod 2$. A fully symmetric argument, where the roles of $x$ and $w$ are replaced, and the condition $p_{3}\paren{x} \notin \ball{p_{1}\paren{w}}{3}{5}$ is applied, yields the symmetric upper bound $\delta_{G}\paren{w,p_{5}\paren{w}} \leq 2\delta_{G}\paren{x,w} + 2w\paren{x,y} + 3w\paren{z,w}$.

For several explicit examples and a structural summary of these initial bounds, we refer the reader to \tabref{_k-1,kW1+_k-4_W2__case_c_index}. We now inductively assume  that the claim holds for all indices up to $i-1$. We proceed to prove the claim for an arbitrary $i \geq 6$. By the premise of the claim,  we assume that $p_{i-2}\paren{w} \notin \ball{p_{1}\paren{x}}{i-2}{i}$. We first bound the distance from $p_{1}\paren{x}$ to its  pivot $p_{i}\paren{p_{1}\paren{x}}$. 

\begin{alignat*}{3}
        \delta_{G}\paren{p_{1}\paren{x},p_{i}\paren{p_{1}\paren{x}}} & \leq \textubrace{ \delta_{G}\paren{p_{1}\paren{x},p_{i-2}\paren{w}} }{$p_{i-2}\paren{w} \notin \ball{p_{1}\paren{x}}{i-2}{i}$} \\
        &
        \leq \textubrace{ \delta_{G}\paren{x,p_{1}\paren{x}} + \delta_{G}\paren{x,w} + \delta_{G}\paren{w,p_{i-2}\paren{w}} }{Triangle inequality} \\
        &
        \leq \textubrace{  \delta_{G}\paren{x,w} + w\paren{x,y} +\delta_{G}\paren{w,p_{i-2}\paren{w}} }{$\paren{x,y}\notin E_{1}$} 
\end{alignat*}

Note that the distance between $x$ and its pivot $p_{i}\paren{x}$ is upper bounded by the weight of the path $x\squiggly p_{1}\paren{x} \squiggly p_{i}\paren{p_{1}\paren{x}}$. Hence:

\begin{alignat}{3}~\label{eqn:_k-1,kW1+_k-4_W2__case_c_index_step}
        \delta_{G}\paren{x,p_{i}\paren{x}} & \leq \textubrace{ \delta_{G}\paren{x,p_{1}\paren{x}} + \delta_{G}\paren{p_{1}\paren{x},p_{i}\paren{p_{1}\paren{x}}} }{Pivot definition} 
        &
        \leq \textubrace{  \delta_{G}\paren{x,w} + 2w\paren{x,y} +\delta_{G}\paren{w,p_{i-2}\paren{w}} }{Our previous conclusion} 
\end{alignat}

Having establish the above inequality, we distinguish between the possible cases of values of  $i \pmod 4$ and apply the inductive hypothesis to $\delta_{G}\paren{w,p_{i-2}\paren{w}}$:

\begin{enumerate}
    \item \dashuline{Case 1: $i = 0 \pmod 4$.} \\
    Since $i = 0 \pmod 4$, it follows that $i-2 = 2 \pmod 4$. By applying the inductive hypothesis for $w$ specifically at index $i-2$, we have:
    $$ \delta_{G}\paren{w,p_{i-2}\paren{w}} \leq \frac{i-2}{2}\cdot\delta_{G}\paren{x,w} + \frac{i-2}{2}\cdot w\paren{x,y} + \frac{i-4}{2}\cdot w\paren{z,w} $$
    Substituting this into our general recursive bound in \eqnref{_k-1,kW1+_k-4_W2__case_c_index_step} gives:
    \begin{align*}
            \delta_{G}\paren{x,p_{i}\paren{x}}&\leq  \delta_{G}\paren{x,w} + 2w\paren{x,y}+\frac{i-2}{2}\cdot\delta_{G}\paren{x,w} + \frac{i-2}{2}\cdot w\paren{x,y} + \frac{i-4}{2}\cdot w\paren{z,w}\\ 
            &\leq \paren{ \delta_{G}\paren{x,w} + \frac{i-2}{2}\cdot\delta_{G}\paren{x,w} } + \paren{ 2w\paren{x,y} + \frac{i-2}{2}\cdot w\paren{x,y} } + \frac{i-4}{2}\cdot w\paren{z,w} \\
            &= \frac{i}{2}\cdot\delta_{G}\paren{x,w} + \frac{i+2}{2}\cdot w\paren{x,y} + \frac{i-4}{2}\cdot w\paren{z,w}
    \end{align*}
    This results in the required upper bound for $i = 0 \pmod 4$. A symmetric argument holds for $w$: utilizing the assumption $p_{i-2}\paren{x} \notin \ball{p_{1}\paren{w}}{i-2}{i}$ and replacing the roles of $x$ and $w$. 

    \item \dashuline{Case 2: $i = 2 \pmod 4$.} \\
    In this case, $i = 2 \pmod 4$, meaning $i-2 = 0 \pmod 4$. Invoking the inductive hypothesis for $w$ at the $i-2$ index results with:
    $$ \delta_{G}\paren{w,p_{i-2}\paren{w}} \leq \frac{i-2}{2}\cdot\delta_{G}\paren{x,w} + \frac{i-6}{2}\cdot w\paren{x,y} + \frac{i}{2}\cdot w\paren{z,w} $$
    Substituting this within \eqnref{_k-1,kW1+_k-4_W2__case_c_index_step}, we get:
    \begin{align*}
            \delta_{G}\paren{x,p_{i}\paren{x}} &\leq \delta_{G}\paren{x,w} + 2w\paren{x,y} + \frac{i-2}{2}\cdot\delta_{G}\paren{x,w} + \frac{i-6}{2}\cdot w\paren{x,y} + \frac{i}{2}\cdot w\paren{z,w} \\
            &\leq \paren{ \delta_{G}\paren{x,w} + \frac{i-2}{2}\cdot\delta_{G}\paren{x,w} } + \paren{ 2w\paren{x,y} + \frac{i-6}{2}\cdot w\paren{x,y} } + \frac{i}{2}\cdot w\paren{z,w} \\
            &= \frac{i}{2}\cdot\delta_{G}\paren{x,w} + \frac{i-2}{2}\cdot w\paren{x,y} + \frac{i}{2}\cdot w\paren{z,w}
    \end{align*}
    The upper bound we received satisfies the stated condition for $i = 2 \pmod 4$. Replacing the roles of $x$ and $w$ alongside the assumption that $p_{i-2}\paren{x} \notin \ball{p_{1}\paren{w}}{i-2}{i}$ similarly establishes the symmetric counterpart for $\delta_{G}\paren{w,p_{i}\paren{w}}$.

    \item \dashuline{Case 3: $i = 1 \pmod 2$.} \\
    Because $i$ is odd, the index $i-2$ satisfies $i-2 = 1 \pmod 2$. Applying the inductive hypothesis for $w$ at index $i-2$ yields:
    $$ \delta_{G}\paren{w,p_{i-2}\paren{w}} \leq \frac{i-3}{2}\cdot\delta_{G}\paren{x,w} + \frac{i-3}{2}\cdot w\paren{x,y} + \frac{i-1}{2}\cdot w\paren{z,w} $$
    Substituting this into \eqnref{_k-1,kW1+_k-4_W2__case_c_index_step} gives:
    \begin{align*}
            \delta_{G}\paren{x,p_{i}\paren{x}} &\leq \delta_{G}\paren{x,w} + 2w\paren{x,y} + \frac{i-3}{2}\cdot\delta_{G}\paren{x,w} + \frac{i-3}{2}\cdot w\paren{x,y} + \frac{i-1}{2}\cdot w\paren{z,w} \\
            &\leq \paren{ \delta_{G}\paren{x,w} + \frac{i-3}{2}\cdot\delta_{G}\paren{x,w} } + \paren{ 2w\paren{x,y} + \frac{i-3}{2}\cdot w\paren{x,y} } + \frac{i-1}{2}\cdot w\paren{z,w} \\
            &= \frac{i-1}{2}\cdot\delta_{G}\paren{x,w} + \frac{i+1}{2}\cdot w\paren{x,y} + \frac{i-1}{2}\cdot w\paren{z,w}
    \end{align*}
    This results in the required upper bound for $i = 1 \pmod 2$. A symmetric argument holds for $w$: utilizing the assumption $p_{i-2}\paren{x} \notin \ball{p_{1}\paren{w}}{i-2}{i}$ and replacing the roles of $x$ and $w$.

\end{enumerate}

This concludes the proof of the claim.

\end{proof}

\input{\emulatorsImagesPath _k-1,kW1+_k-4_W2_/case_c}

Leveraging this property we can now prove:
\begin{lemma}~\label{lem:_k-1,kW1+_k-4_W2__case_c} 
   If \case{c} holds then:
   $$\delta_H\paren{u,v} \leq 
   \begin{cases}
            \paren{k-1}\cdot \delta_{G}\paren{u,v} + \paren{k-4}\cdot W_{1}+k\cdot W_{2} & k = 0 \pmod 2 \\
            \paren{k-2}\cdot \delta_{G}\paren{u,v} + \paren{k-1}\cdot \paren{W_{1}+W_{2}} & k = 1 \pmod 2 
    \end{cases} 
   $$.
\end{lemma}
\begin{proof} 
We assume \case{c} holds. To prove the upper bound for $\delta_{H}\paren{u,v}$, we divide the proof into three distinct subcases, based on whether a pivot of one vertex falls within the corresponding bunch of the other. 

\begin{enumerate}
    \item \dashuline{\case{c}{1}: $p_{1}\paren{x} \in \ball{w}{1}{2}$ or $p_{1}\paren{w} \in \ball{x}{1}{2}$.} \\
    We assume without loss of generality that $p_{1}\paren{w} \in \ball{x}{1}{2}$. Observe that  $\paren{x, p_{1}\paren{w}}\in B_{1}\paren{x}\subseteq  B_{1}\paren{V} \subseteq F$. Using the triangle inequality, we can directly bound the overall distance of $u \squiggly x \squiggly p_{1}\paren{w} \squiggly w \squiggly v$:
    \begin{alignat*}{3}
            \delta_{H}\paren{u,v} & \leq \textubrace{\delta_{H}\paren{u,x} + w_{H}\paren{x,p_{1}\paren{w}} + w_{H}\paren{p_{1}\paren{w},w} + \delta_{H}\paren{w,v}}{Auxiliary edges in $F$} \\
            &= \textubrace{ \delta_{G}\paren{u,x} + \delta_{G}\paren{x,p_{1}\paren{w}} + \delta_{G}\paren{w,p_{1}\paren{w}} + \delta_{G}\paren{w,v} }{Definition of $w_{H}$} 
            \\
            & \leq \textubrace{ \delta_{G}\paren{u,x} + \delta_{G}\paren{w,p_{1}\paren{w}} +\delta_{G}\paren{x,w}+ \delta_{G}\paren{w,p_{1}\paren{w}} + \delta_{G}\paren{w,v} }{Triangle inequality} 
            \\
            & \leq \textubrace{ \delta_{G}\paren{u,v} + 2 \delta_{G}\paren{w,p_{1}\paren{w}} }{$u\squiggly x \squiggly w \squiggly v$ is a shortest path} 
            \\
            & \leq \textubrace{ \delta_{G}\paren{u,v} + 2w\paren{z,w}}{$\paren{x,y},\paren{z,w}\notin E_{1}$} 
    \end{alignat*}
    Which is strictly smaller than the required upper bound for any $k\geq 3$, regardless of the parity of $k$.

    \item \dashuline{\case{c}{2}:
    $p_{1}\paren{x} \notin \ball{w}{1}{2}$, $p_{1}\paren{w} \notin \ball{x}{1}{2}$  and there exists an $2 \leq i \leq k-4$ such that}\\
    \dashuline{$p_{i}\paren{x} \in \ball{p_{1}\paren{w}}{i}{i+2}$ or $p_{i}\paren{w} \in \ball{p_{1}\paren{x}}{i}{i+2}$.} \\
    Let $i$ be the minimal such index. Hence, for every $j \in \bracke{i-1}$, $p_{j}\paren{x} \notin \ball{p_{1}\paren{w}}{j}{j+2}$ and $p_{j}\paren{w} \notin \ball{p_{1}\paren{x}}{j}{j+2}$. Consequently, the premises for \clmref{_k-1,kW1+_k-4_W2__case_c_index} hold up to index $i$. Without loss of generality, let us assume that $p_{i}\paren{x} \in \ball{p_{1}\paren{w}}{i}{i+2}$. Observe that $\paren{p_{i}\paren{x}, p_{1}\paren{w}}\in B_{2}\paren{p_{1}\paren{w}}\subseteq  B_{2}\paren{S_{1}} \subseteq F$. Using the triangle inequality,we can  bound the weight of $u\squiggly v $ by the weight of $u \squiggly x \squiggly p_{i}\paren{x} \squiggly p_{1}\paren{w} \squiggly w \squiggly v$. That is:
    \begin{alignat*}{3}
            \delta_{H}\paren{u,v} & \leq \textubrace{ \delta_{H}\paren{u,x} + w_{H}\paren{x,p_{i}\paren{x}} + w_{H}\paren{p_{i}\paren{x},p_{1}\paren{w}} + w_{H}\paren{p_{1}\paren{w},w} + \delta_{H}\paren{w,v} }{Auxiliary edges in $F$} \\
            &= \textubrace{ \delta_{G}\paren{u,x} + \delta_{G}\paren{x,p_{i}\paren{x}} + \delta_{G}\paren{p_{i}\paren{x},p_{1}\paren{w}} + \delta_{G}\paren{w,p_{1}\paren{w}} + \delta_{G}\paren{w,v} }{Definition of $w_{H}$} 
            \\
            & \leq \textubrace{ \delta_{G}\paren{u,x} + 2\delta_{G}\paren{x,p_{i}\paren{x}} + \delta_{G}\paren{x,w} + 2\delta_{G}\paren{w,p_{1}\paren{w}} + \delta_{G}\paren{w,v} }{Triangle inequality} 
            \\
            & \leq \textubrace{ \delta_{G}\paren{u,v} + 2\delta_{G}\paren{x,p_{i}\paren{x}} + 2\delta_{G}\paren{w,p_{1}\paren{w}} }{$u\squiggly x \squiggly w \squiggly v$ is a shortest path} \\
            & \leq \textubrace{ \delta_{G}\paren{u,v} + 2\delta_{G}\paren{x,p_{i}\paren{x}} + 2w\paren{z,w} }{$\paren{z,w}\notin E_{1}$} 
    \end{alignat*}

    We now consider three possibilities for $i\pmod 4$: 
    \begin{enumerate}
        \item \dashuline{$i=0\pmod 4$}\\
            \begin{alignat*}{3}
                & \leq \textubrace{ \delta_{G}\paren{u,v} + 2\cdot \paren{\frac{i}{2}\cdot \delta_{G}\paren{x,w}+ \frac{i+2}{2}\cdot w\paren{x,y}+\frac{i-4}{2}\cdot w\paren{z,w}}+ 2w\paren{z,w} }{\clmref{_k-1,kW1+_k-4_W2__case_c_index}} \\
                & = \textubrace{ \paren{i+1}\cdot\delta_{G}\paren{u,v} + \paren{i+2} \cdot w\paren{x,y} + \paren{i-2} \cdot w\paren{z,w} }{Simply summing} \\
                & \leq \textubrace{ \paren{k-3}\cdot\delta_{G}\paren{u,v} + \paren{k-2} \cdot w\paren{x,y} + \paren{k-6} \cdot w\paren{z,w} }{$i\leq k-4$} \\
                & \leq \textubrace{ \paren{k-1}\cdot\delta_{G}\paren{u,v} + \paren{k-4} \cdot w\paren{x,y} + \paren{k-8} \cdot w\paren{z,w} }{$w\paren{x,y}+w\paren{z,w}\leq\delta_{G}\paren{u,v}$} \\
                & < \textubrace{ \paren{k-1}\cdot\delta_{G}\paren{u,v} + \paren{k-4} \cdot W_{1}\paren{u\squiggly v} + k \cdot W_{2}\paren{u\squiggly v} }{$w\paren{x,y},w\paren{z,w} \leq W_{1},W_{2}$} 
            \end{alignat*}

        \item \dashuline{$i=2\pmod 4$}\\
            \begin{alignat*}{3}
                & \leq \textubrace{ \delta_{G}\paren{u,v} + 2\cdot \paren{\frac{i}{2}\cdot \delta_{G}\paren{x,w}+ \frac{i-2}{2}\cdot w\paren{x,y}+\frac{i}{2}\cdot w\paren{z,w}}+ 2w\paren{z,w} }{\clmref{_k-1,kW1+_k-4_W2__case_c_index}} \\
                & = \textubrace{ \paren{i+1}\cdot\delta_{G}\paren{u,v} + \paren{i-2} \cdot w\paren{x,y} + \paren{i+2} \cdot w\paren{z,w} }{Simply summing} \\
                & \leq \textubrace{ \paren{k-3}\cdot\delta_{G}\paren{u,v} + \paren{k-6} \cdot w\paren{x,y} + \paren{k-2} \cdot w\paren{z,w} }{$i\leq k-4$} \\
                & \leq \textubrace{ \paren{k-1}\cdot\delta_{G}\paren{u,v} + \paren{k-8} \cdot w\paren{x,y} + \paren{k-4} \cdot w\paren{z,w} }{$w\paren{x,y}+w\paren{z,w}\leq\delta_{G}\paren{u,v}$} \\
                & < \textubrace{ \paren{k-1}\cdot\delta_{G}\paren{u,v} + \paren{k-4} \cdot W_{1}\paren{u\squiggly v} + k \cdot W_{2} \paren{u\squiggly v}}{$w\paren{x,y},w\paren{z,w} \leq W_{1},W_{2}$} 
            \end{alignat*}

        \item \dashuline{$i=1\pmod 2$}\\
            \begin{alignat*}{3}
                & \leq \textubrace{ \delta_{G}\paren{u,v} + 2\cdot \paren{\frac{i-1}{2}\cdot \delta_{G}\paren{x,w}+ \frac{i+1}{2}\cdot w\paren{x,y}+\frac{i-1}{2}\cdot w\paren{z,w}}+ 2w\paren{z,w} }{\clmref{_k-1,kW1+_k-4_W2__case_c_index}} \\
                & = \textubrace{ i\cdot\delta_{G}\paren{u,v} + \paren{i+1} \cdot w\paren{x,y} + \paren{i+1} \cdot w\paren{z,w} }{Simply summing} \\
                & \leq \textubrace{ \paren{k-4}\cdot\delta_{G}\paren{u,v} + \paren{k-3} \cdot w\paren{x,y} + \paren{k-3} \cdot w\paren{z,w} }{$i\leq k-4$} \\
                & < \textubrace{ \paren{k-2}\cdot\delta_{G}\paren{u,v} + \paren{k-1} \paren{W_{1}\paren{u\squiggly v} +  W_{2} \paren{u\squiggly v}}}{$w\paren{x,y},w\paren{z,w} \leq W_{1},W_{2}$} 
            \end{alignat*}
 
    \end{enumerate}

 \item \dashuline{\case{c}{3}:
    $p_{1}\paren{x} \notin \ball{w}{1}{2}$, $p_{1}\paren{w} \notin \ball{x}{1}{2}$  and for any $2 \leq i \leq k-4$ it holds that}\\
    \dashuline{$p_{i}\paren{x} \notin \ball{p_{1}\paren{w}}{i}{i+2}$ and $p_{i}\paren{w} \notin \ball{p_{1}\paren{x}}{i}{i+2}$.}\\
    Observe that as we now  satisfy the condition of \clmref{_k-1,kW1+_k-4_W2__case_c_index}, we can provide an upper bound for both $\delta_{G} \paren{x,p_{k-2}\paren{x}}$ and $\delta_{G} \paren{w,p_{k-2}\paren{w}}$, depending on $k \pmod 4$. Additionally, we may use the auxiliary edges of the form $S_{1}\times S_{k-2}\subseteq F$. At first, we consider the auxiliary edge $\paren{p_{k-2}\paren{x},p_{1}\paren{w}}$: 
    \begin{alignat}{3}~\label{eqn:_k-1,kW1+_k-4_W2__case_c3_x}
            \delta_{H}\paren{u,v} & \leq \textubrace{ \delta_{H}\paren{u,x} + w_{H}\paren{x,p_{k-2}\paren{x}} + w_{H}\paren{p_{k-2}\paren{x},p_{1}\paren{w}} + w_{H}\paren{p_{1}\paren{w},w} + \delta_{H}\paren{w,v} }{Auxiliary edges in $F$} \nonumber \\
            &= \textubrace{ \delta_{G}\paren{u,x} + \delta_{G}\paren{x,p_{k-2}\paren{x}} + \delta_{G}\paren{p_{k-2}\paren{x},p_{1}\paren{w}} + \delta_{G}\paren{w,p_{1}\paren{w}} + \delta_{G}\paren{w,v} }{Definition of $w_{H}$}  \nonumber
            \\
            & \leq \textubrace{ \delta_{G}\paren{u,x}  + 2\delta_{G}\paren{x,p_{k-2}\paren{x}} + \delta_{G}\paren{x,w} + 2\delta_{G}\paren{w,p_{1}\paren{w}}+\delta_{G}\paren{w,v} }{Triangle inequality} \nonumber
            \\
            & \leq \textubrace{ \delta_{G}\paren{u,v} + 2\delta_{G}\paren{x,p_{k-2}\paren{x}} + 2\delta_{G}\paren{w,p_{1}\paren{w}}  }{$u\squiggly x \squiggly w \squiggly v$ is a shortest path} \nonumber \\
            & \leq \textubrace{ \delta_{G}\paren{u,v} + 2\delta_{G}\paren{x,p_{k-2}\paren{x}} + 2w\paren{z,w} }{$\paren{z,w} \notin E_{1}$}
    \end{alignat}

    Symmetrically, we may utilize the auxiliary edge $\paren{p_{1}\paren{x},p_{k-2}\paren{w}}$, obtaining:
    \begin{alignat}{3}~\label{eqn:_k-1,kW1+_k-4_W2__case_c3_w}
            \delta_{H}\paren{u,v} & \leq \textubrace{ \delta_{H}\paren{u,x} + w_{H}\paren{x,p_{1}\paren{x}} + w_{H}\paren{p_{1}\paren{x},p_{k-2}\paren{w}} + w_{H}\paren{p_{k-2}\paren{w},w} + \delta_{H}\paren{w,v} }{Auxiliary edges in $F$} \nonumber \\
            &= \textubrace{ \delta_{G}\paren{u,x} + \delta_{G}\paren{x,p_{1}\paren{x}} + \delta_{G}\paren{p_{1}\paren{x},p_{k-2}\paren{w}} + \delta_{G}\paren{w,p_{k-2}\paren{w}} + \delta_{G}\paren{w,v} }{Definition of $w_{H}$} \nonumber
            \\
            & \leq \textubrace{ \delta_{G}\paren{u,x}  + 2\delta_{G}\paren{x,p_{1}\paren{x}} + \delta_{G}\paren{x,w} + 2\delta_{G}\paren{w,p_{k-2}\paren{w}}+\delta_{G}\paren{w,v} }{Triangle inequality} \nonumber
            \\
            & \leq \textubrace{ \delta_{G}\paren{u,v} + 2\delta_{G}\paren{x,p_{1}\paren{x}} + 2\delta_{G}\paren{w,p_{k-2}\paren{w}}  }{$u\squiggly x \squiggly w \squiggly v$ is a shortest path} \nonumber \\
            & \leq \textubrace{ \delta_{G}\paren{u,v} + 2w\paren{x,y} + 2\delta_{G}\paren{w,p_{k-2}\paren{w}} }{$\paren{x,y} \notin E_{1}$}
    \end{alignat}
    
    \begin{enumerate}
        \item \dashuline{$k = 2 \pmod 4$}\\
        Hence $k-2 = 0 \pmod 4$. Therefore, expanding \eqnref{_k-1,kW1+_k-4_W2__case_c3_x} by \clmref{_k-1,kW1+_k-4_W2__case_c_index}, we get:
        \begin{alignat*}{3}
            \delta_{H}\paren{u,v} & \leq \textubrace{ \delta_{G}\paren{u,v} + 2\cdot\paren{\frac{k-2}{2}\cdot\delta_{G}\paren{x,w} + \frac{k}{2}\cdot w\paren{x,y} + \frac{k-6}{2}\cdot w\paren{z,w}} + 2w\paren{z,w} }{\clmref{_k-1,kW1+_k-4_W2__case_c_index}} \\
            & \leq \textubrace{ \paren{k-1}\cdot\delta_{G}\paren{u,v} + k\cdot w\paren{x,y} + \paren{k-4}\cdot w\paren{z,w} }{Simply summing and $\delta_{G}\paren{x,w} \leq \delta_{G}\paren{u,v}$} 
        \end{alignat*}

        Symmetrically, expanding \eqnref{_k-1,kW1+_k-4_W2__case_c3_w} by \clmref{_k-1,kW1+_k-4_W2__case_c_index}, we get:
        \begin{alignat*}{3}
            \delta_{H}\paren{u,v} & \leq \textubrace{ \delta_{G}\paren{u,v} + 2w\paren{x,y} + 2\cdot\paren{\frac{k-2}{2}\cdot\delta_{G}\paren{x,w} + \frac{k-6}{2}\cdot w\paren{x,y} + \frac{k}{2}\cdot w\paren{z,w}} }{\clmref{_k-1,kW1+_k-4_W2__case_c_index}} \\
            & \leq \textubrace{ \paren{k-1}\cdot\delta_{G}\paren{u,v} + \paren{k-4}\cdot w\paren{x,y} + k\cdot w\paren{z,w} }{Simply summing and $\delta_{G}\paren{x,w} \leq \delta_{G}\paren{u,v}$}
        \end{alignat*}

       As $H$ contains both $u\squiggly x \straight p_{k-2}\paren{x} \straight p_{1}\paren{w} \straight w \squiggly v$ and $u\squiggly x \straight p_{1}\paren{x} \straight p_{k-2}\paren{w} \straight w \squiggly v$,  the actual distance $\delta_{H}\paren{u,v}$ is bounded by their minimum. As only at most one among $\paren{x,y},\paren{z,w}$ can be the heaviest edge (unless they have the same weight), it follows that:
        \begin{align}~\label{eqn:_k-1,kW1+_k-4_W2__case_c3_k=2}
            \delta_{H}\paren{u,v} &\leq \paren{k-1}\cdot\delta_{G}\paren{u,v} + \paren{k-4}\cdot W_{1}\paren{u\squiggly v} + k\cdot W_{2}\paren{u\squiggly v}
        \end{align} 

    \item \dashuline{$k = 0 \pmod 4$}\\
        Hence $k-2 = 2 \pmod 4$. In this case, expanding \eqnref{_k-1,kW1+_k-4_W2__case_c3_x} by \clmref{_k-1,kW1+_k-4_W2__case_c_index} we get:
        \begin{alignat*}{3}
            \delta_{H}\paren{u,v} & \leq \textubrace{ \delta_{G}\paren{u,v} + 2\cdot\paren{\frac{k-2}{2}\cdot\delta_{G}\paren{x,w} + \frac{k-4}{2}\cdot w\paren{x,y} + \frac{k-2}{2}\cdot w\paren{z,w}} + 2w\paren{z,w} }{\clmref{_k-1,kW1+_k-4_W2__case_c_index}} \\
            & = \textubrace{  \paren{k-1}\cdot\delta_{G}\paren{u,v} + \paren{k-4}\cdot w\paren{x,y} + k\cdot w\paren{z,w} }{Simply summing and $\delta_{G}\paren{x,w} \leq \delta_{G}\paren{u,v}$} 
        \end{alignat*}

        Symmetrically, expanding \eqnref{_k-1,kW1+_k-4_W2__case_c3_w} by \clmref{_k-1,kW1+_k-4_W2__case_c_index} bounds the distance through $p_{k-2}\paren{w}$:
        \begin{alignat*}{3}
            \delta_{H}\paren{u,v} & \leq \textubrace{ \delta_{G}\paren{u,v} + 2w\paren{x,y} + 2\cdot\paren{\frac{k-2}{2}\cdot\delta_{G}\paren{x,w} + \frac{k-2}{2}\cdot w\paren{x,y} + \frac{k-4}{2}\cdot w\paren{z,w}} }{\clmref{_k-1,kW1+_k-4_W2__case_c_index} for $w$ at $2 \pmod 4$} \\
            & = \textubrace{  \paren{k-1}\cdot\delta_{G}\paren{u,v} + k\cdot w\paren{x,y} + \paren{k-4}\cdot w\paren{z,w} }{Simply summing and $\delta_{G}\paren{x,w} \leq \delta_{G}\paren{u,v}$}
        \end{alignat*}

        Similarly to \eqnref{_k-1,kW1+_k-4_W2__case_c3_k=2}, $H$ contains both paths, so $\delta_{H}\paren{u,v}$ is bounded by their minimum. As only at most one among $\paren{x,y},\paren{z,w}$ is the heaviest edge on $u\squiggly v $ (unless $w\paren{x,y}=w\paren{z,w}$), it follows that:
        \begin{align}~\label{eqn:_k-1,kW1+_k-4_W2__case_c3_k=0}
            \delta_{H}\paren{u,v} 
            &\leq \paren{k-1}\cdot\delta_{G}\paren{u,v} + \paren{k-4}\cdot W_{1}\paren{u\squiggly v} + k\cdot W_{2}\paren{u\squiggly v}
        \end{align}

        It should be observed that \eqnref{_k-1,kW1+_k-4_W2__case_c3_k=2} and \eqnref{_k-1,kW1+_k-4_W2__case_c3_k=0} together cover the case where $k=0\pmod 2$. It is therefore left only to consider the case where $k=1\mod 2$.
    
        \item \dashuline{$k = 1 \pmod 2$}\\
        Hence $k-2 = 1 \pmod 2$ as well. In this case, we only need to expand one of  \eqnref{_k-1,kW1+_k-4_W2__case_c3_x} or  \eqnref{_k-1,kW1+_k-4_W2__case_c3_w}. By \clmref{_k-1,kW1+_k-4_W2__case_c_index} we get:
        
        \begin{alignat*}{3}
            \delta_{H}\paren{u,v} & \leq \textubrace{ \delta_{G}\paren{u,v} + 2\cdot\paren{\frac{k-3}{2}\cdot\delta_{G}\paren{x,w} + \frac{k-1}{2}\cdot w\paren{x,y} + \frac{k-3}{2}\cdot w\paren{z,w}} + 2w\paren{z,w} }{\clmref{_k-1,kW1+_k-4_W2__case_c_index}} \\
            & = \textubrace{ \delta_{G}\paren{u,v} + \paren{k-3}\cdot\delta_{G}\paren{x,w} + \paren{k-1}\cdot w\paren{x,y} + \paren{k-1}\cdot w\paren{z,w} }{Simply summing} \\
            & \leq \textubrace{ \paren{k-2}\cdot\delta_{G}\paren{u,v} + \paren{k-1}\cdot w\paren{x,y} + \paren{k-1}\cdot w\paren{z,w} }{$\delta_{G}\paren{x,w} \leq \delta_{G}\paren{u,v}$} \\
            & \leq \textubrace{ \paren{k-2}\cdot\delta_{G}\paren{u,v} + \paren{k-1}\cdot\paren{W_{1}\paren{u\squiggly v} + W_{2} \paren{u\squiggly v}} }{$w\paren{x,y}, w\paren{z,w} \leq W_{1}, W_{2}$}
        \end{alignat*}
        
    \end{enumerate}
    
\end{enumerate}

This covers all possible subcases \case{c}{1} to \case{c}{3}, hence we have completed the proof for \case{c}.

\end{proof}

Having established the distance bounds across \case{a} through \case{c}, we can now deduce the overall worst case stretch of the emulator. Since \case{a} guarantees the exact shortest path distance, it is dominated by \case{b} and \case{c}. By taking the maximum over these two remaining configurations, we obtain a unified stretch for all cases. To do so, we consider the parity of $k$.

\begin{theorem}~\label{thm:_k-1,kW1+_k-4_W2_} 
   \algref{_k-1,kW1+_k-4_W2_} constructs a $\paren{2\cdot \floor{\frac{k}{2}} -1, 2\cdot \ceil{\frac{k}{2}}\cdot W_{1} +\mmax{0,2\cdot\paren{\ceil{\frac{k}{2}}-2}}\cdot W_{2}}$-emulator of expected size $\tilde O\paren{n^{1+\frac{1}{k}}}$.
\end{theorem}
\begin{proof}
    We now establish an upper bound for the expected size of $F$. Observe that the probability $q_{i}$ of a vertex from $V$ to enter $S_{i} $  is $n^{-\sum_{j=1}^{i}{\beta_j}}$. 
     By \obvref{hssize2}, each hitting set $S_i$ in the hierarchy satisfies $\abs{S_i} = \tilde O\paren{n^{1 - \sum_{j=1}^{i}{\beta_j}}}$. The total size of the emulator is bounded by the size of the largest component:

   \begin{enumerate}
        \item \dashuline{$D$}: \\
        The set $D$ contributes $\abs{D} = \tilde O\paren{n}$ edges by \obvref{edgesToPivots}. 
   
        \item \dashuline{ $E_{1}$}: \\
        The set $E_{1}$ is of size  $\abs{E_{1}} = \tilde O\paren{n^{1 + \beta_1}} $  by \obvref{edgessize}.
    
        \item \dashuline{$B_1$ and $B_2$}: \\
        By \obvref{edgestobunches2}, the contributions of the bunch edges are $\abs{B_{1}\paren{V}} = \tilde O\paren{\msum{n^{1 + \beta_{i}}}{i=1}{k-1}}$ and of the second-order bunch edges:
        
        $$\abs{B_{2}\paren{S_{1}}} = \tilde O\paren{\msum{\abs{S_1} \cdot n^{\beta_{i} + \beta_{i+1}}}{i=1}{k-2}} = \tilde O\paren{\msum{n^{1 - \beta_1} \cdot n^{\beta_{i} + \beta_{i+1}}}{i=1}{k-2}} = \tilde O\paren{\msum{n^{1 + \beta_{i}+\beta_{i+1}-\beta_{1}}}{i=1}{k-2}}$$ 
    
        \item \dashuline{Product sets $S_{i-1} \times S_{k-i}$}: \\
        For each $i \in \bracke{k}$, the emulator includes the cartesian product $S_{i-1} \times S_{k-i}$. By \obvref{hssize2}, the expected size of each product is:
       $$\abs{S_{i-1} \times S_{k-i}} = \tilde O\paren{n^{1 - \sum_{j=1}^{i-1} \beta_j} \cdot n^{1 - \sum_{j=1}^{k-i} \beta_j}} \\
            = \tilde O\paren{n^{2 - \paren{\sum_{j=1}^{i-1} \beta_j + \sum_{j=1}^{k-i} \beta_j}}}$$
    \end{enumerate}
    
   To minimize $\abs{F}$, we balance the sizes of all edge sets (excluding the linear set $D$). The parameters $\beta_{i}$ govern a trade-off: increasing $\beta_{i}$ expands $E_{1}$ and the bunches, but dually reduces the hitting sets $S_{i}$ that determine the density of the cartesian products. To prevent any component from becoming a bottleneck, we equate the exponents by setting $\beta_{1} = \beta_{2} =  \ldots = \beta_{k-1} = \frac{1}{k}$, ensuring every set in $F$ satisfies the $\tilde{O}\paren{n^{1+\frac{1}{k}}}$ bound. That is: 

   $$\begin{cases} 
        \abs{E_{1}}, \abs{B_{1}\paren{V}} &= \tilde{O}\paren{n^{1 + \beta_{1}}} = \tilde{O}\paren{n^{1 + \frac{1}{k}}} \\
        \abs{B_{2}\paren{S_{1}}} &= \tilde{O}\paren{n^{1 + \beta_{i} + \beta_{i+1} - \beta_{1}}} = \tilde{O}\paren{n^{1 + \frac{1}{k} + \frac{1}{k} - \frac{1}{k}}} = \tilde{O}\paren{n^{1 + \frac{1}{k}}} \\
        \abs{S_{i-1} \times S_{k-i}} &= \tilde{O}\paren{n^{2 - \paren{\frac{i-1}{k} + \frac{k-i}{k}}}} = \tilde{O}\paren{n^{2 - \frac{k-1}{k}}} = \tilde{O}\paren{n^{1 + \frac{1}{k}}}
    \end{cases}$$
   
    It is now left to establish the approximation guarantee of $H$. Let $u,v\in V$. We first remark that \case{a} is dominated by both \case{b} and \case{c} (See: \lemref{_k-1,kW1+_k-4_W2__case_a}). For \case{b}, we know due to \lemref{_k-1,kW1+_k-4_W2__case_b} that $\delta_{H}\paren{u,v} \leq \delta_{G}\paren{u,v} + 2\cdot\paren{k-1}\cdot W_{1}\paren{u\squiggly v}$. For \case{c}, we consider the cases where $k$ is an even and when $k$ is odd. 
    
    When $k = 0 \pmod 2$, by \lemref{_k-1,kW1+_k-4_W2__case_c} we get that $\delta_{H}\paren{u,v} \leq \paren{k-1}\cdot\delta_{G}\paren{u,v} + \paren{k-4}\cdot W_{1}\paren{u\squiggly v} + k\cdot W_{2}\paren{u\squiggly v}$. This bound also covers \case{b} by shifting weight from the additive stretch to the multiplicative stretch. Since $W_{1}\paren{u\squiggly v} + W_{2}\paren{u\squiggly v}\leq \delta_{G}\paren{u,v}$, we have:
    \begin{alignat*}{3}
        \delta_{H}\paren{u,v} & \leq \textubrace{ \delta_{G}\paren{u,v} + 2\cdot \paren{k-1}\cdot W_{1}\paren{u\squiggly v} }{\lemref{_k-1,kW1+_k-4_W2__case_b}}  \\
        & \leq \textubrace{ \paren{k-1}\cdot \delta_{G}\paren{u,v} + k\cdot W_{1}\paren{u\squiggly v} }{$W_{1}\paren{u\squiggly v} \leq \delta_{G}\paren{u,v}$} \\
        & < \textubrace{ \paren{k-1}\cdot \delta_{G}\paren{u,v} + k\cdot W_{1}\paren{u\squiggly v} +\paren{k-4}\cdot W_{2}\paren{u\squiggly v} }{Assuming $k\geq 4$}
    \end{alignat*}

    Additionally, we have: 
    \begin{alignat*}{3}
        \delta_{H}\paren{u,v} & \leq \textubrace{ \paren{k-1}\cdot\delta_{G}\paren{u,v} + \paren{k-4}\cdot W_{1}\paren{u\squiggly v}+ k\cdot W_{2}\paren{u\squiggly v} }{\lemref{_k-1,kW1+_k-4_W2__case_c}}  \\
        & \leq \textubrace{ \paren{k-1}\cdot\delta_{G}\paren{u,v} + k\cdot W_{1}\paren{u\squiggly v}+ \paren{k-4}\cdot W_{2}\paren{u\squiggly v} }{$W_{2}\paren{u\squiggly v}\leq W_{1}\paren{u\squiggly v}$}
    \end{alignat*}

Combining both we get that when $k=0\pmod 2$:

\begin{equation}~\label{eqn:_k-1,kW1+_k-4_W2__k=0}
    \delta_{H}\paren{u,v} \leq \paren{k-1}\cdot\delta_{G}\paren{u,v} + k\cdot W_{1}\paren{u\squiggly v}+ \paren{k-4}\cdot W_{2}\paren{u\squiggly v}
\end{equation}

When $k = 1 \pmod 2$, by \lemref{_k-1,kW1+_k-4_W2__case_c} we have that $\delta_{H}\paren{u,v} \leq \paren{k-2}\cdot\delta_{G}\paren{u,v} + \paren{k-1}\cdot \paren{W_{1}\paren{u\squiggly v} + W_{2}\paren{u\squiggly v}}$. Similarly, for \case{b}, we shift weights to align with the multiplicative stretch of \case{c}:
\begin{alignat*}{3}
    \delta_{H}\paren{u,v} & \leq \textubrace{ \delta_{G}\paren{u,v} + 2\cdot \paren{k-1}\cdot W_{1}\paren{u\squiggly v} }{\lemref{_k-1,kW1+_k-4_W2__case_b}}  \\
    & \leq \textubrace{ \paren{k-2}\cdot \delta_{G}\paren{u,v} + \paren{k+1}\cdot W_{1}\paren{u\squiggly v} }{$W_{1}\paren{u\squiggly v} \leq \delta_{G}\paren{u,v}$} \\
    & < \textubrace{ \paren{k-2}\cdot \delta_{G}\paren{u,v} + \paren{k+1}\cdot W_{1}\paren{u\squiggly v} +\paren{k-3}\cdot W_{2}\paren{u\squiggly v} }{Assuming $k\geq 3$}
\end{alignat*}

Additionally, for \case{c} when $k=1\pmod 2$:
\begin{alignat*}{3}
    \delta_{H}\paren{u,v} & \leq \textubrace{ \paren{k-2}\cdot\delta_{G}\paren{u,v} + \paren{k-1}\cdot W_{1}\paren{u\squiggly v}+ \paren{k-1}\cdot W_{2}\paren{u\squiggly v} }{\lemref{_k-1,kW1+_k-4_W2__case_c}}  \\
    & \leq \textubrace{ \paren{k-2}\cdot\delta_{G}\paren{u,v} + \paren{k+1}\cdot W_{1}\paren{u\squiggly v}+ \paren{k-3}\cdot W_{2}\paren{u\squiggly v} }{$W_{2}\paren{u\squiggly v}\leq W_{1}\paren{u\squiggly v}$}
\end{alignat*}

Combining both we get that when $k=1 \pmod 2$:
\begin{equation}~\label{eqn:_k-1,kW1+_k-4_W2__k=1}
    \delta_{H}\paren{u,v} \leq \paren{k-2}\cdot\delta_{G}\paren{u,v} + \paren{k+1}\cdot W_{1}\paren{u\squiggly v}+ \paren{k-3}\cdot W_{2}\paren{u\squiggly v}
\end{equation}

By combining the bounds derived in \eqnref{_k-1,kW1+_k-4_W2__k=0} for $k=0 \pmod 2$ and in  \eqnref{_k-1,kW1+_k-4_W2__k=1} for $k=1\pmod 2$, we observe that, regardless of the parity of $k$, it holds that:
$$
    \delta_{H}\paren{u,v} \leq \paren{2\cdot \floor{\frac{k}{2}} -1}\cdot\delta_G \paren{u,v} + 2\cdot \ceil{\frac{k}{2}}\cdot W_{1}\paren{u\squiggly v} +2\cdot\paren{\ceil{\frac{k}{2}}-2}\cdot W_{2}\paren{u\squiggly v}
$$

Taking into consideration $k=2$ and $k=3$, we get:

$$
    \delta_{H}\paren{u,v} \leq \paren{2\cdot \floor{\frac{k}{2}} -1}\cdot\delta_G \paren{u,v} + 2\cdot \ceil{\frac{k}{2}}\cdot W_{1}\paren{u\squiggly v} +\mmax{0,2\cdot\paren{\ceil{\frac{k}{2}}-2}}\cdot W_{2}\paren{u\squiggly v}
$$

This concludes the correctness of \algref{_k-1,kW1+_k-4_W2_}.
\end{proof}

\begin{remark}\label{rem:refined_stretch}
The proof of \thmref{_k-1,kW1+_k-4_W2_} combines the
 bounds of
\lemref{_k-1,kW1+_k-4_W2__case_b} and
\lemref{_k-1,kW1+_k-4_W2__case_c}
into a simple united stretch guarantee.
In \appref{refined}, we retain the distinction between
\case{b} and \case{c} and derive a strictly tighter 
stretch guarantee for the same emulator construction.
\end{remark}

\section{Improvement for unweighted graphs}\label{_k-1,2k-4_}

In this section, we consider the unweighted setting and compare our construction for unweighted graphs
with the seminal emulator of Thorup and Zwick \cite{ThoZwi2006, HuaPet2019B}. In an unweighted graph $w\paren{e} = 1$ for all $e \in E$. This implies that $W_{1}\paren{u \squiggly v} = W_{2}\paren{u \squiggly v} = 1$ for any pair $u,v \in V$, regardless of $u\squiggly v$. Therefore, our construction from \secref{_k-1,kW1+_k-4_W2_} yields a $\paren{2\cdot \floor{\frac{k}{2}} - 1, 2\cdot \ceil{\frac{k}{2}} + \max\{0, 2\cdot\paren{\ceil{\frac{k}{2}} - 2}\}}$-emulator of size $\tilde{O}\paren{n^{1+\frac{1}{k}}}$. This is a $\paren{k-1, 2k-4}$-emulator for an even $k$, and a $\paren{k-2, 2k-2}$-emulator for an odd $k$.

To compare it to the $+ \paren{6^{k} -1} \cdot \paren{\delta_{G}\paren{u,v}}^{1-\frac{1}{k}}$-emulator of Thorup and Zwick \cite{ThoZwi2006, HuaPet2019B}, we parameterize our construction to match their $O\paren{n^{1+\frac{1}{2^{k+1}-1}}}$ size. By setting our hierarchy depth to the odd integer $2^{k+1}-1$, we obtain a $\paren{2^{k+1}-3, 2^{k+2}-4}$-emulator of size $\tilde{O}\paren{n^{1+\frac{1}{2^{k+1}-1}}}$. While their stretch is \textbf{asymptotically} superior, we demonstrate that our construction yields a strictly tighter approximation for all vertex pairs up to an exponentially large distance threshold of $\delta_{G}\paren{u,v} \leq \tilde{O}\paren{3^{k^{2}}}$. For a precise list of values, see examples in \tabref{_k-1,2k-4__vsThorupZwick}.

\begin{table}[H]
  \begin{center}
    
\small\addtolength{\tabcolsep}{-3pt}
    \begin{tabular}{|c|c|c|} 
      \hline
      \textbf{\small Hierarchy Depth $k$} & \textbf{\small Root $r_k\approx$} & \textbf{\small Distance Threshold $\delta_{G}\paren{u,v} \leq$}\\
      \hhline{|=|=|=|}
      &&\\ [-1em]
      $2$ & $8.3925397$ & $70$ \\
      &&\\ [-1em]
      \hline

    &&\\ [-1em]
    $3$ & $17.9097233$ & $5,744$ \\
    &&\\ [-1em]
    \hline

    &&\\ [-1em]
    $4$ & $46.2499783$ & $4,575,579$ \\
    &&\\ [-1em]
    \hline

    &&\\ [-1em]
    $5$ & $129.5833333$ & $36,538,081,150$ \\
    &&\\ [-1em]
    \hline

    &&\\ [-1em]
    $6$ & $376.25$ & $2,836,998,143,584,445\textbf{}$ \\
    &&\\ [-1em]
    \hline
    
    &&\\ [-1em]
    $7$ & $1110.8531746$ & $2,087,356,475,107,196,611,113$ \\
    \hline
    \end{tabular}
    \caption{Different values of $k\in\mathbb{N}$ and the regime of dominance for our $\paren{2^{k+1}-3, 2^{k+2}-4}$-emulator over the  Thorup and Zwick $+\paren{6^{k} -1} \cdot \delta_{G}^{1-\frac{1}{k}}$-emulator \cite{ThoZwi2001}.}~\label{tab:_k-1,2k-4__vsThorupZwick}
  \end{center}
\end{table}

To find these distance thresholds, we consider when:

\begin{equation}~\label{eqn:_k-1,2k-4_}
    \paren{2^{k+1}-3} \cdot \delta_{G}\paren{u,v} + 2^{k+2}-4 \leq \delta_{G}\paren{u,v} + \paren{6^{k}-1} \cdot \paren{\delta_{G}\paren{u,v}}^{1-\frac{1}{k}}
\end{equation}

Subtracting $\delta_{G}\paren{u,v}$ from both sides and applying the substitution $x = \paren{\delta_{G}\paren{u,v}}^{\frac{1}{k}}$ reduces Equation \ref{eqn:_k-1,2k-4_} to a polynomial inequality. We bound the valid range for this inequality in the following lemma:

\begin{lemma} \label{lem:_k-1,2k-4__root}
Let $f_{k}\paren{x} = \paren{2^{k+1}-4} \cdot x^{k} - \paren{6^{k}-1} \cdot x^{k-1} + 2^{k+2}-4$. For any integer $k \geq 2$, it holds that $f_{k}\paren{x} <0$ for any $x \in \paren{1, r_{k}}$, where $r_{k} \in \paren{\frac{\paren{k-1}\cdot\paren{6^{k}-1}}{k\cdot\paren{2^{k+1}-4}}, \frac{6^{k}-1}{2^{k+1}-4}}$ is the unique real root of $f_{k}\paren{x}$.
\end{lemma}

\begin{proof}
Our proof goes as follows: first, we identify the unique minimum point $x_{k} = \frac{\paren{k-1}\cdot\paren{6^{k}-1}}{k\cdot\paren{2^{k+1}-4}}>1$ of $f_{k}\paren{x}$ and establish that $f_{k}\paren{1} < 0$. Because $f_{k}\paren{x}$ is a continuous polynomial, this guarantees the function remains negative throughout the interval $\paren{1, x_{k}}$. To extend this range and find $r_{k}$, we evaluate $\hat{x}_{k} = \frac{6^{k}-1}{2^{k+1}-4}$ to show that $f_{k}\paren{\hat{x}_{k}} > 0$. It then strictly follows that a unique root $r_{k} \in \paren{x_{k}, \hat{x}_{k}}$ exists where $f_{k}\paren{r_{k}} = 0$, guaranteeing $f_{k}\paren{x} < 0$ for all $x \in \paren{1, r_{k}}$.

We start by evaluating the function at $x = 1$:

\begin{alignat*}{2}
    f_{k}\paren{1} &= \textubrace{\paren{2^{k+1}-4} - \paren{6^{k}-1} + 2^{k+2}-4}{Definition of $f_{k}\paren{x}$} &&= \textubrace{2^{k+1} + 2^{k+2} - 6^{k} - 7}{Simply summing} \\
    &= \textubrace{2^{k} \cdot \paren{6 - 3^{k}} - 7}{Rearranging terms} &&< \textubrace{0}{$k \geq 2$} 
\end{alignat*}

The  derivative is $f'_{k}\paren{x} = x^{k-2} \cdot \paren{k\cdot\paren{2^{k+1}-4} \cdot x - \paren{k-1}\cdot\paren{6^{k}-1}}$. Since $x \geq 1$, it follows that $x^{k-2} \neq 0$. Therefore, setting $f'_{k}\paren{x} = 0$ yields a unique extremum point:
$$x_{k} = \frac{\paren{k-1}\cdot\paren{6^{k}-1}}{k\cdot\paren{2^{k+1}-4}}$$

To ensure $x_{k}> 1$, we bound it from below. Because $\frac{k-1}{k} \geq \frac{1}{2}$ for $k \geq 3$, we obtain:
\begin{alignat*}{2}
    x_{k} &\geq \frac{1}{2} \cdot \frac{6^{k}-1}{2^{k+1}-4} &&= \frac{6^{k}-1}{4\cdot\paren{2^{k}-2}} = \frac{2^{k}\cdot 3^{k}-1}{4\cdot\paren{2^{k}-2}} \\
    &= \frac{1}{4} \cdot \frac{3^{k} - \frac{1}{2^{k}}}{1 - \frac{1}{2^{k-1}}} &&> \frac{1}{4} \cdot \frac{27 - \frac{1}{8}}{1} > 1
\end{alignat*}

We analyze the sign of the sign of the derivative $f'_{k}\paren{x}$ around the extremum $x_{k}$. Because $x^{k-2} > 0$ for all $x > 1$, the sign of $f'_{k}\paren{x}$ is entirely determined by its linear factor: $k\cdot\paren{2^{k+1}-4} \cdot x - \paren{k-1}\cdot\paren{6^{k}-1}$. For $x < x_{k}$, this linear factor evaluates to a negative value, yielding $f'_{k}\paren{x} < 0$, and for $x > x_{k}$, this factor evaluates to a positive value, yielding $f'_{k}\paren{x} > 0$. Hence, $x_{k}$ is the unique  minimum point for $x\geq 1$. 

Next, we evaluate, for $k\geq 2$, the function at the upper threshold value $\hat{x}_{k} = \frac{6^{k}-1}{2^{k+1}-4}$:
\begin{alignat*}{2}
    f_{k}\paren{\hat{x}_{k}} &= \paren{2^{k+1}-4} \cdot \paren{\frac{6^{k}-1}{2^{k+1}-4}}^{k} - \paren{6^{k}-1} \cdot \paren{\frac{6^{k}-1}{2^{k+1}-4}}^{k-1} + 2^{k+2}-4 && \\
    &= \frac{\paren{6^{k}-1}^{k}}{\paren{2^{k+1}-4}^{k-1}} - \frac{\paren{6^{k}-1}^{k}}{\paren{2^{k+1}-4}^{k-1}} + 2^{k+2} - 4 && \\
    &= 2^{k+2} - 4 > 0 &&
\end{alignat*}

Observe that $x_{k} = \frac{k-1}{k} \cdot \hat{x}_{k} < \hat{x}_{k}$. As $f_{k}\paren{x}$ is continuous,  decreases over $\paren{1,x_{k}}$,  increases over $\paren{x_{k},\infty}$ and satisfies $f_{k}\paren{\hat{x}_{k}}>0$, it follows that a unique real root  $r_{k} \in \paren{x_{k}, \hat{x}_{k}}$ of $f_{k}\paren{x}$ exists. Consequently, for all $x \in \paren{1, r_{k}}$, the polynomial satisfies $f_{k}\paren{x} < 0$.

\end{proof}

By \lemref{_k-1,2k-4__root}, $f_{k}\paren{x} < 0$ for all $x \in \paren{1, r_{k}}$. By reverting the substitution $x = \paren{\delta_{G}\paren{u,v}}^{\frac{1}{k}}$, it follows that \eqnref{_k-1,2k-4_} holds for any distance $\delta_{G}\paren{u,v} \leq \floor{\paren{r_{k}}^{k}}$. We formalize this:

\begin{theorem}~\label{thm:_k-1,2k-4_}
Let $G = \paren{V,E}$ be an unweighted, undirected graph, and let $k \geq 2$ be an integer. \algref{_k-1,kW1+_k-4_W2_} computes a $\paren{2^{k+1}-3, 2^{k+2}-4}$-emulator of size $\tilde{O}\paren{n^{1+\frac{1}{2^{k+1}-1}}}$. Furthermore, for any pair of vertices $u, v \in V$ whose distance satisfies
$$\delta_{G}\paren{u,v} \leq \paren{\floor{\frac{6^{k}-1}{2^{k+1}-4}}}^{k} \in O\paren{3^{k^{2}}}$$
this emulator provides a strictly tighter distance approximation than the $+ \paren{6^{k}-1} \cdot \paren{\delta_{G}\paren{u,v}}^{1-\frac{1}{k}}$-emulator of Thorup and Zwick \cite{ThoZwi2006}.
\end{theorem}

\begin{proof}
The correctness and size of the emulator follow directly from \thmref{_k-1,kW1+_k-4_W2_} and our preceding discussion about unweighted graphs. 

By \lemref{_k-1,2k-4__root}, $f_{k}\paren{x} < 0$ for all $x \in \paren{1, r_{k}}$. Note that $\floor{\hat{x}_{k}} \leq r_{k}$ for $\hat{x}_{k} = \frac{6^{k}-1}{2^{k+1}-4}$, hence the substitution $x = \paren{\delta_{G}\paren{u,v}}^{\frac{1}{k}}$ yields $f_{k}\paren{x} < 0$ for any distance $\delta_{G}\paren{u,v} \leq \paren{\floor{\hat{x}_{k}}}^{k}$, satisfying \eqnref{_k-1,2k-4_}.
\end{proof}

\begin{figure}[H]
\centering

\begin{tikzpicture}
\begin{semilogyaxis}[
        width=12cm, height=9cm,
        xlabel={Hierarchy Depth ($k$)},
        ylabel={Distance Threshold $\delta_{G}\paren{u,v} \leq$},
        xmin=1.5, xmax=7.5,
        ymin=10, ymax=1e24, 
        xtick={2,3,4,5,6,7},
        ytick={70, 5744, 4575462, 3.6538e10, 2.8369e15, 2.0873e21},
        yticklabels={
            $70$, 
            $5\text{,}744$, 
            $4.57 \times 10^6$, 
            $3.65 \times 10^{10}$, 
            $2.84 \times 10^{15}$, 
            $2.09 \times 10^{21}$
        },
        grid=both,
        grid style={dashed, gray!30},
        thick,
        title={Regime of Dominance over Thorup and Zwick emulator},
        legend pos=north west
    ]
    
    \addplot [
        color=blue,
        mark=*,
        mark options={scale=1.2, fill=white, thick},
        very thick
    ] coordinates {
        (2, 70)
        (3, 5744)
        (4, 4575462)
        (5, 3.6538e10)
        (6, 2.8369e15)
        (7, 2.0873e21)
    };
    \addlegendentry{Distance threshold $\delta_{G}\paren{u,v}\leq$}
    
    \addplot[fill=blue, fill opacity=0.08, draw=none] coordinates {
        (2, 70) 
        (3, 5744) 
        (4, 4575462) 
        (5, 3.6538e10) 
        (6, 2.8369e15) 
        (7, 2.0873e21)
        (7.5, 2.0873e21) 
        (7.5, 1)        
        (1.5, 1)        
        (1.5, 70)       
        (2, 70)         
    } \closedcycle;
    
    \end{semilogyaxis}
\end{tikzpicture}

\caption[.]{\label{fig:_k-1,2k-4__vsThorupZwick} The regime of dominance for the $\paren{2^{k+1}-3, 2^{k+2}-4}$-emulator. The \qoute{$x$} axis represents the hierarchy depth $k\in\mathbb{N}$, and the \qoute{$y$} axis represents the maximum graph distance $\delta_{G}\paren{u,v}$ for which our construction yields a strictly tighter approximation than the Thorup and Zwick emulator \cite{ThoZwi2001}.}
\vspace{1mm}
\end{figure}

For a visualization of this exponential regime of dominance, we plot the maximum distance threshold $\delta_{G}\paren{u,v}$ as a function of the hierarchy depth $k$ in \figref{_k-1,2k-4__vsThorupZwick}.

\newpage
\clearpage
\renewcommand*{\bibfont}{\small}
\begingroup
\phantomsection\printbibliography[heading=bibintoc]
\endgroup

\appendix

\section{Runtime Efficient Construction}~\label{runtime}
In this section, we present \algref{_k-1,kW1+_k-4_W2__runtime}, which is a runtime efficient algorithm that essentially constructs the same $\paren{2\cdot \floor{\frac{k}{2}} -1, 2\cdot \ceil{\frac{k}{2}}\cdot W_{1} +\mmax{0,2\cdot\paren{\ceil{\frac{k}{2}}-2}}\cdot W_{2}}$-emulator as \algref{_k-1,kW1+_k-4_W2_} from \secref{_k-1,kW1+_k-4_W2_}. It has an expected runtime of $\tilde O\paren{n^{2+\frac{2}{k}}}$.
It utilizes an approach similar to Baswana and Kavitha \cite{BasKav2010}, yet is much simpler, for our usages. 

\algref{_k-1,kW1+_k-4_W2__runtime} differs from \algref{_k-1,kW1+_k-4_W2_} by not having the exact distances $\delta_{G}\paren{u,v}$ apriori. To overcome this, \algref{_k-1,kW1+_k-4_W2__runtime} must compute these distances. This affects the computation of the sets that compose $F$ as well. The major bottleneck is $B_{2}\paren{S_{1}}$ which requires $\tilde O \paren{n^{2+\frac{2}{k}}}$ runtime. The rest can be computed within $\tilde O \paren{n^{2+\frac{1}{k}}}$ runtime. 

A key observation is that  $d\bracke{u,v} =\delta_{G}\paren{u,v}$ for certain vertices $u,v\in V$. For our usage, we would like to show that for \qoute{necessary} auxiliary edges $\paren{u,v}\in F \setminus E$ it holds that $w_{H}=\delta_{G}\paren{u,v}$. In other words, we need to show that for these vertices $d\bracke{u,v}=\delta_{G}\paren{u,v}$. 

\begin{remark}\label{rem:undirected}
As $G$ is undirected, we may assume whenever a value $d\bracke{u,v}$ is updated to hold a smaller value than it previously held, so does $d\bracke{v,u}$. This has no effect on the runtime.
\end{remark}

\begin{algorithm}[H]
\caption{runtimeEfficientGeneralEmulator$\paren{G=\paren{V,E,w},k}$}\label{alg:_k-1,kW1+_k-4_W2__runtime}

\textbf{Input:} An undirected non-negative weighted graph $G=\paren{V,E,w}$ and an integer $k\geq 2$.

\textbf{Output:} A $(2\cdot \floor{\frac{k}{2}} -1, 2\cdot \ceil{\frac{k}{2}}\cdot W_{1} +\max\{0,2\cdot(\ceil{\frac{k}{2}}-2)\}\cdot W_{2})$-emulator. 

\medskip

Let $\beta_{1},\beta_{2},\ldots, \beta_{k-1} \in \paren{ 0 , 1}$ to be fixed later

Initialize $d$ as the edges' weights and $\infty$ otherwise

Set $S_0 \leftarrow V$

\For{$i \in \bracke{k-1}$}
{
    Sample each vertex of $S_{i-1}$ with probability $n^{-\beta_i}$ to $S_{i}$
    
    For every  $u\in V$ find the pivot $p_{i}\paren{u}$ and update $d$ accordingly 

    Construct the edge set $D_{i}$

    Construct the edge set $E_{i}$
}

Set $S_{k} \leftarrow \varnothing$

Set $E_{k} \leftarrow E$

$D\leftarrow \mcup{D_{i}}{i=1}{k-1}$

\For{$i \in \bracke{k-1}\cup\bracce{0}$}
{
    \For{$s\in  S_{i}$}
    {
        Invoke SSSP from $s$ on $ G_{i}  = \paren{ V,E_{i+2}\cup D, d }$ and update $d$ accordingly
    }

    \For{$j \in \bracke{k-1}\cup\bracce{0}$}
    {
        \For{$\paren{x,y}\in E$}
        {
            $d\bracke{p_{i}\paren{x}, p_{j}\paren{y}} \leftarrow \mmin{d\bracke{x,p_{i}\paren{x}}+w\paren{x,y}+d\bracke{y,p_{j}\paren{y}}, d\bracke{p_{i}\paren{x}, p_{j}\paren{y}}}$
        }
    }
}

Construct the set of edges $B_{1}\paren{V}$

Construct the set of edges $B_{2}\paren{S_{1}}$

Set $F\leftarrow D \cup E_{1}  \cup B_{1}\paren{V} \cup B_{2}\paren{S_{1}}$

\For{$i \in \bracke{k}$}
{
    $F\leftarrow F \cup \paren{S_{i-1}\times S_{k-i}}$
}

Initialize $w_{H}\leftarrow \varnothing$

\For{$\paren{u,v}\in F\setminus E$}
    {
        $w_{H} \paren{u,v} \leftarrow d\bracke{u,v}$
    }

\For{$\paren{u,v}\in F\cap E$}
    {
        $w_{H} \paren{u,v} \leftarrow w\paren{u,v}$
    }   

\Return $H=\paren{V,F,w_{H}}$

\end{algorithm}

We start by specifying pairs of vertices $u,v\in V$ such that $\paren{u,v}\in F$ and $d\bracke{u,v} = \delta_{G}\paren{u,v}$. 

\begin{claim2}\label{clm:_k-1,kW1+_k-4_W2__runtime_case_a}
For any $\paren{u,v}\in E_{1}$ it holds that $w_{H}\paren{u,v}=w\paren{u,v}$.
\end{claim2}
\begin{proof}
    By definition we set $w_{H}\paren{u,v} = w\paren{u,v}$ for any $\paren{u,v}\in F \cap E$. As $E_{1}\subseteq F\cap E$, the claim follows.
\end{proof}

For the next lemma, we consider the case where there is an edge and we are interested in the distance estimation between pivots of its endpoints.

\begin{claim2}\label{clm:_k-1,kW1+_k-4_W2__runtime_case_b}
For any $\paren{x,y}\in E$ and $i,j\in\bracke{k-1}$ it holds that $w_{H}\paren{p_{i}\paren{x},p_{j}\paren{y}}\leq \delta_{G}\paren{x,p_{i}\paren{x}}+w\paren{x,y}+\delta_{G}\paren{y,p_{j}\paren{y}}$.
\end{claim2}
\begin{proof}
    \algref{_k-1,kW1+_k-4_W2__runtime}  iterates over all $\paren{x,y}\in E$ and by \obvref{pivotsdistance} it follows that $d\bracke{x,p_{i}\paren{x}}=\delta_{G}\paren{x,p_{i}\paren{x}}$ and $d\bracke{y,p_{j}\paren{y}}=\delta_{G}\paren{y,p_{j}\paren{y}}$. Hence, $d\bracke{p_{i}\paren{x},p_{j}\paren{y}}\leq\delta_{G}\paren{x,p_{i}\paren{x}}+w\paren{x,y}+\delta_{G}\paren{y,p_{j}\paren{y}} $ .
\end{proof}

We now consider moving through an edge from $B_{1}\paren{V}$:

\begin{claim2}\label{clm:_k-1,kW1+_k-4_W2__runtime_case_c1}
For any $\paren{s,w}\in B_{1}\paren{V}$ where $s\in S_{1}$ and $w\in V$, it follows that 
$d\bracke{s,w} =\delta_{G}\paren{s,w}$.
\end{claim2}
\begin{proof}
    By \obvref{pathinball} it follows that $s\squiggly w \subseteq E_{2}$ as $s\in S_{1}$. As we invoke SSSP from $s\in S_{1}$ including the edge set  $E_{2}$, it follows that $d\bracke{s,w} =\delta_{G}\paren{s,w}$.
\end{proof}

Furthermore, we need to consider the case where we utilize an edge from $B_{2}\paren{S_{1}}$:

\begin{claim2}\label{clm:_k-1,kW1+_k-4_W2__runtime_case_c2}
For any $\paren{s,p_{1}\paren{w}}\in B_{2}\paren{S_{1}}$ where $s\in S_{i}$ and $w\in V$, it follows that 
$d\bracke{s,p_{1}\paren{w}} =\delta_{G}\paren{s,p_{1}\paren{w}}$.
\end{claim2}
\begin{proof}
    Similarly, by \obvref{pathinball} we know that $s\squiggly p_{1}\paren{w} \subseteq E_{i+2}$ as $s\in S_{i}$. We invoke SSSP from $s\in S_{i}$ over the edge set  $E_{i+2}$. Therefore  $d\bracke{s,p_{1}\paren{w}} =\delta_{G}\paren{s,p_{1}\paren{w}}$.
\end{proof}

Lastly, we consider the case where either $u\squiggly v \subseteq E_{k-1}$ or there exists an edge with an appropriate pivot.

\begin{claim2}\label{clm:_k-1,kW1+_k-4_W2__runtime_case_c3}
Either $u\squiggly v\subseteq E_{k-1}$  or there exists an edge $\paren{q,r}\in u\squiggly v $ such that $\paren{q,r}\notin E_{k-1}$. 
\end{claim2}
\begin{proof}
    Let us assume $u\squiggly v\not\subseteq E_{k-2}$. It follows in a straightforward manner that there exists an edge $\paren{q,r}\in u\squiggly v$ such that $\paren{q,r}\notin E_{k-1}$. 
\end{proof}

We now turn to argue about the runtime for constructing each set that composes $F$.

\begin{lemma}\label{lem:_k-1,kW1+_k-4_W2__runtime_F}
\algref{_k-1,kW1+_k-4_W2__runtime} requires $\tilde O \paren{n^{2+\frac{2}{k}}}$ runtime.
\end{lemma}
\begin{proof}
    Let $i\in\bracke{k-1}$. By \obvref{pivotsdistance} we can find $p_{i}\paren{u}$ for all $u\in V$ in $\tilde O\paren{m}$. Also, $d\bracke{u,p_{i}\paren{u}}=\delta_{G}\paren{u,p_{i}\paren{u}}$. We construct the set $E_{i}\paren{u}$ for any $u\in V$ in $\tilde O \paren{m}$ time by simply checking whether $w\paren{u,v}<d\bracke{u,v}$  for any neighbour $v$ of $u$. Hence, constructing $E_{i}$ would require  $\tilde O \paren{m}$ overall runtime. We assume $k\in \tilde O\paren{1}$, hence constructing all of the sets $E_{1},\ldots ,E_{k-1}$ and $D$ requires the same $\tilde O\paren{m}$ runtime.

    To construct the set $B_{1}\paren{V}$, we iterate over any $\paren{u,v}\in E$ and $i\in \bracke{k-1}$. If $v\in S_{i}$ we check whether $d\bracke{u,v} < d\bracke{u,p_{i+1}\paren{u}}$. Recall that $d\bracke{u,p_{i+1}\paren{u}} = \delta_{G}\paren{u,p_{i+1}\paren{u}}$ by \obvref{pivotsdistance}. As for $d\bracke{u,v}$, if $v\in \ball{u}{i}{i+1}$, it follows by \obvref{pathinball} that $u\squiggly v \subseteq E_{i+1}$. Hence, when \algref{_k-1,kW1+_k-4_W2__runtime} invokes SSSP from $v\in S_{i}$ over the edge set $E_{i+1}$, the path $u\squiggly v$ will be taken into consideration as well. In plain words, $d\bracke{u,v}=\delta_{G}\paren{u,v}$. Hence, checking whether $d\bracke{u,v} < d\bracke{u,p_{i+1}\paren{u}}$ is equivalent to checking whether $\delta_{G}\paren{u,v}<\delta_{G}\paren{u,p_{i+1}\paren{u}}$. Therefore, the set $B_{1}\paren{V}$ is properly constructed this way.

    A similar argument occurs for $B_{2}\paren{S_{1}}$. To check whether $\paren{u,v}\in B_{2}\paren{S_{1}}$, we need to check whether there exists an index $i\in \bracke{k-2}$ such that $v\in \ball{u}{i}{i+2}$. That is, whether $v\in S_{i}$ and $\delta_{G}\paren{u,v} < \delta_{G}\paren{u,p_{i+2}\paren{u}}$. We recall that $d\bracke{u,p_{i+2}\paren{u}}= \delta_{G}\paren{u,p_{i+2}\paren{u}}$ by \obvref{pivotsdistance}. Furthermore, if $v\in \ball{u}{i}{i+2}$ then by \obvref{pathinball} it follows that $u\squiggly v \subseteq E_{i+2}$. Hence, when invoking SSSP from $v\in S_{i}$ over $E_{i+2}$ results in the exact distance $d\bracke{u,v} = \delta_{G}\paren{u,v}$. 
    
    By \obvref{hssize} it follows that $\abs{S_{i-1}\times S_{k-i}}=\tilde O \paren{n^{1+\frac{1}{k}}}$, hence 
    iterating over $i\in\bracke{k}$ and accumulating the edges into $F$ would still take  $O \paren{n^{1+\frac{1}{k}}}$ runtime. 

    Finally, iterating over $\paren{u,v}\in F$ and setting $w_{H}$ would take $\tilde O\paren{n^{1+\frac{1}{k}}}$ runtime. 
\end{proof}

We may now conclude:

\begin{theorem}\label{thm:_k-1,kW1+_k-4_W2__runtime}
    Let $k\in\mathbb{N}$ such that $k\geq 2$. \algref{_k-1,kW1+_k-4_W2__runtime} requires $\tilde O\paren{n^{2+\frac{2}{k}}}$ runtime to compute a  $\paren{2\cdot \floor{\frac{k}{2}}-1, 2\cdot \ceil{\frac{k}{2}} \cdot W_{1} + \max\bracce{0, 2\cdot\paren{\ceil{\frac{k}{2}}-2}}\cdot W_{2}}$-emulator with $\tilde O \paren{n^{1+\frac{1}{k}}}$ edges. 
\end{theorem}
\begin{proof}
    The runtime of \algref{_k-1,kW1+_k-4_W2__runtime} is due to \lemref{_k-1,kW1+_k-4_W2__runtime_F}. As for the correctness, let $u,v\in V$ and fix a $u\squiggly v$. We distinguish between \case{a}, \case{b} and \case{c}. If \case{a} holds, $d\bracke{u,v}=\delta_{G}\paren{u,v}$ by \clmref{_k-1,kW1+_k-4_W2__runtime_case_b}. 

    If \case{b} holds, we recall \lemref{_k-1,kW1+_k-4_W2__case_b}, where we had an auxiliary edge from $p_{i}\paren{x}$ to some $p_{j}\paren{y}$. Regardless of the indices $i,j$, by \clmref{_k-1,kW1+_k-4_W2__runtime_case_b}, it follows that $d\bracke{p_{i}\paren{x},p_{j}\paren{y}} = \delta_{G}\paren{p_{i}\paren{x},p_{j}\paren{y}}$. 

    Therefore, it only remains to show the analysis for \case{c}  in \lemref{_k-1,kW1+_k-4_W2__case_c} still holds. To do so, we consider \case{c}{1}, \case{c}{2} and \case{c}{3}.

    \begin{enumerate}
        \item \dashuline{\case{c}{1}: $p_{1}\paren{x} \in \ball{w}{1}{2}$ or $p_{1}\paren{w} \in \ball{x}{1}{2}$.} \\
        We assume without loss of generality that $p_{1}\paren{w} \in \ball{x}{1}{2}$. By  \clmref{_k-1,kW1+_k-4_W2__runtime_case_c1} it follows that $d\bracke{p_{1}\paren{w},x}=\delta_{G}\paren{p_{1}\paren{w},x}$. The rest follows as in \case{c}{1} from \lemref{_k-1,kW1+_k-4_W2__case_c}.
        \item \dashuline{\case{c}{2}:
        $p_{1}\paren{x} \notin \ball{w}{1}{2}$, $p_{1}\paren{w} \notin \ball{x}{1}{2}$  and there exists an $2 \leq i \leq k-4$ }\\
        \dashuline{such that $p_{i}\paren{x} \in \ball{p_{1}\paren{w}}{i}{i+2}$ or $p_{i}\paren{w} \in \ball{p_{1}\paren{x}}{i}{i+2}$.} \\
        Let $i$ be the minimal such index.
        By  \clmref{_k-1,kW1+_k-4_W2__runtime_case_c2} it follows that $d\bracke{p_{i}\paren{x},p_{1}\paren{w}}=\delta_{G}\paren{p_{i}\paren{x},p_{1}\paren{w}}$. The rest follows as in \case{c}{2} from \lemref{_k-1,kW1+_k-4_W2__case_c}.

        \item \dashuline{\case{c}{3}:
        $p_{1}\paren{x} \notin \ball{w}{1}{2}$, $p_{1}\paren{w} \notin \ball{x}{1}{2}$  and for any $2 \leq i \leq k-4$ it }\\
        \dashuline{holds that $p_{i}\paren{x} \notin \ball{p_{1}\paren{w}}{i}{i+2}$ and $p_{i}\paren{w} \notin \ball{p_{1}\paren{x}}{i}{i+2}$.}\\
        In this case we introduce a slight modification to our case analysis.
        By \clmref{_k-1,kW1+_k-4_W2__runtime_case_c3} either $u\squiggly v \subseteq E_{k-1}$, in which case either $d\bracke{p_{k-2}\paren{x},p_{1}\paren{w}}= \delta_{G}\paren{p_{k-2}\paren{x},p_{1}\paren{w}}$ and we continue as in \case{c}{3} from \lemref{_k-1,kW1+_k-4_W2__case_c}, or there exists an edge $\paren{q,r}\in u\squiggly v$ such that $\paren{q,r}\notin E_{k-1}$, in which case the path $u\straight p_{k-1}\paren{q} \straight v \subseteq V\times S_{k-1} \subseteq F$ and then $\delta_{H}\paren{u,v}\leq \delta_{G}\paren{u,v}+2W_{1}\paren{u\squiggly v}$. 
    \end{enumerate}

This concludes the stretch proof for \algref{_k-1,kW1+_k-4_W2__runtime}. 
    
\end{proof}

We observe that for $k\leq 4$, the set $B_{2}\paren{S_{1}}$ is not being used. Hence, the runtime, for these values of $k$, can be reduced to $\tilde O\paren{n^{2+\frac{1}{k}}}$. We therefore conclude:

\begin{remark}\label{rem:_k-1,kW1+_k-4_W2__runtime_k<=4}
For $k\leq 4$ the runtime of \algref{_k-1,kW1+_k-4_W2__runtime} is $\tilde O\paren{n^{2+\frac{1}{k}}}$. 
\end{remark}

\subsection{Pseudocode for $k=3$}~\label{runtime_k=3}
For completeness, we provide the full pseudocode of a runtime efficient construction of the $+4W_{1}$-emulator presented in \secref{+4W1}. The size of the emulator is $\tilde O\paren{n^{\frac{4}{3}}}$ and the construction runtime is $\tilde O\paren{n^{\frac{7}{3}}}$.

\begin{algorithm}[H]
\caption{$+4W_{1}$-emulator$\paren{G=\paren{V,E,w}}$}\label{alg:+4W1_runtime}

\textbf{Input:} An undirected non-negative weighted graph $G=\paren{V,E,w}$.

\textbf{Output:} A $+4W_{1}$-emulator. 

\medskip

Let $\beta,\gamma \in \paren{ 0 , 1}$ to be fixed later

Initialize $d$ as the edges' weights and $\infty$ otherwise

Set $S_{0}\leftarrow V$

Sample each vertex of $V$ with probability $n^{-\beta}$ to $S_{1}$

Sample each vertex of $S_{1}$ with probability $n^{-\gamma}$ to $S_{2}$

Set $S_{3}\leftarrow \varnothing$

For every $u\in V$ find the pivot $p_{1}\paren{u}$ (respectively, $p_{2}\paren{u}$) and update $d$ accordingly

Construct the set of edges $D$

Construct the set of edges $E_{1}$ (respectively, $E_{2}$)

Set $E_{3}\leftarrow E$

\For{$i \in \bracce{1,2}$}
{
    \For{$s\in  S_{i}$}
    {
        Invoke SSSP from $s$ on $G_{i}  = \paren{V,E_{i+1} \cup D  , d }$ and update $d$ accordingly
    }
}

Set $F\leftarrow D \cup E_{1} \cup \paren{S_{1} \times S_{1} } \cup \paren{V\times S_{2}}$

Initialize $w_{H}\leftarrow \varnothing$

\For{$\paren{u,v}\in F\setminus E$}
    {
        $w_{H} \paren{u,v} \leftarrow d\bracke{u,v}$
    }

\For{$\paren{u,v}\in F\cap E$}
    {
        $w_{H} \paren{u,v} \leftarrow w\paren{u,v}$
    } 

\Return $H=\paren{V,F,w_{H}}$

\end{algorithm}

\pagebreak

\subsection{Pseudocode for $k=4$}~\label{runtime_k=4}
We next provide the full pseudocode of a runtime efficient construction of the $\paren{3,4W_{1}}$-emulator presented in \secref{_3,4W1_}. This constitutes the next case of our general construction and requires a slightly larger set of auxiliary edges than the $k=3$ case. The resulting emulator has $\tilde O\paren{n^{\frac{5}{4}}}$ edges and can be constructed in $\tilde O\paren{n^{\frac{9}{4}}}$ time.

\begin{algorithm}[H]
\caption{$\paren{3,4W_{1}}$-emulator$\paren{G=\paren{V,E,w}}$}\label{alg:_3,4W1__runtime}

\textbf{Input:} An undirected non-negative weighted graph $G=\paren{V,E,w}$.

\textbf{Output:} A $\paren{3,4W_{1}}$-emulator. 

\medskip

Let $\beta_1,\beta_2,\beta_3 \in \paren{ 0 , 1}$ to be fixed later

Initialize $d$ as the edges' weights and $\infty$ otherwise

Set $S_{0}\leftarrow V$

\For{$i\in\bracke{3}$}
{
    Sample each vertex of $S_{i-1}$ with probability $n^{-\beta_{i}}$ to $S_i$

    For every $u\in V$ find the pivot $p_i\paren{u}$ and update $d$ accordingly

    Construct the set of edges $D_{i}$

    Construct the set of edges $E_{i}$
}

Set $S_{4}\leftarrow\varnothing$

Set $E_{4}\leftarrow E$

$D\leftarrow\mcup{D_i}{i=1}{3}$

\For{$i \in \bracke{3}$}
{
    \For{$s\in S_{i}$}
    {
        Invoke SSSP from $s$ on
        $G_{i}=\paren{V,E_{i+1}\cup D,d}$
        and update $d$ accordingly
    }
}

Initialize $F\leftarrow \varnothing$

Construct the set of edges $B_{1}\paren{V}$

Set $F\leftarrow D \cup E_{1} \cup B_{1}\paren{V}
\cup \paren{S_{1}\times S_{2}}
\cup \paren{V\times S_{3}}$

Initialize $w_H\leftarrow\varnothing$

\For{$\paren{u,v}\in F\setminus E$}
{
    $w_{H}\paren{u,v}\leftarrow d\bracke{u,v}$
}

\For{$\paren{u,v}\in F\cap E$}
{
    $w_{H}\paren{u,v}\leftarrow w\paren{u,v}$
}

\Return $H=\paren{V,F,w_{H}}$

\end{algorithm}

\pagebreak

\section{Refined Stretch}~\label{refined}
In \secref{_k-1,kW1+_k-4_W2_}, we presented
\algref{_k-1,kW1+_k-4_W2_} and established its stretch in
\thmref{_k-1,kW1+_k-4_W2_}.
The proof distinguishes between \case{a}, \case{b} and \case{c}.
While \case{a} preserves the exact distance, the upper bounds obtained
for \case{b} and \case{c} have different forms.
In the proof of \thmref{_k-1,kW1+_k-4_W2_}, these two bounds are
upper bounded by a single expression in order to obtain a
concise statement that applies to every pair.

In this appendix, we retain the distinction between the two cases.
This yields a tighter pair-dependent stretch guarantee for
the same emulator.
We emphasize that no modification is made to
\algref{_k-1,kW1+_k-4_W2_} or to its edge set.
Consequently, the expected size of the emulator remains
$\tilde O\paren{n^{1+\frac{1}{k}}}$.

Throughout this appendix, let $k\geq 4$, let $u,v\in V$, and fix a
shortest path $u\squiggly v$. If $u\squiggly v = u\straight v$ we set $W_{2}\paren{u\squiggly v}=0$.
We now state the refined stretch guarantee.

\begin{theorem}~\label{thm:refined} 
\algref{_k-1,kW1+_k-4_W2_} constructs an emulator of expected size $\tilde O\paren{n^{1+\frac{1}{k}}}$ whose stretch is:

$$
\begin{cases}
\paren{
    k-1,\,
    \paren{k-2}\cdot W_{1}
    +W_{2}
    +\abs{
        2W_{1}
        -\paren{k-1}\cdot W_{2}
    }
}
&
k=0\pmod 2,
\\[3mm]
\paren{
    k-2,\,
    k\cdot W_{1}
    +W_{2}
    +\abs{
        W_{1}
        -\paren{k-2}\cdot W_{2}
    }
}
&
k=1\pmod 2.
\end{cases}
$$

This is a
$
\paren{
    2\cdot\floor{\frac{k}{2}}-1,\,
    \paren{
        3\cdot\ceil{\frac{k}{2}}
        -\floor{\frac{k}{2}}
        -2
    }\cdot W_{1}
    +W_{2}
    +\abs{
        \paren{
            \floor{\frac{k}{2}}
            -\ceil{\frac{k}{2}}
            +2
        }\cdot W_{1}
        -
        \paren{
            2\cdot\floor{\frac{k}{2}}-1
        }\cdot W_{2}
    }
}
$-emulator.

\end{theorem}

\begin{proof}
As the construction of
\algref{_k-1,kW1+_k-4_W2_}
remains unchanged, its expected size
$\tilde O\paren{n^{1+\frac{1}{k}}}$
follows directly from
\thmref{_k-1,kW1+_k-4_W2_}.
It therefore remains only to establish the refined stretch. 
Recall that the analysis of
\algref{_k-1,kW1+_k-4_W2_}
distinguishes between \case{a}, \case{b} and \case{c}.
If \case{a} holds, then
\lemref{_k-1,kW1+_k-4_W2__case_a}
gives
$\delta_H\paren{u,v}=\delta_G\paren{u,v}$,
and hence the required upper bound immediately follows.
It therefore remains to combine the upper bounds obtained for
\case{b} and \case{c}.

We first consider the case where $k=0\pmod 2$.
If \case{b} holds, then by
\lemref{_k-1,kW1+_k-4_W2__case_b}:

\begin{alignat*}{3}
\delta_H\paren{u,v}
&\leq
\textubrace{
\delta_G\paren{u,v}
+
2\cdot\paren{k-1}\cdot
W_{1}\paren{u\squiggly v}
}{
\lemref{_k-1,kW1+_k-4_W2__case_b}
}
\\
&=
\textubrace{
\delta_G\paren{u,v}
+
k\cdot W_{1}\paren{u\squiggly v}
+\paren{k-2}
\cdot W_{1}\paren{u\squiggly v}
}{
Simply summing
}\\
&\leq
\textubrace{
\paren{k-1}\cdot\delta_G\paren{u,v}
+
k\cdot W_{1}\paren{u\squiggly v}
-
\paren{k-2}\cdot W_{2}\paren{u\squiggly v}
}{
Because $W_{1}\paren{u\squiggly v}
+
W_{2}\paren{u\squiggly v}
\leq
\delta_G\paren{u,v}$
}
\end{alignat*}

If \case{c} holds 
\lemref{_k-1,kW1+_k-4_W2__case_c}
gives: $
\paren{k-1}\cdot\delta_G\paren{u,v}
+
\paren{k-4}\cdot W_{1}\paren{u\squiggly v}
+
k\cdot W_{2}\paren{u\squiggly v}$. Consequently the upper bounds obtained for
\case{b} and \case{c} imply:

\begin{alignat*}{3}
\delta_H\paren{u,v}
&\leq
\textubrace{
\paren{k-1}\cdot\delta_G\paren{u,v}
+
\mmax{
k\cdot W_{1}\paren{u\squiggly v}
-\paren{k-2}\cdot W_{2}\paren{u\squiggly v},
\,
\paren{k-4}\cdot W_{1}\paren{u\squiggly v}
+k\cdot W_{2}\paren{u\squiggly v}
}
}{
Our preceding discussion
}
\\
&=
\textubrace{
\paren{k-1}\cdot\delta_G\paren{u,v}
+
\paren{k-2}\cdot W_{1}\paren{u\squiggly v}
+
W_{2}\paren{u\squiggly v}
+
\abs{
2W_{1}\paren{u\squiggly v}
-\paren{k-1}\cdot W_{2}\paren{u\squiggly v}
}
}{
$\max\bracce{a,b}
=
\frac{a+b}{2}
+
\frac{\abs{a-b}}{2}$
}
\end{alignat*}

We now consider the case where $k=1\pmod 2$.
If \case{b} holds, then by
\lemref{_k-1,kW1+_k-4_W2__case_b}:

\begin{alignat*}{3}
\delta_H\paren{u,v}
&\leq
\textubrace{
\delta_G\paren{u,v}
+
2\cdot\paren{k-1}\cdot
W_{1}\paren{u\squiggly v}
}{
\lemref{_k-1,kW1+_k-4_W2__case_b}
}
\\
&=
\textubrace{
\delta_G\paren{u,v}
+
\paren{k+1}\cdot W_{1}\paren{u\squiggly v}
+
\paren{k-3}\cdot W_{1}\paren{u\squiggly v}
}{
Simply summing
}
\\
&\leq
\textubrace{
\paren{k-2}\cdot\delta_G\paren{u,v}
+
\paren{k+1}\cdot W_{1}\paren{u\squiggly v}
-
\paren{k-3}\cdot W_{2}\paren{u\squiggly v}
}{
Because $W_{1}\paren{u\squiggly v}
+
W_{2}\paren{u\squiggly v}
\leq
\delta_G\paren{u,v}$
}
\end{alignat*}

If \case{c} holds by
\lemref{_k-1,kW1+_k-4_W2__case_c}
we get:
$
\paren{k-2}\cdot\delta_G\paren{u,v}
+
\paren{k-1}\cdot
\paren{
W_{1}\paren{u\squiggly v}
+
W_{2}\paren{u\squiggly v}
}
$.
Consequently the upper bounds obtained for
\case{b} and \case{c} imply:

\begin{alignat*}{3}
\delta_H\paren{u,v}
&\leq
\textubrace{
\paren{k-2}\cdot\delta_G\paren{u,v}
+
\mmax{
\paren{k+1}\cdot W_{1}\paren{u\squiggly v}
-
\paren{k-3}\cdot W_{2}\paren{u\squiggly v},
\,
\paren{k-1}\cdot
\paren{
W_{1}\paren{u\squiggly v}
+
W_{2}\paren{u\squiggly v}
}
}
}{
Our preceding discussion
}
\\
&=
\textubrace{
\paren{k-2}\cdot\delta_G\paren{u,v}
+
k\cdot W_{1}\paren{u\squiggly v}
+
W_{2}\paren{u\squiggly v}
+
\abs{
W_{1}\paren{u\squiggly v}
-
\paren{k-2}\cdot W_{2}\paren{u\squiggly v}
}
}{
$\max\bracce{a,b}
=
\frac{a+b}{2}
+
\frac{\abs{a-b}}{2}$
}
\end{alignat*}

Lastly, we remark that if $k=0\pmod 2$, then
$\ceil{\frac{k}{2}}=\floor{\frac{k}{2}}=\frac{k}{2}$, and if
$k=1\pmod 2$, then $\ceil{\frac{k}{2}}=\frac{k+1}{2}$ and
$\floor{\frac{k}{2}}=\frac{k-1}{2}$. Substituting these values results in the unified expression. 

\end{proof}

We next compare the refined stretch to the bound stated in
\thmref{_k-1,kW1+_k-4_W2_}.

\begin{corollary}~\label{cor:refined}
The stretches satisfy:
$$
\begin{cases}
\paren{
k-1,\,
\paren{k-2}\cdot W_{1}
+W_{2}
+\abs{
2W_{1}
-\paren{k-1}\cdot W_{2}
}
}
\leq
\paren{
k-1,\,
k\cdot W_{1}
+\paren{k-4}\cdot W_{2}
}
&
k=0\pmod 2,
\\[3mm]
\paren{
k-2,\,
k\cdot W_{1}
+W_{2}
+\abs{
W_{1}
-\paren{k-2}\cdot W_{2}
}
}
\leq
\paren{
k-2,\,
\paren{k+1}\cdot W_{1}
+\paren{k-3}\cdot W_{2}
}
&
k=1\pmod 2
\end{cases}
$$

Furthermore, for any $u,v\in V$ and fixed shortest path
$u\squiggly v$ such that
$0<W_{2}\paren{u\squiggly v}<W_{1}\paren{u\squiggly v}$,
the additive stretch of the refined stretch is strictly smaller.
\end{corollary}

\begin{proof}
We first consider the case where $k=0\pmod 2$.
By the proof of \thmref{refined}, the additive stretch of the
refined stretch is:
$
\mmax{
k\cdot W_{1}\paren{u\squiggly v}
-\paren{k-2}\cdot W_{2}\paren{u\squiggly v},
\,
\paren{k-4}\cdot W_{1}\paren{u\squiggly v}
+k\cdot W_{2}\paren{u\squiggly v}
}
$, 
whereas the additive component of
\thmref{_k-1,kW1+_k-4_W2_}
is
$
k\cdot W_{1}\paren{u\squiggly v}
+\paren{k-4}\cdot W_{2}\paren{u\squiggly v}
$.
We compare this expression with each of the two expressions inside
the maximum:

\begin{alignat*}{3}
&
k\cdot W_{1}\paren{u\squiggly v}
+\paren{k-4}\cdot W_{2}\paren{u\squiggly v}
-
\paren{
k\cdot W_{1}\paren{u\squiggly v}
-\paren{k-2}\cdot W_{2}\paren{u\squiggly v}
}
\\
&=
\textubrace{
2\cdot\paren{k-3}\cdot W_{2}\paren{u\squiggly v}
}{
Simply summing
}
\\
&\geq
\textubrace{
0
}{
$W_{2}\paren{u\squiggly v}\geq 0$
}
\end{alignat*}

Additionally:

\begin{alignat*}{3}
&
k\cdot W_{1}\paren{u\squiggly v}
+\paren{k-4}\cdot W_{2}\paren{u\squiggly v}
-
\paren{
\paren{k-4}\cdot W_{1}\paren{u\squiggly v}
+k\cdot W_{2}\paren{u\squiggly v}
}
\\
&=
\textubrace{
4\cdot
\paren{
W_{1}\paren{u\squiggly v}
-
W_{2}\paren{u\squiggly v}
}
}{
Simply summing
}
\\
&\geq
\textubrace{
0
}{
$W_{1}\paren{u\squiggly v}
\geq
W_{2}\paren{u\squiggly v}$
}
\end{alignat*}

Hence, both expressions inside the maximum are no larger than the
additive component of
\thmref{_k-1,kW1+_k-4_W2_}.
As the multiplicative component is $k-1$ in both results, the desired
inequality follows for even $k$.

We now consider the case where $k=1\pmod 2$.
By the proof of \thmref{refined}, the additive stretch of the
refined stretch is
$
\mmax{
\paren{k+1}\cdot W_{1}\paren{u\squiggly v}
-\paren{k-3}\cdot W_{2}\paren{u\squiggly v},
\,
\paren{k-1}\cdot
\paren{
W_{1}\paren{u\squiggly v}
+
W_{2}\paren{u\squiggly v}
}
}
$, whereas the additive component of
\thmref{_k-1,kW1+_k-4_W2_}
is
$
\paren{k+1}\cdot W_{1}\paren{u\squiggly v}
+\paren{k-3}\cdot W_{2}\paren{u\squiggly v}
$. We again compare this expression with each of the two expressions
inside the maximum:

\begin{alignat*}{3}
&
\paren{k+1}\cdot W_{1}\paren{u\squiggly v}
+\paren{k-3}\cdot W_{2}\paren{u\squiggly v}
-
\paren{
\paren{k+1}\cdot W_{1}\paren{u\squiggly v}
-\paren{k-3}\cdot W_{2}\paren{u\squiggly v}
}
\\
&=
\textubrace{
2\cdot\paren{k-3}\cdot W_{2}\paren{u\squiggly v}
}{
Simply summing
}
\\
&\geq
\textubrace{
0
}{
$W_{2}\paren{u\squiggly v}\geq 0$
}
\end{alignat*}

Additionally:

\begin{alignat*}{3}
&
\paren{k+1}\cdot W_{1}\paren{u\squiggly v}
+\paren{k-3}\cdot W_{2}\paren{u\squiggly v}
-
\paren{
\paren{k-1}\cdot
\paren{
W_{1}\paren{u\squiggly v}
+
W_{2}\paren{u\squiggly v}
}
}
\\
&=
\textubrace{
2\cdot
\paren{
W_{1}\paren{u\squiggly v}
-
W_{2}\paren{u\squiggly v}
}
}{
Simply summing
}
\\
&\geq
\textubrace{
0
}{
$W_{1}\paren{u\squiggly v}
\geq
W_{2}\paren{u\squiggly v}$
}
\end{alignat*}

Hence, both expressions inside the maximum are no larger than the
additive component of
\thmref{_k-1,kW1+_k-4_W2_}.
As the multiplicative component is $k-2$ in both results, the desired
inequality follows for odd $k$.

Finally, if
$0<W_{2}\paren{u\squiggly v}<W_{1}\paren{u\squiggly v}$,
then both differences considered above are strictly positive,
regardless of the parity of $k$.
Hence, the refined additive component is strictly smaller.

\end{proof}

To emphasize the refined stretch, we list in
\tabref{refined} the first few instances of our general
construction.
The multiplicative stretch remains unchanged, while the refined
analysis yields a smaller additive stretch whenever
$0<W_{2}<W_{1}$.

\begin{table}[H]
  \begin{center}

    \small\addtolength{\tabcolsep}{-4pt}
    \begin{tabular}{|c|c|c|c|c|c|}
      \hline
      \textbf{\small $k$}
      &
      \textbf{\small Multiplicative Stretch}
      &
      \textbf{\small Original Additive Stretch}
      &
      \textbf{\small Refined Additive Stretch}
      &
      \textbf{\small Non-Negative $\abs{\cdot}$}
      &
      \textbf{\small Negative $\abs{\cdot}$}
      \\
      \hhline{|=|=|=|=|=|=|}

      &&&&&\\[-1em]
      $4$
      &
      $3$
      &
      $4W_{1}$
      &
      $2W_{1}+W_{2}+\abs{2W_{1}-3W_{2}}$
      &
      $4W_{1}-2W_{2}$
      &
      $4W_{2}$
      \\
      &&&&&\\[-1em]
      \hline

      &&&&&\\[-1em]
      $5$
      &
      $3$
      &
      $6W_{1}+2W_{2}$
      &
      $5W_{1}+W_{2}+\abs{W_{1}-3W_{2}}$
      &
      $6W_{1}-2W_{2}$
      &
      $4W_{1}+4W_{2}$
      \\
      &&&&&\\[-1em]
      \hline

      &&&&&\\[-1em]
      $6$
      &
      $5$
      &
      $6W_{1}+2W_{2}$
      &
      $4W_{1}+W_{2}+\abs{2W_{1}-5W_{2}}$
      &
      $6W_{1}-4W_{2}$
      &
      $2W_{1}+6W_{2}$
      \\
      &&&&&\\[-1em]
      \hline

      &&&&&\\[-1em]
      $7$
      &
      $5$
      &
      $8W_{1}+4W_{2}$
      &
      $7W_{1}+W_{2}+\abs{W_{1}-5W_{2}}$
      &
      $8W_{1}-4W_{2}$
      &
      $6W_{1}+6W_{2}$
      \\
      \hline

    \end{tabular}

    \caption{The additive stretch of
    \thmref{_k-1,kW1+_k-4_W2_}
    compared to the refined additive stretch of
    \thmref{refined}.}~\label{tab:refined}

  \end{center}
\end{table}



\end{document}